\documentclass[a4paper,12pt]{article}
\usepackage[latin1]{inputenc}
\usepackage{color,graphicx}
\usepackage{amsmath,amssymb,amsthm,dsfont}
\DeclareMathOperator{\tr}{Tr}
\DeclareMathOperator{\hilbert}{\mathcal{H}}
\renewcommand{\Re}{\operatorname{Re}}
\DeclareMathOperator{\sgn}{sgn}

\usepackage{setspace,cite}
\usepackage{color}
\title{Open Quantum Systems. \\An Introduction}
\author{\'Angel Rivas$^{1,2}$ and Susana F. Huelga$^1$}
\date{$^1$ Institut f\"{u}r Theoretische Physik, Universit\"{a}t Ulm, Ulm D-89069, Germany.\\
$^2$ Departamento de F\'{\i}sica Te\'orica I, Universidad Complutense, 28040 Madrid, Spain.}
\begin{document}
\maketitle
\begin{abstract}
We revise fundamental concepts in the dynamics of open quantum systems in the light of modern developments in the field. Our aim
is to present a unified approach to the quantum evolution of open systems that incorporates the concepts and methods traditionally employed by different communities.
We present in some
detail the mathematical structure and the general properties of the dynamical
maps underlying open system dynamics.
We also discuss the microscopic derivation of dynamical equations, including both Markovian and non-Markovian evolutions.
\end{abstract}

\tableofcontents

\section{Introduction}

To write an introduction to the dynamics of open quantum systems may seem at first a complicated, albeit perhaps unnecessary, task. On the one hand, the field is quite broad and encompasses many different topics which are covered by several books and reviews \cite{Louisell73,Agarwal,Haake,Daviesbook76,RevKoss,Spohn80,Kraus83,Leggett87,AlickiLendi87,Cohen92,Carmichael93,Bohn96,Blum96,Gardiner97,Carmichael99,Zwanzig01,Puri01,AlickiFanes01,Lindblad01,Fain02,BrPe02,Giovanna02,Brandes,Benatti03,Decoherence03,GardinerZoller04,quimicos04,Benatti05,Franceses06,quimicos06,Tanimura,DecoherenceBrain07,Weiss08,ModernCohen08,Hornberger09,Rivas09,Hu12}. On the other hand, the approach taken to study the dynamics of an open quantum system can be quite different depending of the specific system being analyzed; for instance, references in quantum optics (e.g. \cite{Louisell73,Cohen92,Carmichael93,Carmichael99,Puri01,GardinerZoller04}) use tools and approximations which are usually quite different from the studies focused on condensed matter or open systems of relevance in chemical physics (e.g. \cite{Leggett87,Blum96,quimicos04,quimicos06,Weiss08}). Moreover, very rigorous studies by the mathematical physics community were carried out in the 1970's, 
usually from a statistical physics point of view. More recently, and mainly motivated by the development of quantum information science, there has been a strong revival in the study of open many body systems aimed at further understanding the impact of decoherence phenomena on quantum information protocols \cite{NC00}.

This plethora of results and approaches can be quite confusing to novices in the field. Even within the more experienced scientists, discussions on fundamental properties (e.g. the concept of Markovianity) have often proven not to be straightforward when researches from different communities are involved.

Our motivation in writing this work has been to try to put black on white the results of several thoughts and discussions on this topic with scientists of different backgrounds during the last three years, as well as putting several concepts in the context of modern developments in the field. Given the extension of the topic and the availability of many excellent reviews covering different aspects of the physics of open quantum systems, we will focus here on one specific issue and that is embedding the theory as usually explained in the quantum optical literature within the mathematical core developed in the mathematical physics community. Finally, we make a connection with the current view point of the dynamics of open quantum systems in the light of quantum information science, given a different, complementary, perspective to the meaning of fundamental concepts such that complete positivity, dynamical semigroups, Markovianity, master equations, etc.

In order to do this, we will restrict the proof of important mathematical results to the case of finite dimensional systems; otherwise this study would be much more extensive and potentially interested people who are not expert in functional analysis or operators algebras, could not easily benefit from it. However, since the most important application of the theory of open quantum system nowadays is arguably focused on the field of controlled quantum technologies, including quantum computation, quantum communication, quantum metrology, etc., where the most useful systems are finite (or special cases of infinite systems), the formal loss of generality is justified in this case. The necessary background for this work is therefore reduced to familiarity with the basics of quantum mechanics and quantum information theory.

Furthermore, we have tried to incorporate some novel concepts and techniques developed in the last few years to address old problems such that those exemplified by the use of dynamical maps, Markovian and non-Markovian evolutions or microscopic derivations of reduced dynamics. We have organized the presentation as follows. In section \ref{sectionMates} we introduce the mathematical tools required in order to ensure a minimum degree of self-consistency in the study; a mathematical experienced reader may skip this part. In section \ref{sectionClosedSystems} we review very succinctly the concept of ``closed quantum system'' and discuss its time evolution. We explain in section \ref{sectionOpenSystems} the general features of dynamics of open quantum systems and its mathematical properties. Section \ref{sectionMarkov} deals with the mathematical structure of a very special important dynamics in open quantum systems, which are the Markovian quantum processes. The microscopic derivation of Markovian quantum dynamics is explained in section \ref{sectionMicrosMarkov}, whereas some methods for obtaining non-Markovian dynamics are briefly introduced in section \ref{sectionMicrosnonMarkov}. Finally, we conclude the review by making a summary of the main ideas we have discussed, and the topics that we have left out, in the section \ref{conclusions}.

\section{Mathematical tools} \label{sectionMates}
\newtheorem{theorem}{Theorem}[section]
\newtheorem{proposition}{Proposition}[section]
\newtheorem{corollary}{Corollary}[section]
\theoremstyle{definition}
\newtheorem{definition}{Definition}[section]

\subsection{Banach spaces, norms and linear operators} \label{sectionBanach}

In different parts within this work, we will make use of some properties of Banach spaces. For our interest it is sufficient to restrict our analysis to finite dimensional spaces. In what follows we revise a few basic concepts \cite{Kreyszig78,reedsimon1,Boccara}.

\begin{definition} \label{defBanach} A Banach space $\mathfrak{B}$ is a complete linear space with a norm $\|\cdot\|$.
\end{definition}

Recall that a space is complete if every Cauchy sequence converges in it. An example of Banach space is provided by the set of self-adjoint trace-class linear operators on a Hilbert space $\mathcal{H}$. This is, the self-adjoint operators whose trace norm is finite $\|\cdot\|_1$,
\[
\|A\|_1=\tr\sqrt{A^\dagger A}=\tr\sqrt{A^2}, \quad A\in\mathfrak{B}.
\]

Now we focus our attention onto possible linear transformations on a Banach space, that we will denote by $T(\mathfrak{B}):\mathfrak{B}\mapsto\mathfrak{B}$. The set of all of them forms the dual space $\mathfrak{B}^\ast$.

\begin{proposition} The dual space of $\mathfrak{B}$, $\mathfrak{B}^\ast$ is a (finite) Banach space with the induced norm
\[
\|T\|=\sup_{x\in\mathfrak{B}, x\neq0}\frac{\|T(x)\|}{\|x\|} =\sup_{x\in\mathfrak{B}, \|x\|=1} \|T(x)\|, \quad T\in\mathfrak{B}^\ast.
\]
\end{proposition}

\noindent\textit{Proof}: The proof is basic.

Apart from the triangle inequality, the induced norm fulfills some other useful inequalities. For example, immediately from its definition one obtains that
\begin{equation}\label{ineqinducednorm}
\|T(x)\|\leq\|T\|\|x\|, \quad \forall x\in \mathfrak{B}.
\end{equation}
Moreover,
\begin{proposition} Let $T_1$ and $T_2$ be two linear operators in $\mathfrak{B}^\ast$, then
\begin{equation}\label{submultiplicative}
\|T_1T_2\|\leq\|T_1\|\|T_2\|.
\end{equation}
\end{proposition}
\begin{proof} This is just a consequence of the inequality (\ref{ineqinducednorm}); indeed,
\[
\|T_1T_2\|=\sup_{x\in\mathfrak{B}, \|x\|=1} \|T_1 T_2(x)\|\leq\sup_{x\in\mathfrak{B}, \|x\|=1} \|T_1\| \|T_2(x)\|=\|T_1\| \|T_2\|.
\]
\end{proof}

The next concept will appear quite often in the forthcoming sections.
\begin{definition} A linear operator $T$ on a Banach space $\mathfrak{B}$, is said to be a contraction if
\[
\|T(x)\|\leq\|x\|, \quad \forall x\in \mathfrak{B},
\]
this is $\|T\|\leq1$.
\end{definition}

Finally, we briefly revise some additional basic concepts that will be featured in subsequent sections.

\begin{definition} For any linear operator $T$ on a (finite) Banach space $\mathfrak{B}$, its resolvent operator is defined as
\[
R_T(\lambda)=(\lambda\mathds{1}-T)^{-1}, \lambda\in\mathds{C}.
\]
It is said that $\lambda\in\mathds{C}$ belongs to the (point) spectrum of $T$, $\lambda\in\sigma(T)$, if $R_T(\lambda)$ does not exist; otherwise $\lambda\in\mathds{C}$ belongs to the resolvent set of $T$, $\lambda\in\varrho(T)$. As a result $\sigma(T)$ and $\varrho(T)$ are said to be complementary sets.

The numbers $\lambda\in\sigma(T)$ are called \emph{eigenvalues} of $T$, and each element in the kernel $x\in\ker(\lambda\mathds{1}-T)$, which is obviously different from $\{0\}$ as $(\lambda\mathds{1}-T)^{-1}$ does not exist, is called the \emph{eigenvector} associated with the eigenvalue $\lambda$.
\end{definition}

\subsection{Exponential of an operator}
The operation of exponentiation of an operator is of special importance in relation to linear differential equations.
\begin{definition} The exponential of a linear operator $L$ on a finite Banach space is defined as
\begin{equation}\label{exponential}
e^L=\sum_{n=0}^\infty\frac{L^n}{n!}.
\end{equation}
\end{definition}
The following proposition shows that this definition is indeed meaningful.
\begin{proposition} The series defining $e^L$ is absolutely convergent.
\end{proposition}
\begin{proof} The following chain of inequalities holds
\[
\Vert e^L\Vert=\left\Vert \sum_{n=0}^\infty\frac{L^n}{n!}\right\Vert\leq\sum_{n=0}^\infty\frac{\Vert L^n\Vert}{n!}\leq\sum_{n=0}^\infty\frac{\Vert L\Vert^n}{n!}=e^{\Vert L\Vert}.
\]
Since $e^{\Vert L\Vert}$ is a numerical series which is convergent, the norm of $e^L$ is upper bounded and therefore the series converges absolutely.

\end{proof}

\begin{proposition} If $L_1$ and $L_2$ commute with each other, then $e^{(L_1+L_2)}=e^{L_1}e^{L_2}=e^{L_2}e^{L_1}$.
\end{proposition}
\begin{proof} By using the binomial formula as in the case of a numerical series, we obtain
\begin{multline*}
e^{(L_1+L_2)}=\sum_{n=0}^\infty\frac{(L_1+L_2)^n}{n!}=\sum_{n=0}^\infty \frac{1}{n!} \sum_{k=0}^n\binom{n}{k}L_1^kL_2^{n-k}\\
=\sum_{n=0}^\infty \sum_{k=0}^n\frac{1}{k!(n-k)!}L_1^kL_2^{n-k}=\sum_{p=0}^\infty \sum_{q=0}^\infty\frac{1}{p!q!}L_1^pL_2^q=e^{L_1}e^{L_2}=e^{L_2}e^{L_1}.
\end{multline*}
\end{proof}

The above result does not hold when $L_1$ and $L_2$ are not commuting operators, in that case there is a formula which is sometimes useful.
\begin{theorem}[Lie-Trotter product formula] Let $L_1$ and $L_2$ be linear operators on a finite Banach space, then:
\[
e^{(L_1+L_2)}=\lim_{n\rightarrow\infty}\left(e^{\frac{L_1}{n}}e^{\frac{L_2}{n}}\right)^n.
\]
\end{theorem}
\begin{proof} We present a two-step proof inspired by \cite{Galindoproblemas}. First let us prove that for any two linear operators $A$ and $B$ the following identity holds:
\begin{equation}\label{trotter1}
A^n-B^n=\sum_{k=0}^{n-1}A^k(A-B)B^{n-k-1}.
\end{equation}
We observe that, in the end, we are only adding and subtracting the products of powers of $A$ and $B$ in the right hand side of (\ref{trotter1}), so that
\begin{multline*}
A^{n}-B^{n}=A^{n}+(A^{n-1}B-A^{n-1}B)-B^{n}\\
=A^{n}+(A^{n-1}B-A^{n-1}B+A^{n-2}B^2-A^{n-2}B^2)-B^{n}=\\
=A^{n}+(A^{n-1}B-A^{n-1}B+A^{n-2}B^2-A^{n-2}B^2+\ldots \\
\dots+AB^{n-1}-AB^{n-1})-B^{n}.
\end{multline*}
Now we can add terms with positive sign to find:
\[
A^{n}+A^{n-1}B+\ldots+AB^{n-1}=\sum_{k=0}^{n-1}A^{k+1}B^{n-k-1},
\]
and all terms with negative sign to obtain:
\[
-A^{n-1}B-\ldots -AB^{n-1}-B^{n}=-\sum_{k=0}^{n-1}A^{k}B^{n-k}.
\]
Finally, by adding both contributions one gets that
\[
A^{n}-B^{n}=\sum_{k=0}^{n-1}A^{k+1}B^{n-k-1}-A^{k}B^{n-k}=\sum_{k=0}^{n-1}A^k(A-B)B^{n-k-1},
\]
as we wanted to prove.

Now, we denote
\[
X_n=e^{(L_1+L_2)/n}, \quad Y_n=e^{L_1/n}e^{L_2/n},
\]
and take $A=X_n$ and $B=Y_n$ in formula (\ref{trotter1}),
\[
X_n^n-Y_n^n=\sum_{k=0}^{n-1}X_n^k(X_n-Y_n)Y_n^{n-k-1}.
\]
On the one hand,
\begin{multline*}
\Vert X_n^k\Vert=\Vert \left(e^{(L_1+L_2)/n}\right)^k\Vert \leq \Vert e^{(L_1+L_2)/n}\Vert^k \leq e^{\Vert L_1+L_2\Vert k/n} \\
\leq e^{(\Vert L_1\Vert+\Vert L_2\Vert)k/n}\leq e^{\Vert L_1\Vert+\Vert L_2\Vert}, \quad \forall k\leq n.
\end{multline*}
Similarly, one proves that
\[
\Vert Y_n^k\Vert\leq e^{\Vert L_1\Vert+\Vert L_2\Vert}, \quad \forall k\leq n.
\]
So we get
\begin{eqnarray*}
\Vert X_n^n-Y_n^n\Vert&\leq&\sum_{k=0}^{n-1}\Vert X_n^k(X_n-Y_n)Y_n^{n-k-1}\Vert\leq\sum_{k=0}^{n-1}\Vert X_n^k\Vert \Vert X_n-Y_n\Vert \Vert Y_n^{n-k-1}\Vert\\
&\leq&\sum_{k=0}^{n-1}\Vert X_n-Y_n\Vert e^{2(\Vert L_1\Vert+\Vert L_2\Vert)}=n\Vert X_n-Y_n\Vert e^{2(\Vert L_1\Vert+\Vert L_2\Vert)}.
\end{eqnarray*}
On the other hand, by expanding the exponential we find that
\begin{multline*}
\lim_{n\rightarrow\infty}n\Vert X_n-Y_n\Vert=\lim_{n\rightarrow\infty}n\left\Vert \sum_{k=0}^\infty\frac{1}{k!}\left(\frac{L_1+L_2}{n}\right)^k-\sum_{k,j=0}^\infty \frac{1}{k!j!}\left(\frac{L_1}{n}\right)^k\left(\frac{L_2}{n}\right)^j\right\Vert\\
\leq\lim_{n\rightarrow\infty}n\left\Vert\frac{1}{2n^2}[L_2,L_1]+\mathcal{O}\left(\frac{1}{n^3}\right)\right\Vert=0.
\end{multline*}
Therefore we finally conclude that:
\[
\lim_{n\rightarrow\infty}\left\Vert e^{L_1+L_2}-\left(e^{L_1/n}e^{L_2/n}\right)^n\right\Vert=\lim_{n\rightarrow\infty}\Vert X_n^n-Y_n^n\Vert=0.
\]
\end{proof}

\noindent Using a similar procedure one can also show that for a sum of $N$ operators, the following identity holds:
\begin{equation}\label{Lie-TrotterGeneral}
e^{\sum_{k=1}^NL_k}=\lim_{n\rightarrow\infty}\left(\prod_{k=1}^N e^{\frac{L_k}{n}}\right)^n.
\end{equation}

\subsection{Semigroups of operators}

\begin{definition}[Semigroup] A family of linear operators $T_t$ $(t\geq0)$ on a finite Banach Space forms a \emph{one-parameter semigroup} if
\begin{enumerate}
\item $T_tT_s=T_{t+s}$,$\quad \forall t,s$
\item $T_0=\mathds{1}$.
\end{enumerate}
\end{definition}

\begin{definition}[Uniformly Continuous Semigroup] A one-parameter semigroup $T_t$ is said to be \emph{uniformly continuous} if the map
\[
t\mapsto T_t
\]
is continuous, this is, $\lim_{t\rightarrow s}\|T_t-T_s\|=0$, $\forall s$. This kind of continuity is normally referred to as ``continuity in the uniform operator topology'' \cite{Kreyszig78}, and it is sufficient to analyze the finite dimensional case, which is the focus of this work.
\end{definition}

\begin{theorem} If $T_t$ forms a uniformly continuous one-parameter semigroup, then the map $t\mapsto T_t$
is differentiable, and the derivative of $T_t$ is given by
\[
\frac{dT_t}{dt}=LT_t,
\]
with $L=\frac{dT_t}{dt}|_{t=0}$.
\end{theorem}

\begin{proof}
Since $T_t$ is uniformly continuous on $t$, the function $V(t)$ defined by
\[
V(t)=\int_0^tT_sds, \quad t\geq0,
\]
is differentiable with $\frac{dV(t)}{dt}=T_t$ (to justify this, it is enough to consider the well-known arguments used for the case of $\mathds{R}$, see for example \cite{Spivak}, but substituting the absolute value by the concept of norm). In particular,
\[
\lim_{t\rightarrow0}\frac{V(t)}{t}=\lim_{t\rightarrow0}\frac{V(t)-V(0)}{t}=\left.\frac{dV(t)}{dt}\right|_{t=0}=T_0=\mathds{1},
\]
this implies that there exist some $t_0>0$ small enough such that $V(t_0)$ is invertible. Then,
\begin{eqnarray*}
T_t &=& V^{-1}(t_0) V (t_0) T_t=V^{-1}(t_0)\int_0^{t_0} T_sT_tds=V^{-1}(t_0)\int_0^{t_0} T_{t+s}ds\\
&=&V^{-1}(t_0)\int_t^{t+t_0} T_{s}ds=V^{-1}(t_0)\left[V(t+t_0)-V(t)\right],
\end{eqnarray*}
since the difference of differentiable functions is differentiable, $T_t$ is differentiable. Moreover, its derivative is
\[
\frac{dT_t}{dt}=\lim_{h\rightarrow0}\frac{T_{t+h}-T_t}{h}=\lim_{h\rightarrow0}\frac{T_{h}-\mathds{1}}{h}T_t=LT_t.
\]
\end{proof}

\begin{theorem} \label{theoExp} The exponential operator $T(t)=e^{Lt}$ is the only solution of the differential problem
\begin{equation*}\left\{\begin{array}{l}
\frac{dT(t)}{dt}=LT(t),\quad t\in\mathds{R}^+,\\
T(0)=\mathds{1}.
\end{array}\right.
\end{equation*}
\end{theorem}
\begin{proof} From the definition of exponential given by Eq. (\ref{exponential}), it is clear that $T(0)=\mathds{1}$, and because of the commutativity of the exponents we find
\[
\frac{de^{Lt}}{dt}=\lim_{h\rightarrow0}\frac{e^{L(t+h)}-e^{Lt}}{h}=\left(\lim_{h\rightarrow0}\frac{e^{Lh}-\mathds{1}}{h}\right)e^{Lt}.
\]
By expanding the exponential one obtains
\[
\lim_{h\rightarrow0}\frac{e^{Lh}-\mathds{1}}{h}=\lim_{h\rightarrow0}\frac{\mathds{1}+Lh+\mathcal{O}(h^2)-\mathds{1}}{h}=L.
\]
At this point, it just remains to prove uniqueness. Let us consider another function $S(t)$ and assume that it also satisfies the differential problem. Then we define
\[
Q(s)=T(s)S(t-s), \quad\text{ for }t\geq s\geq0,
\]
for some fixed $t>0$, so $Q(s)$ is differentiable with derivative
\[
\frac{dQ(s)}{ds}=LT(s)S(t-s)-T(t)LS(t-s)=LT(s)S(t-s)-LT(t)S(t-s)=0.
\]
$Q(s)$ is therefore a constant function. In particular $Q(s)=Q(0)$ for any $t\geq s\geq0$, which implies that
\[
T(t)=T(t)S(t-t)=Q(t)=Q(0)=S(t).
\]
\end{proof}

\begin{corollary} Any uniformly continuous semigroup can be written in the form $T_t=T(t)=e^{Lt}$, where $L$ is called the generator of the semigroup and it is the only solution to the differential problem
\begin{equation}\label{diffprobhomo}\left\{\begin{array}{l}
\frac{dT_t}{dt}=LT_t,\quad t\in\mathds{R}^+\\
T_0=\mathds{1}.
\end{array}\right.
\end{equation}
\end{corollary}

Given the omnipresence of differential equations describing time evolution in physics, it is now evident why semigroups are important. If we can extend the domain of $t$ to negative values, then $T_t$ forms a \emph{one-parameter group} whose inverses are given by $T_{-t}$, so that $T_{t}T_{-t}=\mathds{1}$ for all $t\in\mathds{R}$.

There is a special class of semigroups which is important for our proposes, and that will be the subject of the next section.

\subsubsection{Contraction semigroups}

\begin{definition} A one-parameter semigroup satisfying $\|T_t\|\leq1$ for every $t\geq0$, is called a contraction semigroup.
\end{definition}

Now one question arises, namely, which properties should $L$ fulfill to generate a contraction semigroup? This is an intricate question in general that is treated with more detail in functional analysis textbooks such as \cite{Yosida,reedsimon2,EngelNagel}. For finite dimensional Banach spaces, the proofs of the main theorems can be drastically simplified, and they are presented in the following. The main condition is based in properties of the resolvent set and the resolvent operator.

\begin{theorem}[Hille-Yoshida]\label{theoHill-Yoshida}
A necessary and sufficient condition for a linear operator $L$ to generate a contraction semigroup is that
\begin{enumerate}
\item $\{\lambda, \Re\lambda>0\}\subset\varrho(L)$. \label{THiYo1}
\item $\Vert R_L(\lambda)\Vert\leq (\Re \lambda)^{-1}$. \label{THiYo2}
\end{enumerate}
\end{theorem}
\begin{proof}
Let $L$ be the generator of a contraction semigroup and $\lambda$ a complex number. For $a\geq b$ we define the operator
\[
\mathcal{L}(a,b)=\int_a^{b} e^{(-\lambda\mathds{1}+L)t}dt.
\]
Since the exponential series converges uniformly (to see this one uses again tools of the elementary analysis, specifically the M-test of Weierstrass \cite{Spivak}, on the Banach space $\mathfrak{B}$) it is possible to integrate term by term in the previous definition
\begin{eqnarray*}
\mathcal{L}(a,b)=\int_a^{b} \sum_{n=0}^\infty \frac{(-\lambda\mathds{1}+L)^n t^n}{n!}dt&=&\sum_{n=0}^\infty \frac{(-\lambda\mathds{1}+L)^n}{n!}\left(\int_a^b t^n dt\right)\\
&=&\sum_{n=0}^\infty \frac{(-\lambda\mathds{1}+L)^n}{n!}\left[\frac{b^{n+1}}{n+1}-\frac{a^{n+1}}{n+1}\right].
\end{eqnarray*}
Thus by multiplying $\mathcal{L}(a,b)$ by $(\lambda\mathds{1}-L)$ we get
\[
(\lambda\mathds{1}-L)\mathcal{L}(a,b)=\mathcal{L}(a,b)(\lambda\mathds{1}-L)=\left[e^{(-\lambda\mathds{1}+L)a}-e^{(-\lambda\mathds{1}+L)b}\right].
\]
Now assume $\Re\lambda>0$ and take the limit $a\rightarrow0$, $b\rightarrow\infty$, then as
\begin{eqnarray*}
\lim_{b\rightarrow\infty}\Vert e^{(-\lambda\mathds{1}+L)b}\Vert=\lim_{b\rightarrow\infty}\Vert e^{-\lambda\mathds{1}b}e^{Lb}\Vert&\leq&\lim_{b\rightarrow\infty}\Vert e^{-\lambda\mathds{1}b}\Vert\Vert e^{Lb}\Vert\\
&\leq&\lim_{b\rightarrow\infty}\Vert e^{-\lambda\mathds{1}b}\Vert=0,
\end{eqnarray*}
and
\[
\lim_{a\rightarrow0}e^{-\lambda\mathds{1}a}e^{La}=\mathds{1},
\]
we obtain
\[
(\lambda\mathds{1}-L)\left[\int_0^\infty e^{-\lambda\mathds{1} t} e^{L t}dt\right]=\left[\int_0^\infty e^{-\lambda\mathds{1} t} e^{L t}dt\right](\lambda\mathds{1}-L)=\mathds{1}.
\]
So if $\mathrm{Re}\lambda>0$ then $\lambda\in\varrho(L)$ and from the definition of the resolvent operator we conclude
\[
R_L(\lambda)=\int_0^\infty e^{-\lambda\mathds{1} t} e^{L t}dt.
\]
Moreover,
\begin{eqnarray*}
\Vert R_L(\lambda)\Vert&=&\left\Vert\int_0^\infty e^{-\lambda\mathds{1} t} e^{L t}dt\right\Vert\leq\int_0^\infty \Vert e^{-\lambda\mathds{1} t} e^{L t}\Vert dt\leq\int_0^\infty \Vert e^{-\lambda\mathds{1} t}\Vert \Vert e^{L t}\Vert dt\\
&\leq& \int_0^\infty  \Vert e^{-\lambda\mathds{1} t}\Vert dt=\int_0^\infty  |e^{-\lambda t}| dt=\int_0^\infty e^{-\Re(\lambda) t} dt=\frac{1}{\Re(\lambda)}.
\end{eqnarray*}

Conversely, suppose that $L$ satisfies the above conditions \textit{\ref{THiYo1}} and \textit{\ref{THiYo2}}. Let $\lambda$ be some real number in the resolvent set of $L$, which means it is also a real positive number. We define the operator
\[
L_\lambda=\lambda LR_L(\lambda),
\]
and it turns out that $L_\lambda\rightarrow L$ as $\lambda\rightarrow\infty$. Indeed, it is enough to show that $\lambda R_L(\lambda)\rightarrow\mathds{1}$. Since $(0,\infty)\subset\varrho(L)$ there is no problem with the resolvent operator taking the limit, so by writing $\lambda R_L(\lambda)=LR_L(\lambda)+\mathds{1}$ we obtain
\[
\lim_{\lambda\rightarrow\infty}\Vert \lambda R_L(\lambda)-\mathds{1}\Vert=\lim_{\lambda\rightarrow\infty}\Vert L R_L(\lambda)\Vert\leq\lim_{\lambda\rightarrow\infty}\Vert L\Vert \Vert R_L(\lambda)\Vert \leq \lim_{\lambda\rightarrow\infty}\frac{\Vert L\Vert}{\lambda}=0,
\]
where we have used \textit{\ref{THiYo2}}. Moreover,
\begin{multline*}
\Vert e^{L_\lambda t}\Vert=\Vert e^{\lambda LR_L(\lambda)t}\Vert=\Vert e^{[\lambda^2 R_L(\lambda)-\lambda\mathds{1}]t}\Vert\leq e^{-\lambda t}\Vert e^{\lambda^2 R_L(\lambda)t}\Vert\\
\leq e^{-\lambda t} e^{\lambda^2 \Vert R_L(\lambda)\Vert t}\leq e^{-\lambda t} e^{\lambda t}=1,
\end{multline*}
and therefore
\[
e^{L t}=\lim_{\lambda\rightarrow\infty} e^{L_\lambda t}
\]
is a contraction semigroup.

\end{proof}

Apart from imposing conditions such that $L$ generates a contraction semigroup, this theorem is of vital importance in general operator theory on arbitrary dimensional Banach spaces. If the conditions of the theorem \ref{theoHill-Yoshida} are fulfilled then $L$ generates a well-defined semigroup. Note that for general operators $L$ the exponential cannot be rigourously defined as a power series, and more complicated tools are needed. On the other hand, to apply this theorem directly can be complicated in many cases, so one needs more manageable equivalent conditions.

\begin{definition}
A \emph{semi-inner product} is an operation in which for each pair ${x_1,x_2}$ of elements of a Banach space $\mathfrak{B}$ there is an associated complex number $[x_1,x_2]$, such that
\begin{align}
[x_1,\lambda x_2+\mu x_3]&=\lambda[x_1,x_2]+\mu[x_1,x_3],\label{sip1}\\
[x_1,x_1]&=\Vert x_1\Vert^2,\label{sip2}\\
\vert[x_1,x_2]\vert&\leq\Vert x_1\Vert\Vert x_2\Vert \label{sip3}.
\end{align}
for any $x_1,x_2,x_3\in\mathfrak{B}$ and $\lambda,\mu\in\mathds{C}$.
\end{definition}

\begin{definition}
A linear operator $L$ on $\mathfrak{B}$ is called dissipative if
\begin{equation}
\Re[x,Lx]\leq0 \quad \forall x\in\mathfrak{B}.
\end{equation}
\end{definition}

\begin{theorem}[Lumer-Phillips] \label{Lumer-PhillipsTheo}
A necessary and sufficient condition for a linear operator $L$ on $\mathfrak{B}$ (of finite dimension) to be the generator of a contraction semigroup is that $L$ be dissipative.
\end{theorem}
\begin{proof}
Suppose that $\{T_t,t\geq0\}$ is a contraction semigroup, then just by using the properties of a semi-inner product we have that:
\begin{eqnarray*}
\Re[x,(T_t-\mathds{1})x]&=&\Re[x,T_tx]-\Vert x\Vert^2\leq |[x,T_tx]|-\Vert x\Vert^2 \\
&\leq&\Vert x\Vert \Vert T_tx\Vert-\Vert x\Vert^2\leq\Vert x\Vert^2 \Vert T_t\Vert-\Vert x\Vert^2\leq0,
\end{eqnarray*}
as $\Vert T_t\Vert\leq1$. So $L$ is dissipative
\[
\Re[x,Lx]=\lim_{t\rightarrow0}\frac{1}{t}\Re[x,(T_t-\mathds{1})x]\leq0.
\]
Conversely, let $L$ be dissipative. Considering a point within its spectrum $\lambda\in\sigma(L)$ and a corresponding eigenvector $x\in\ker(\lambda\mathds{1}-L)$. Since $[x,(\lambda\mathds{1}-L)x]=0$ we have
\[
\Re[x,(\lambda\mathds{1}-L)x]=\Re(\lambda)\Vert x\Vert^2-\Re[x,Lx]=0,
\]
given that $\Re[x,Lx]\leq0$, $\Re(\lambda)\leq0$. Thus any complex number $\Re(\lambda)>0$ is in the resolvent of $L$. Moreover, let $\lambda\in\varrho(L)$, again from the properties of a semi-inner product we have that
\begin{eqnarray*}
\Re(\lambda)\Vert x\Vert^2=\Re(\lambda)[x,x]&\leq&\Re(\lambda[x,x]-[x,Lx])=\Re[x,(\lambda\mathds{1}-L)x]\\
&\leq& |[x,(\lambda\mathds{1}-L)x]|\leq\Vert x\Vert\Vert (\lambda\mathds{1}-L)x\Vert,
\end{eqnarray*}
so $\Vert (\lambda\mathds{1}-L)x\Vert\geq\Re(\lambda)\Vert x\Vert$, $\forall x\neq0$. Then, by setting $x=(\lambda\mathds{1}-L)\tilde{x}$ in the expression for the norm of $R_L(\lambda)$:
\begin{eqnarray*}
\Vert R_L(\lambda)\Vert=\sup_{x\in\mathfrak{B}, x\neq0}\frac{\Vert R_L(\lambda)x\Vert}{\Vert x\Vert}&=&\sup_{\tilde{x}\in\mathfrak{B}, \tilde{x}\notin\ker(\lambda\mathds{1}-L)}\frac{\Vert \tilde{x}\Vert}{\Vert (\lambda\mathds{1}-L)\tilde{x}\Vert}\\
&\leq&\frac{\Vert \tilde{x}\Vert}{\Re(\lambda)\Vert \tilde{x}\Vert}=\frac{1}{\Re(\lambda)},
\end{eqnarray*}
and therefore $L$ satisfies the conditions of the theorem \ref{theoHill-Yoshida}.

\end{proof}

\subsection{Evolution families}\label{sectionEvolutionFamilies}

As we have seen semigroups arise naturally in the context of problems involving linear differential equations, as exemplified by Eq. (\ref{diffprobhomo}), with time-independent generators $L$. However some situations in physics require solving time-inhomogeneous differential problems,
\begin{equation*}\left\{\begin{array}{l}
\frac{dx}{dt}=L(t)x, \quad x\in\mathfrak{B}\\
x(t_0)=x_0,
\end{array}\right.
\end{equation*}
where $L(t)$ is a time-dependent linear operator.
Given that the differential equation is linear, its solution must depend linearly on the initial conditions, so we can write
\begin{equation}
x(t)=T_{(t,t_0)}x(t_0),
\end{equation}
where $T_{(t,t_0)}$ is a linear operator often referred to as the \textit{evolution operator}. Obviously, $T_{(t_0,t_0)}=\mathds{1}$ and the composition of two consecutive evolution operators is well defined; since for $t_2 \geq t_1 \geq t_0$
\[
x(t_1)=T_{(t_1,t_0)}x(t_0)\quad\text{and}\quad x(t_2)=T_{(t_2,t_1)}x(t_1).
\]
The evolution operator between $t_2$ and $t_0$ is given by
\begin{equation}\label{composicion}
T_{(t_2,t_0)}=T_{(t_2,t_1)}T_{(t_1,t_0)},
\end{equation}
and the two-parameter family of operators $T_{(t,s)}$ fulfills the properties
\begin{equation}\label{evfamily}
\begin{array}{l}
T_{(t,s)}=T_{(t,r)}T_{(r,s)}, \quad t\geq r \geq s\\
T_{(s,s)}=\mathds{1}.
\end{array}
\end{equation}
These operators are sometimes called \textit{evolution families}, \textit{propagators}, or \textit{fundamental solutions}. One important fact to realize is that if the generator is time-independent, $L$, the evolution family is reduced to a one-parameter semigroup in the time difference $T_{(t,s)}=T_{\tau=t-s}$. However, in contrast to general semigroups, it is not enough (even in finite dimension) to require that the map $(t,s)\mapsto T_{(t,s)}$ be continuous to ensure that it is differentiable \cite{EvFamilynoDiff}, although in practice we usually assume that the temporal evolution is smooth enough and families fulfilling Eq. (\ref{evfamily}) will also considered to be differentiable. Under these conditions we can formulate the following result.

\begin{theorem} A differentiable family $T_{(t,s)}$ obeying the condition (\ref{evfamily}) is the only solution to the differential problems
\begin{equation*}\left\{\begin{array}{l}
\frac{dT_{(t,s)}}{dt}=L(t)T_{(t,s)}, \quad t\geq s\\
T_{(s,s)}=\mathds{1},
\end{array}\right.
\end{equation*}
and
\begin{equation*}\left\{\begin{array}{l}
\frac{dT_{(t,s)}}{ds}=-T_{(t,s)}L(t), \quad t\geq s\\
T_{(s,s)}=\mathds{1}.
\end{array}\right.
\end{equation*}
\end{theorem}

\begin{proof} If the evolution family $T_{(t,s)}$ is differentiable
\[
\frac{dT_{(t,s)}}{dt}=\lim_{h\rightarrow0}\frac{T_{(t+h,s)}-T_{(t,s)}}{h}=\lim_{h\rightarrow0}\frac{T_{(t+h,t)}-\mathds{1}}{h}T_{(t,s)}=L(t)T_{(t,s)}.
\]
Analogously one proves that
\[
\frac{dT_{(t,s)}}{ds}=-T_{(t,s)}L(t).
\]
To show that this solution is unique, as in theorem $\ref{theoExp}$, let us assume that there exist some other $S_{(t,s)}$ solving both differential problems. Then, for every fixed $t$ and $s$, we define now $Q(r)=T_{(t,r)}S_{(r,s)}$ for $t\geq r\geq s$. It is immediate to show that $\frac{dQ(r)}{dr}=0$ and therefore
\[
T_{(t,s)}=Q(s)=Q(t)=S_{(t,s)}.
\]
\end{proof}

Again opposed to the case of semigroups, to give a general expression for an evolution family $T_{(t,s)}$ in terms of the generator $L(t)$ is complicated. To illustrate why let us define first the concept of time-ordering operator.

\begin{definition} The \emph{time-ordered operator} $\mathcal{T}$ of a product of two operators $L(t_1)L(t_2)$ is defined by
\[
\mathcal{T}L(t_1)L(t_2)=\theta(t_1-t_2)L(t_1)L(t_2)+\theta(t_2-t_1)L(t_2)L(t_1),
\]
where $\theta(x)$ is the Heaviside step function. Similarly, for a product of three operators $L(t_1)L(t_2)L(t_3)$, we have
\[
\begin{split}
\mathcal{T}L(t_1)L(t_2)L(t_3)=\ &\theta(t_1-t_2)\theta(t_2-t_3)L(t_1)L(t_2)L(t_3)\\
+\ &\theta(t_1-t_3)\theta(t_3-t_2)L(t_1)L(t_3)L(t_2)\\
+\ &\theta(t_2-t_1)\theta(t_1-t_3)L(t_2)L(t_1)L(t_3)\\
+\ &\theta(t_2-t_3)\theta(t_3-t_1)L(t_2)L(t_3)L(t_1)\\
+\ &\theta(t_3-t_1)\theta(t_1-t_2)L(t_3)L(t_1)L(t_2)\\
+\ &\theta(t_3-t_2)\theta(t_2-t_1)L(t_3)L(t_2)L(t_1),
\end{split}
\]
so that for a product of $k$ operators $L(t_1)\cdots L(t_k)$ we can write:
\begin{equation}\label{timeordering}
\mathcal{T}L(t_1)\cdots L(t_k)=\sum_{\pi}\theta[t_{\pi(1)}-t_{\pi(2)}]\cdots\theta[t_{\pi(k-1)}-t_{\pi(k)}]L[t_{\pi(1)}]\cdots L[t_{\pi(k)}],
\end{equation}
where $\pi$ is a permutation of $k$ indexes and the sum extends over all $k!$ different permutations.
\end{definition}

\begin{theorem}[Dyson Series] Let the generator $L(t')$ to be bounded in the interval $[t,s]$, i.e. $\|L(t')\|<\infty,\ t'\in[t,s]$. Then the evolution family $T_{(t,s)}$ admits the following series representation
\begin{equation}\label{seriepropagador}
T(t,s)=\mathds{1}+\sum_{m=1}^\infty\int_s^t\int_s^{t_1}\cdots\int_s^{t_{m-1}}L(t_1)\cdots L(t_m)dt_m\cdots dt_1,
\end{equation}
where $t\geq t_1\geq \ldots \geq t_m\geq s$, and it can be written symbolically as
\begin{equation}\label{serieDyson}
T(t,s)=\mathcal{T}e^{\int_s^tL(t')dt'}\equiv\mathds{1}+\sum_{m=1}^\infty\frac{1}{m!}\int_s^t\cdots\int_s^{t}\mathcal{T}L(t_1)\cdots L(t_m)dt_1\cdots dt_m.
\end{equation}
\end{theorem}
\begin{proof}
To prove that expression (\ref{seriepropagador}) can be formally written as (\ref{serieDyson}) note that by introducing the time-ordering definition Eq. (\ref{timeordering}), the $k^\mathrm{th}$ term of this series can be written as
\begin{multline*}
\frac{1}{k!}\int_s^t\cdots\int_s^{t}\mathcal{T}L(t_1)\cdots L(t_k)dt_1\cdots dt_k\\
=\frac{1}{k!}\sum_\pi\int_s^t\int_s^{t_{\pi(1)}}\cdots\int_s^{t_{\pi(k-1)}}L[t_{\pi(1)}]\cdots L[t_{\pi(k)}]dt_1\cdots dt_k\\
=\int_s^t\int_s^{t_1}\cdots\int_s^{t_{k-1}}L(t_{1})\cdots L(t_{k})dt_k\cdots dt_1,
\end{multline*}
where in the last step we have used that $t\geq t_{\pi(1)}\geq \ldots \geq t_{\pi(k)}\geq s$, and that the name of the variables does not matter, so the sum over $\pi$ is the addition of the same integral $k!$ times.

To prove Eq. (\ref{seriepropagador}), we need first of all to show that the series is convergent, but this is simple because by taking norms in (\ref{serieDyson}) we get the upper bound
\[
\|T(t,s)\|\leq \sum_{m=0}^\infty\frac{|t-s|^m}{m!}\left(\sup_{t'\in[s,t]}\|L(t')\|\right)^m=\exp\left[\sup_{t'\in[s,t]}\|L(t')\||t-s|\right],
\]
which is convergent as a result of $L(t')$ being bounded in $[t,s]$. Finally, the fact that (\ref{seriepropagador}) is the solution of the differential problem is easy to verify by differentiating the series term by term since it converges uniformly and so does its derivative (to see this one uses again the M-test of Weierstrass \cite{Spivak} on the Banach space $\mathfrak{B}$).
\end{proof}

This result is the well-known Dyson expansion, which is widely used for instance in scattering theory. We have proved its convergence for finite dimensional systems and bounded generator. This is normally difficult to prove in the infinite dimensional case, but the expression is still used in this context in a formal sense.

Apart from the Dyson expansion, there is an approximation formula for evolution families which will be useful in this work.

\begin{theorem}[Time-splitting formula] The evolution family $T_{(t,s)}$ can be given by the formula
\begin{equation}\label{time-splitting}
T_{(t,s)}=\lim_{\max|t_{j+1}-t_j|\rightarrow0}\prod_{j={n-1}}^{0}e^{(t_{j+1}-t_j)L(t_j)} \quad \text{for }t=t_{n}\geq\ldots\geq t_0=s,
\end{equation}
where $L(t_j)$ is the generator evaluated in the instantaneous time $t_j$ (note the descending order in the product symbol).
\end{theorem}

\begin{proof} The idea for this proof is taken from \cite{Schwabl}, more rigorous and general proofs can be found in \cite{Yosida,reedsimon2}.
Given that the product in (\ref{time-splitting}) is already temporally ordered, nothing happens if we introduce the temporal-ordering operator
\begin{eqnarray*}
T_{(t,s)}&=&\lim_{\max|t_{j+1}-t_j|\rightarrow0}\mathcal{T}\prod_{j={n-1}}^{0}e^{(t_{j+1}-t_j)L(t_j)},\\
&=&\lim_{\max|t_{j+1}-t_j|\rightarrow0}\mathcal{T}e^{\sum_{j=0}^{n-1}(t_{j+1}-t_j)L(t_j)},
\end{eqnarray*}
where one realizes that the time-ordering operator is then making the work of ordering the terms in the exponential series in the bottom term in order to obtain the same series expansion as in the top term. Now it only remains to recognize the Riemann sum \cite{Spivak} in the exponent to arrive at the expression
\[
T_{(t,s)}=\mathcal{T}e^{\int_s^tL(t')dt'},
\]
which is Eq. (\ref{serieDyson}).
\end{proof}

In analogy to the case of contraction semigroups, one can define contraction evolution families. However, since we have now two parameters, there is not a unique possible definition, as illustrated below.

\begin{definition} An evolution family $T_{(t,s)}$ is called \emph{contractive} if $\|T_{(t,s)}\|\leq1$
for every $t$ and every $s$ such that $t\geq s$.
\end{definition}

\begin{definition} An evolution family $T_{(t,s)}$ is called \emph{eventually contractive} if there exist a $s_0$ such that $\|T_{(t,s_0)}\|\leq1$ for every $t\geq s_0$. \label{eventuallycontractive}
\end{definition}

It is clear that contractive families are also eventually contractive, but in this second case there exist a privileged initial time $s_0$, such that one can have $\|T_{(t,s)}\|>1$ for some $s\neq s_0$. This difference is relevant and will allow us later on to discriminate universal dynamical maps given by differential equations which are Markovian from those which are not.

\section{Time evolution in closed quantum systems}\label{sectionClosedSystems}

Probably time evolution of physical systems is the main factor in order to understand their nature and properties. In classical systems time evolution is usually formulated in terms of differential equations (i.e. Euler-Lagrange's equations, Hamilton's equations, Liouville equation, etc.), which can present very different characteristics depending on which physical system they correspond to. From the beginning of the quantum theory, physicists have been often trying to translate the methods which were useful in the classical case to the quantum one, so was that Erwin Schr\"{o}dinger obtained the first quantum evolution equation in 1926 \cite{Sc26}.

This equation, called Schr\"{o}dinger's equation since then, describes the behavior of an \emph{isolated} or \emph{closed} quantum system, that is, by definition, a system which does not interchanges information (i.e. energy and/or matter \cite{OpenQuantumSystemRemark}) with another system. So if our isolated system is in some pure state $|\psi(t)\rangle\in\mathcal{H}$ at time $t$, where $\mathcal{H}$ denotes the Hilbert space of the system, the time evolution of this state (between two consecutive measurements) is according to
\begin{equation}
\frac{d}{dt}|\psi(t)\rangle=-\frac{i}{\hbar}H(t) |\psi(t)\rangle,
\end{equation}
where $H(t)$ is the Hamiltonian operator of the system. In the case of mixed states $\rho(t)$ the last equation naturally induces
\begin{equation}\label{Liouville}
\frac{d\rho(t)}{dt}=-\frac{i}{\hbar}[H(t),\rho(t)],
\end{equation}
sometimes called von Neumann or Liouville-von Neumann equation, mainly in the context of statistical physics. From now on we consider units where $\hbar=1$.

There are several ways to justify the Schr\"{o}dinger equation, some of them related with heuristic approaches or variational principles, however in all of them it is necessary to make some extra assumptions apart from the postulates about states and observables; for that reason the Schr\"{o}dinger equation can be taken as another postulate. Of course the ultimate reason for its validity is that the experimental tests are in accordance with it \cite{Zeilinger,Weinberg,Wineland}. In that sense, the Schr\"{o}dinger equation is for quantum mechanics as fundamental as the Maxwell's equations for electromagnetism or the Newton's Laws for classical mechanics.

An important property of Schr\"{o}dinger equation is that it does not change the norm of the states, indeed
\begin{multline}\label{norma0}
\frac{d}{dt}\langle \psi(t)|\psi(t)\rangle=\left(\frac{d\langle \psi(t)|}{dt}\right)|\psi(t)\rangle+\langle \psi(t)|\left(\frac{d|\psi(t)\rangle}{dt}\right)=\\
=i\langle \psi(t)|H^\dagger(t) |\psi(t)\rangle-i\langle \psi(t)|H(t) |\psi(t)\rangle=0,
\end{multline}
as the Hamiltonian is self-adjoint, $H(t)=H^\dagger(t)$. In view of the fact that the Schr\"{o}dinger equation is linear, its solution is given by an evolution family $U(t,r)$,
\begin{equation}
|\psi(t)\rangle=U(t,t_0)|\psi(t_0)\rangle,
\end{equation}
where $U(t_0,t_0)=\mathds{1}$.

The equation (\ref{norma0}) implies that the evolution operator is an isometry (i.e. norm-preserving map), which means it is a unitary operator in case of finite dimensional systems. In infinite dimensional systems one has to prove that $U(t,t_0)$ maps the whole $\mathcal{H}$ onto $\mathcal{H}$ in order to exclude partial isometries, this is easy to verify from the composition law of evolution families \cite{GP90} and then for any general Hilbert space
\begin{equation}
U^\dagger(t,t_0) U(t,t_0)=U(t,t_0)U^\dagger(t,t_0)=\mathds{1}\Rightarrow U^{-1}=U^\dagger.
\end{equation}

On the other hand, according to these definitions for pure states, the evolution of some density matrix $\rho(t)$ is
\begin{equation}\label{rhosanguis}
\rho(t_1)=U(t_1,t_0)\rho(t_0)U^\dagger(t_1,t_0).
\end{equation}
We can also rewrite this relation as
\begin{equation}\label{rhoOperator}
\rho(t_1)=\mathcal{U}_{(t_1,t_0)}[\rho(t_0)],
\end{equation}
where $\mathcal{U}_{(t_1,t_0)}$ is a linear (unitary \cite{unitaryLiouville}) operator acting as
\begin{equation}
\mathcal{U}_{(t_1,t_0)}[\cdot]=U(t_1,t_0)[\cdot]U^\dagger(t_1,t_0).
\end{equation}

Of course the specific form of the evolution operator in terms of the Hamiltonian depends on the properties of the Hamiltonian itself. In the simplest case in which $H$ is time-independent (conservative system), the formal solution of the Sch\"{o}dinger equation is straightforwardly obtained as
\begin{equation}\label{expH}
U(t,t_0)=e^{-i(t-t_0)H/\hbar}.
\end{equation}
When $H$ is time-dependent (non-conservative system) we can formally write the evolution as a Dyson expansion,
\begin{equation}\label{DysonH}
U(t,t_0)=\mathcal{T}e^{\int_{t_0}^tH(t')dt'}.
\end{equation}


Note now that if $H(t)$ is self-adjoint, $-H(t)$ is also self-adjoint, so in this case we have physical inverses for every evolution. In fact when the Hamiltonian is time-independent, $U(t_1,t_0)=U(t_1-t_0)$. So the evolution operator is not only a one-parameter semigroup, but a one-parameter group, with an inverse for every element $U^{-1}(\tau)=U(-\tau)=U^\dagger(\tau)$, such that $U^{-1}(\tau)U(\tau)=U(\tau)U^{-1}(\tau)=\mathds{1}$. Of course in case of time-dependent Hamiltonians, every member of the evolution family $U(t,s)$ has the inverse $U^{-1}(t,s)=U(s,t)=U^\dagger(t,s)$, being this one also unitary.

If we look into the equation (\ref{rhoOperator}), we can identify $\rho$ as a member of the Banach space of the set of trace-class self-adjoint operators, and $\mathcal{U}_{(t_1,t_0)}$ as some linear operator acting on this Banach space. Then everything considered for pure states concerning the form of the evolution operator is also applicable for general mixed states, actually the general solution to the equation (\ref{Liouville}) can be written as
\begin{equation}
\rho(t)=\mathcal{U}_{(t,t_0)}\rho(t_0)=\mathcal{T}e^{\int_{t_0}^t\mathcal{L}_{t'}dt'}\rho(t_0),
\end{equation}
where the generator is given by the commutator $\mathcal{L}_t(\cdot)=-i[H(t),\cdot]$. This is sometimes called the \emph{Liouvillian}, and by taking the derivative one immediately arrives to (\ref{Liouville}).

\section{Time evolution in open quantum systems}\label{sectionOpenSystems}

Consider now the case where the total Hilbert space is decomposed into two parts, in the form $\mathcal{H}=\mathcal{H}_A\otimes\mathcal{H}_B$, each subspace corresponding to a certain quantum system; the question to answer is then, what properties does the time evolution of each subsystem have?.

First of all, one must establish the relation between the total density matrix $\rho\in\mathcal{H}$ and the density matrix of a subsystem, say $A$, which is given by the partial trace over the other subsystem $B$:
\begin{equation}
\rho_A=\tr_B(\rho).
\end{equation}
There are several ways to justify this; see for example \cite{NC00,ShashThesis}. Let us focus first in the description of a measurement of the system $\mathcal{H}_A$ viewed from the whole system $\mathcal{H}=\mathcal{H}_A\otimes\mathcal{H}_B$. For the system $\mathcal{H}_A$ a measurement is given by a set of self-adjoint positive operators $\{M_k\}$ each of them associated with a possible result $k$. Suppose that we are blind to the system $B$, and we see the state of system $A$ as some $\rho_A$. Then the probability to get the result $k$ when we make the measurement $\{M_k\}$ will be
\[
p_A(k)=\tr(M_k\rho_A).
\]
Now one argues that if we are viewing only the part $A$ through the measurement $M_k$, we are in fact observers of an extended system $\mathcal{H}$ through some measurement $\tilde{M}_k$ such that, for physical consistency, for any state of the composed system $\rho$ compatible with $\rho_A$ which is seen from $A$ (see figure \ref{fig1Rev}), the following holds
\[
p_A(k)=\tr(\tilde{M}_k\rho).
\]
But assuming this is true for every state $\rho_A$, we assume that $\tilde{M}_k$ is a genuine measurement operator on $\mathcal{H}$, which implicitly implies that is independent of the state $\rho$ of the whole system on which we are making the measurement. Now consider the special case where the state of the whole system has the form $\rho=\rho_A\otimes\rho_B$, then the above equation is clearly satisfied with the choice $\tilde{M}_k=M_k\otimes \mathds{1}$,
\[
p_A(k)=\tr(M_k\rho_A)=\tr[(M_k\otimes \mathds{1})\rho_A\otimes\rho_B]=\tr(M_k\rho_A)\tr(\rho_B)=\tr(M_k\rho_A).
\]
Actually this is the only possible choice such that $\tilde{M}_k$ is independent of the particular states $\rho_A$ and $\rho_B$ (and so of $\rho$). Indeed, imagine that there exist another solution, $\tilde{M}_k$, independent of $\rho$, which is also a genuine measure on $\hilbert$, then from the linearity of the trace we would get
\[
p_A(k)=\tr(M_k\rho_A)=\tr(M_k\otimes\mathds{1}\rho)=\tr(\tilde{M}_k\rho)\Rightarrow \tr[(\tilde{M}_k-M_k\otimes\mathds{1})\rho]=0,
\]
for all $\rho_A$ and so for all $\rho$. This means
\[
\tr[(\tilde{M}_k-M_k\otimes\mathds{1})\sigma]=0,
\]
for any self-adjoint operator $\sigma$, therefore as the set of self-adjoint operators forms a Banach space this implies $\tilde{M}_k-M_k\otimes\mathds{1}=0$ and then $\tilde{M}_k=M_k\otimes\mathds{1}$.

\begin{figure}[h]
\centering
\includegraphics[width=0.9\textwidth]{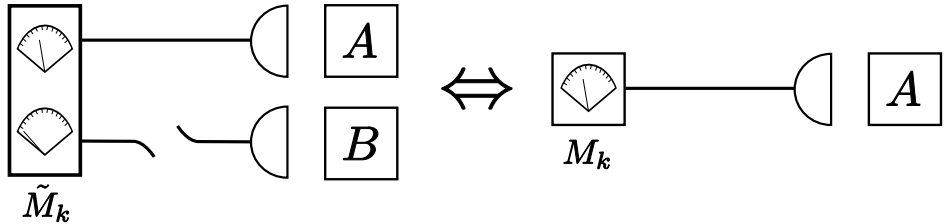}
\caption{On the left side we perform a measure over the combined system $A$ and $B$, which we denote by $\tilde{M}_k$, but we have ``disconnected'' the information provided by $B$. If one assume that this situation is equivalent to a measure $M_k$ only the system $A$, this is illustrated on the right side, then $\tilde{M}_k=M_k\otimes\mathds{1}$ (see discussion on the text).} \label{fig1Rev}
\end{figure}

Once established this relation, one can wonder about the connection between $\rho_A$ and $\rho$. If $\{|\psi_j^A\rangle|\psi_\ell^B\rangle\}$ is a basis of $\hilbert$ then
\begin{eqnarray*}
p_A(k)=\tr[(M_k\otimes \mathds{1})\rho]&=&\sum_{j\ell}\langle\psi_j^A|\langle\psi_\ell^B|(M_k\otimes \mathds{1})\rho|\psi_j^A\rangle|\psi_\ell^B\rangle\\
&=&\sum_{j,\ell}\langle\psi_j^A|M_k\langle\psi_\ell^B|\rho|\psi_\ell^B\rangle|\psi_j^A\rangle=\tr(M_k\rho_A),
\end{eqnarray*}
hence $\rho_A$ is univocally given by the partial trace
\[
\rho_A=\sum_\ell\langle\psi_\ell^B|\rho|\psi_\ell^B\rangle\equiv\tr_B(\rho).
\]

One important property of the partial trace operation is that, even if $\rho=|\psi\rangle\langle\psi|$ was pure, $\rho_A$ can be mixed, this occurs if $|\psi\rangle$ is entangled. 

Given the total density matrix $\rho(t_0)$, by taking the partial trace in (\ref{rhosanguis}), the state of $A$ at time $t_1$, $\rho_A(t_1)$, is given by
\begin{equation}\label{rhoAevolucion}
\rho_A(t_1)=\tr_B[U(t_1,t_0)\rho(t_0)U^\dagger(t_1,t_0)].
\end{equation}

Provided that the total evolution is not factorized $U(t_1,t_0)=U_A(t_1,t_0)\otimes U_B(t_1,t_0)$, both quantum subsystems $A$ and $B$ are interchanging information with each other (they are interacting), and so they are \emph{open quantum systems}.

\subsection{Dynamical maps}

Let us set $A$ as our system to study. A natural task arising from the equation (\ref{rhoAevolucion}) is trying to rewrite it as a \emph{dynamical map} acting on $\mathcal{H}_A$ which connects the states of the subsystem $A$ at times $t_0$ and $t_1$:
\[
\mathcal{E}_{(t_1,t_0)}:\rho_A(t_0)\rightarrow\rho_A(t_1).
\]
The problem with this map is that in general it does not only depend on the global evolution operator $U(t_1,t_0)$ and on the properties of the subsystem $B$, but also on the system $A$ itself.

In order to clarify this, let us write the total initial state as the sum of two contributions \cite{SB01}:
\begin{equation}\label{rhoDesc}
\rho(t_0)=\rho_A(t_0)\otimes\rho_B(t_0)+\rho_{corr}(t_0),
\end{equation}
where the term $\rho_{corr}(t_0)$ is not a quantum state; satisfies
\[
\tr_A[\rho_{corr}(t_0)]=\tr_B[\rho_{corr}(t_0)]=0,
\]
and contains all correlations (classical as well as quantum) between the two subsystems. The substitution of (\ref{rhoDesc}) in (\ref{rhoAevolucion}) provides
\begin{eqnarray}\label{KrausInno}
\rho_A(t_1)&=&\tr_B\{U(t_1,t_0)[\rho_A(t_0)\otimes\rho_B(t_0)+\rho_{corr}(t_0)]U^\dagger(t_1,t_0)\}\nonumber\\
&=&\sum_i\lambda_i\tr_B\{U(t_1,t_0)[\rho_A(t_0)\otimes|\psi_i\rangle\langle\psi_i|]U^\dagger(t_1,t_0)\}\nonumber\\
&+&\tr_B[U(t_1,t_0)\rho_{corr}(t_0)U^\dagger(t_1,t_0)]\nonumber\\
&=&\sum_{\alpha}K_{\alpha}(t_1,t_0)\rho_A(t_0)K_{\alpha}^\dagger(t_1,t_0)+\delta\rho(t_1,t_0)\nonumber\\
&\equiv&\mathcal{E}_{(t_1,t_0)}[\rho_A(t_0)],
\end{eqnarray}
where $\alpha=\{i,j\}$, $K_{\{i,j\}}(t_1,t_0)=\sqrt{\lambda_i}\langle\psi_j|U(t_1,t_0)|\psi_i\rangle$ and we have used the spectral decomposition of $\rho_B(t_0)=\sum_i\lambda_i|\psi_i\rangle\langle\psi_i|$. Apparently the operators $K_\alpha(t_1,t_0)$ depend only on the global unitary evolution and on the initial state of the subsystem $B$, but the inhomogeneous part $\delta\rho(t_1,t_0)=\tr_B[U(t_1,t_0)\rho_{corr}(t_0)U^\dagger(t_1,t_0)]$ may no longer be independent of $\rho_A$ because of the correlation term $\rho_{corr}$ \cite{FactorizedDynamics}. That is, because of the positivity requirement of the total density matrix $\rho(t_0)$, given some $\rho_A(t_0)$ and $\rho_B(t_0)$ not any $\rho_{corr}(t_0)$ is allowed such that $\rho(t_0)$ given by (\ref{rhoDesc}) is a positive operator; in the same way, given $\rho_A(t_0)$ and $\rho_{corr}(t_0)$ not any $\rho_B(t_0)$ is allowed.

A paradigmatic example of this is the case of reduced pure states. Then the following theorem is fulfilled.

\begin{theorem}\label{rhoApuro} Let one of the reduced states of a bipartite quantum system be pure, say $\rho_A=|\psi\rangle\mbox{}_A\langle\psi|$, then
$\rho_{corr}=0$.
\end{theorem}
\begin{proof}
Indeed, let us consider first that $\rho=|\psi\rangle\langle\psi|$ is pure as well. By using the Schmidt decomposition of $|\psi\rangle=\sum_{i}\lambda_i|u_i,v_i\rangle$, where for all $i$ in the sum $\lambda_i>0$, and $\{|u_i\rangle\}$ and $\{|v_i\rangle\}$ are orthonormal basis of $\hilbert_A$ and $\hilbert_B$ respectively, we obtain
\begin{equation}\label{pureNoCorr}
\rho_A=\tr_B\left[\sum_{i,j}\lambda_i\lambda_j|u_i,v_i\rangle\langle u_j,v_j|\right]=\sum_{i}\lambda_i^2|u_i\rangle\langle u_i|.
\end{equation}
It is elementary to show that a pure state cannot be written as non-trivial convex combination of another two states \cite{Peres,Ballentine}; then one coefficient, say $\lambda_{i'}$ must be equal to one, and $\lambda_{i}=0$ for all $i\neq i'$. This implies that $\rho$ has the product form
\[
\rho=|u_{i'}\rangle\langle u_{i'}|\otimes |v_{i'}\rangle\langle v_{i'}|\equiv |\psi\rangle\mbox{}_A\langle\psi|\otimes|\psi\rangle\mbox{}_B\langle\psi|.
\]
If $\rho$ is mixed, it can be written as convex combination of pure states
\[
\rho=\sum_kp_k|\psi_k\rangle\langle\psi_k|,
\]
if $\rho_A^{(k)}=\tr_A(|\psi_k\rangle\langle\psi_k|)$, the state of $A$ is
\[
\rho_A=\sum_kp_k\rho_A^{(k)},
\]
but as $\rho_A$ is pure, we have only two possibilities, either some $p_{k'}=1$ and $p_{k}=0$ for $k\neq k'$ or $\rho_A^{(k)}$ is the same state $\rho_A$ for every $k$. In the first case $\rho$ is pure obviously, for the second option let us introduce the Schmidt decomposition again for each state $|\psi_k\rangle=\sum_{i}\lambda_{i,k}|u_{i,k},v_{i,k}\rangle$:
\[
\rho_A=\rho_A^{(k)}=\tr_A(|\psi_k\rangle\langle\psi_k|)=\sum_{i}\lambda_{i,k}^2|u_{i,k}\rangle\langle u_{i,k}|,\quad \forall k
\]
so both $\lambda_{i,k}$ and $|u_{i,k}\rangle$ can be taken to be independent of $k$, and
\[
\rho_A=\sum_{i}\lambda_{i}^2|u_{i}\rangle\langle u_{i}|.
\]
This equation is actually the same as (\ref{pureNoCorr}) and so one $\lambda_{i'}$ is equal to one, and $\lambda_{i}=0$ for all $i\neq i'$. With these requirements we get $|\psi_k\rangle=|u_{i'},v_{i',k}\rangle$, and finally
\begin{multline*}
\rho=\sum_kp_k|\psi_k\rangle\langle\psi_k|=\sum_kp_k|u_{i'},v_{i',k}\rangle\langle u_{i'},v_{i',k}|\\
=|u_{i'}\rangle\langle u_{i'}|\otimes\sum_kp_k|v_{i',k}\rangle\langle v_{i',k}|\equiv|\psi\rangle\mbox{}_A\langle\psi|\otimes\rho_B,
\end{multline*}
and hence $\rho_{corr}=0$.

\end{proof}

As a result, the positivity requirement forces each term in the dynamical map to be interconnected and dependent on the state upon which they act; this means that a dynamical map with some values of $K_\alpha(t_1,t_0)$ and $\delta\rho(t_1,t_0)$ may describe a physical evolution for some states $\rho_A(t_0)$ and an unphysical evolution for others.

Given this fact, a possible approach to the dynamics of open quantum systems is by means of the study of dynamical maps and their positivity domains \cite{JSS04,SS05,Sh05,CTZ08,Modi10,CesarKavanAlan}, however that will not be our aim here.

On the other hand, the following interesting result given independently by D. Salgado \textit{et al.} \cite{SSF04} and D. M. Tong \textit{et al.} \cite{TKOJM04} provides an equivalent parametrization of dynamical maps.

\begin{theorem} Any kind of time evolution of a quantum state $\rho_A$ can always be written in the form
\begin{equation}\label{Salgado}
\rho_A(t_1)=\sum_\alpha K_{\alpha}\left(t_1,t_0,\rho_A\right)\rho_A(t_0)K_{\alpha}^\dagger\left(t_1,t_0,\rho_A\right),
\end{equation}
where $K_{\alpha}\left(t_1,t_0,\rho_A\right)$ are operators which depend on the state $\rho_A$ at time $t_0$.
\end{theorem}
\begin{proof} We give here a very simple proof based on an idea presented in \cite{CTZ08}. Let $\rho_A(t_0)$ and $\rho_A(t_1)$ be the states of a quantum system at times $t_0$ and $t_1$ respectively. Consider formally the product $\rho_A(t_0)\otimes\rho_A(t_1) \in \hilbert_A\otimes\hilbert_{A}$, and a unitary operation $U_{\mathrm{SWAP}}$ which interchanges the states in a tensor product $U_{\mathrm{SWAP}}|\psi_1\rangle\otimes|\psi_2\rangle=|\psi_2\rangle\otimes|\psi_1\rangle$, so $U_{\mathrm{SWAP}}(A\otimes B) U^\dagger_{\mathrm{SWAP}}=B\otimes A$. It is evident that the time evolution between $\rho_A(t_0)$ and $\rho_A(t_1)$ can be written as
\[
\mathcal{E}_{(t_1,t_0)}[\rho_A(t_0)]=\tr_{2}\left[U_{\mathrm{SWAP}}\rho_A(t_0)\otimes\rho_A(t_1)U^\dagger_{\mathrm{SWAP}}\right]=\rho_A(t_1), \]
where $\tr_{2}$ denotes the partial trace with respect to the second member of the composed state. By taking the spectral decomposition of $\rho_A(t_1)$ in the central term of the above equation we obtain an expression of the form (\ref{Salgado}).

\end{proof}
\noindent Note that the decomposition of (\ref{Salgado}) is clearly not unique. For more details see references \cite{SSF04} and \cite{TKOJM04}.

It will be convenient for us to understand the dependence on $\rho_A$ of the operators $K_{\alpha}\left(t_1,t_0,\rho_A\right)$ in (\ref{Salgado}) as a result of limiting the action of (\ref{KrausInno}). That is, the above theorem asserts that restricting the action of any dynamical map (\ref{KrausInno}) from its positivity domain to a small enough set of states, actually to one state ``$\rho_A$'' if necessary, is equivalent to another dynamical map without inhomogeneous term $\delta\rho(t_1,t_0)$. However this fact can also be visualized as a by-product of the specific nonlinear features of a dynamical map as (\ref{KrausInno}) \cite{Pe94,vWL00}.

As one can already guess, the absence of the inhomogeneous part will play a key role in the dynamics. Before we move on to examine this in the next section, it is worth remarking that C. A. Rodr\'iguez \textit{et al.} \cite{RMKSS08} have found that for some classes of initial separable states $\rho(t_0)$ (the ones which have vanishing quantum discord), the dynamical map (\ref{Salgado}) is the same for a set of commuting states of $\rho_A$. So we can remove the dependence on $\rho_A$ of $K_{\alpha}\left(t_1,t_0,\rho_A\right)$ when applying (\ref{Salgado}) inside of the set of states commuting with $\rho_A$. 

\subsection{Universal dynamical maps} \label{sectionUDM}

We shall call a \emph{universal dynamical map (UDM)} to a dynamical map which is independent of the state it acts upon, and so it describes a plausible physical evolution for any state $\rho_{A}$. In particular, since for pure $\rho_{A}(t_0)$,  $\rho_{corr}(t_0)=0$, according to (\ref{KrausInno}) the most general form of a UDM is given by
\begin{equation}\label{Kraus}
\rho_A(t_1)=\mathcal{E}_{(t_1,t_0)}[\rho_A(t_0)]=\sum_\alpha K_{\alpha}\left(t_1,t_0\right)\rho_A(t_0)K_{\alpha}^\dagger\left(t_1,t_0\right),
\end{equation}
in addition, the equality
\begin{equation}\label{KrausN}
\sum_\alpha K_{\alpha}^\dagger\left(t_1,t_0\right) K_{\alpha}\left(t_1,t_0\right) = \mathds{1}
\end{equation}
is fulfilled because $\tr[\rho_A(t_1)]=1$ for any initial state $\rho_{A}(t_0)$. Compare (\ref{Salgado}) and (\ref{Kraus}).

An important result is the following (see figure \ref{figuredilation}):

\begin{theorem} \label{dilation}
A dynamical map is a UDM if and only if it is induced from an extended system with the initial condition $\rho(t_0)=\rho_A(t_0)\otimes\rho_B(t_0)$ where $\rho_B(t_0)$ is fixed for any $\rho_A(t_0)$.
\end{theorem}
\begin{proof}
Given such an initial condition the result is elementary to check. Conversely, given a UDM (\ref{Kraus}) one can always consider the extended state $\rho(t_0)=\rho_A(t_0)\otimes|\psi\rangle_B\langle\psi|$, with $|\psi\rangle_B\in\hilbert_B$ fixed for every state $\rho_A(t_0)$, and so $K_{\alpha}\left(t_1,t_0\right)=\mbox{}_B\langle\phi_\alpha|U(t_1,t_0)|\psi\rangle_B$, where ${|\phi_\alpha\rangle_B}$ is a basis of $\hilbert_B$. Since $K_{\alpha}\left(t_1,t_0\right)$ only fixes a few partial elements of $U(t_1,t_0)$ and the requirement (\ref{KrausN}) is consistent with the unitarity condition
\begin{multline*}
\sum_\alpha\mbox{}_B\langle\psi|U^\dagger(t_1,t_0)|\phi_\alpha\rangle\mbox{}_B\langle\phi_\alpha|U(t_1,t_0)|\psi\rangle_B=\mbox{}_B\langle\psi|U^\dagger(t_1,t_0)U(t_1,t_0)|\psi\rangle_B\\
=\mbox{}_B\langle\psi|\mathds{1}\otimes\mathds{1}|\psi\rangle_B=\mathds{1},
\end{multline*}
such a unitary operator $U(t_1,t_0)$ always exists.

\end{proof}

\begin{figure}[h]
\centering
\includegraphics[width=0.6\textwidth]{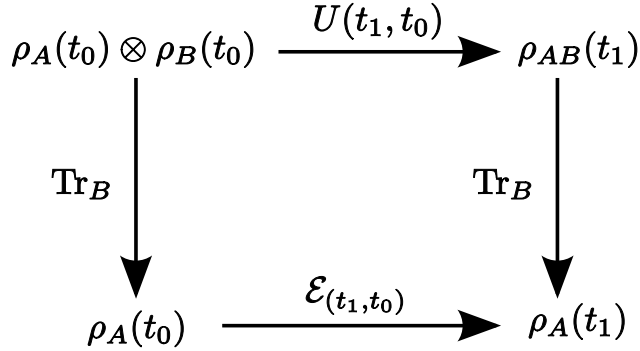}
\caption{Schematic illustration of the result presented in theorem \ref{dilation}.}\label{figuredilation}
\end{figure}

Moreover UDMs possess other interesting properties. For instance they are linear maps, that is for $\rho_A(t_0)=\sum_i\lambda_i\rho_A^{(i)}(t_0)$ $\left(\text{with }\lambda_i\geq0\text{ and }\sum_i\lambda_i=1\right)$, a straightforward consequence from definition of UDM is 
\[
\mathcal{E}_{(t_1,t_0)}[\rho_A(t_0)]=\sum_i\lambda_i\mathcal{E}_{(t_1,t_0)}[\rho_A^{(i)}(t_0)].
\]
In addition, UDMs are positive operations (i.e. the image of a positive operator is also a positive operator).

Furthermore, a UDM holds another remarkable mathematical property. Consider the following ``Gedankenexperiment''; imagine that the system $AB$ is actually part of another larger system $ABW$, where the $W$ component interact neither with $A$ nor with $B$, and so it is hidden from our eyes, playing the role of a simple bystander of the motion of $A$ and $B$. Then let us consider the global dynamics of the three subsystems. Since $W$ is disconnected from $A$ and $B$, the unitary evolution of the whole system is of the form $U_{AB}(t_1,t_0)\otimes U_W(t_1,t_0)$, where $U_W(t_1,t_0)$ is the unitary evolution of the subsystem $W$. If the initial state is $\rho_{ABW}(t_0)$ then
\[
\rho_{ABW}(t_1)=U_{AB}(t_1,t_0)\otimes U_W(t_1,t_0)\rho_{ABW}(t_0)U_{AB}^\dagger(t_1,t_0)\otimes U_W^\dagger(t_1,t_0).
\]
Note that by taking the partial trace with respect to $W$ we get
\begin{eqnarray*}
\rho_{AB}(t_1)&=&U_{AB}(t_1,t_0)\tr_W\left[U_W(t_1,t_0)\rho_{ABW}(t_0)U_W^\dagger(t_1,t_0)\right]U_{AB}^\dagger(t_1,t_0) \\
&=&U_{AB}(t_1,t_0)\tr_W\left[U_W^\dagger(t_1,t_0)U_W(t_1,t_0)\rho_{ABW}(t_0)\right]U_{AB}^\dagger(t_1,t_0)\\
&=&U_{AB}(t_1,t_0)\rho_{AB}(t_0)U_{AB}^\dagger(t_1,t_0),
\end{eqnarray*}
so indeed the presence of $W$ does not perturb the dynamics of $A$ and $B$. Since we assume that the evolution of $A$ is described by a UDM, according to the theorem \ref{dilation} the initial condition must be $\rho_{AB}(t_0)=\rho_{A}(t_0)\otimes\rho_{B}(t_0)$, and then the most general initial condition of the global system which is compatible with that has to be of the form $\rho_{ABW}(t_0)=\rho_{AW}(t_0)\otimes\rho_{B}(t_0)$, where $\rho_{B}(t_0)$ is fixed for any $\rho_{AW}(t_0)$. Now let us address our attention to the dynamics of the subsystem $AW$:
\begin{multline*}
\rho_{AW}(t_1)=\tr_B\left[U_{AB}(t_1,t_0)\otimes U_W(t_1,t_0)\rho_{AW}(t_0)\otimes\rho_{B}(t_0)\right.\\
\left. U_{AB}^\dagger(t_1,t_0)\otimes U_W^\dagger(t_1,t_0)\right],
\end{multline*}
by using the spectral decomposition of $\rho_{B}(t_0)$, as we have already done a couple of times, one obtains
\begin{eqnarray*}
\rho_{AW}(t_1)&=&\sum_\alpha K_{\alpha}\left(t_1,t_0\right)\otimes U_W(t_1,t_0)\rho_{AW}(t_0)K_{\alpha}^\dagger\left(t_1,t_0\right)\otimes U_W^\dagger(t_1,t_0)\\
&=&\mathcal{E}_{(t_1,t_0)}\otimes\mathcal{U}_{(t_1,t_0)}\left[\rho_{AW}(t_0)\right].
\end{eqnarray*}
Here $\mathcal{E}_{(t_1,t_0)}$ is a dynamical map on the subsystem $A$ and the operator $\mathcal{U}_{(t_1,t_0)}[\cdot]=U_W(t_1,t_0)[\cdot]U_W^\dagger(t_1,t_0)$ is the unitary evolution of $W$. As we see, the dynamical map on the system $AW$ is the tensor product of each individual dynamics. So if the evolution of a system is given by a UDM, $\mathcal{E}_{(t_1,t_0)}$, then for any unitary evolution $\mathcal{U}_{(t_1,t_0)}$ of any dimension, $\mathcal{E}_{(t_1,t_0)}\otimes\mathcal{U}_{(t_1,t_0)}$ is also a UDM; and in particular is a positive-preserving operation. By writing
\[
\mathcal{E}_{(t_1,t_0)}\otimes\mathcal{U}_{(t_1,t_0)}=\left[\mathcal{E}_{(t_1,t_0)}\otimes\mathds{1}\right]\left[\mathds{1}\otimes\mathcal{U}_{(t_1,t_0)}\right],
\]
as $\mathds{1}\otimes\mathcal{U}_{(t_1,t_0)}$ is a unitary operator, the condition for positivity-preserving asserts that $\mathcal{E}_{(t_1,t_0)}\otimes\mathds{1}$ is positive. Linear maps fulfilling this property are called \emph{completely positive} maps \cite{St55}.

Since we can never dismiss the possible existence of some extra system $W$, out of our control, a UDM must always be completely positive. This may be seen directly from the decomposition (\ref{Kraus}). The reverse statement, i.e. any completely positive map can be written as (\ref{Kraus}), is a famous result proved by Kraus \cite{Kraus83}; thus the formula (\ref{Kraus}) is referred as \emph{Kraus decomposition} of a completely positive map. Therefore another definition of a UDM is a (trace-preserving) linear map which is completely positive, see for instance \cite{BrPe02,NC00}.

Complete positivity is a stronger requirement than positivity, the best known example of a map which is positive but not completely positive is the transposition map (see for example \cite{NC00}), so it cannot represent a UDM. It is appropriate to stress that, despite the crystal-clear reasoning about the necessity of complete positivity in the evolution of a quantum system, it is only required if the evolution is given by a UDM, a general dynamical map may not possess this property \cite{Pe94,vWL00,JSS04,SSF04,SS05,Sh05,CTZ08}. Indeed, if $\delta\rho(t_1,t_0)\ne0$, $\mathcal{E}_{(t_1,t_0)}$ given by (\ref{KrausInno}) is not positive for every state, so certainly it cannot be completely positive. In addition note that if $\rho_{ABW}(t_0)\neq\rho_{AW}(t_0)\otimes\rho_{B}(t_0)$ we can no longer express the dynamical map on $AW$ as $\mathcal{E}_{(t_1,t_0)}\otimes\mathcal{U}_{(t_1,t_0)}$, so the physical requirement for complete positivity is missed.


\subsection{Universal dynamical maps as contractions} \label{sectionContractions}

Let us consider now the Banach space $\mathfrak{B}$ of the trace-class operators with the trace norm (see comments to the definition \ref{defBanach}). Then, as we know, an operator $\rho\in\mathfrak{B}$ is a quantum state if it is positive-semidefinite, this is
\[
\langle \psi|\rho|\psi\rangle\geq0, \quad  |\psi\rangle\in\hilbert,
\]
and it has trace one, which is equal to the trace norm because the positivity of its spectrum,
\[
\rho\in\mathfrak{B}\text{ is a quantum state }\Leftrightarrow \rho\geq0 \text{ and } \|\rho\|_1=1.
\]
We denote the set of quantum states by $\bar{\mathfrak{B}}^+\subset\mathfrak{B}$, actually the Banach space $\mathfrak{B}$ the smallest linear space which contains $\bar{\mathfrak{B}}^+$; in other words, any other linear space containing $\bar{\mathfrak{B}}^+$ also contains $\mathfrak{B}$.

We now focus our attention in possible linear transformations on the Banach space $\mathcal{E}(\mathfrak{B}):\mathfrak{B}\mapsto\mathfrak{B}$, and in particular, as we have said in section \ref{sectionBanach}, the set of these transformations is the dual Banach space $\mathfrak{B}^\ast$, in this case with the induced norm
\[
\|\mathcal{E}\|_{1}=\sup_{\rho\in\mathfrak{B}, \rho\neq0}\left\{\frac{\|\mathcal{E}(\rho)\|_1}{\|\rho\|_1}\right\}.
\]

Particularly we are interested in linear maps which leave invariant the subset $\bar{\mathfrak{B}}^+$, i.e. they connect any quantum state with another quantum state, as a UDM does. The following theorem holds

\begin{theorem} A linear map $\mathcal{E}\in\mathfrak{B}^\ast$ leaves invariant $\bar{\mathfrak{B}}^+$ if and only if it preserves the trace and is a contraction on $\mathfrak{B}$ \label{theoCPTcontractions}
\[
\|\mathcal{E}\|_{1}\leq1.
\]
\end{theorem}
\begin{proof} Assume that $\mathcal{E}\in\mathfrak{B}^\ast$ leaves invariant the set of quantum states $\bar{\mathfrak{B}}^+$, then it is obvious that it preserves the trace; moreover, it implies that for every positive $\rho$ the trace norm is preserved, so $\|\mathcal{E}(\rho)\|_1=\|\rho\|_1$. Let $\sigma\in\mathfrak{B}$ but $\sigma\notin\bar{\mathfrak{B}}^+$, then by the spectral decomposition we can split $\sigma=\sigma^+-\sigma^-$, where
\[
\begin{cases}
\sigma^+=\sum\lambda_j|\psi_j\rangle\langle\psi_j|, \text{ for }\lambda_i\geq0,\\
\sigma^-=-\sum\lambda_j|\psi_j\rangle\langle\psi_j|, \text{ for }\lambda_j<0,
\end{cases}
\]
here $\lambda_j$ are the eigenvalues and $|\psi_j\rangle\langle\psi_j|$ the corresponding eigenvectors. Note that $\sigma^+$ and $\sigma^-$ are both positive operators. Because $|\psi_j\rangle\langle\psi_j|$ are orthogonal projections we get (in other words as $\|\sigma\|_1=\sum_j|\lambda_j|$):
\[
\|\sigma\|_1=\|\sigma^+\|_1+\|\sigma^-\|_1,
\]
but then
\[
\|\mathcal{E}(\sigma)\|_1=\|\mathcal{E}(\sigma^+)-\mathcal{E}(\sigma^-)\|_1\leq\|\mathcal{E}(\sigma^+)\|+\|\mathcal{E}(\sigma^-)\|_1=\|\sigma^+\|_1+\|\sigma^-\|_1=\|\sigma\|_1,
\]
so $\mathcal{E}$ is a contraction.

Conversely, assume that $\mathcal{E}$ is a contraction and preserves the trace, then for $\rho\in\bar{\mathfrak{B}}^+$ we have the next chain of inequalities:
\[
\|\rho\|_1=\mathrm{Tr}(\rho)=\mathrm{Tr}[\mathcal{E}(\rho)]\leq\|\mathcal{E}(\rho)\|_1\leq\|\rho\|_1,
\]
so $\|\mathcal{E}(\rho)\|_1=\|\rho\|_1$ and $\mathrm{Tr}[\mathcal{E}(\rho)]=\|\mathcal{E}(\rho)\|_1$. Since $\rho\in\bar{\mathfrak{B}}^+$ if and only if $\|\rho\|_1=\mathrm{Tr}(\rho)=1$, the last equality implies that $\mathcal{E}(\rho)\in\bar{\mathfrak{B}}^+$ for any $\rho\in\bar{\mathfrak{B}}^+$.

\end{proof}

This theorem was first proven by A. Kosakowski in the references \cite{Kossakowski1,Kossakowski2}, later on, M. B. Ruskai also analyzed the necessary condition in \cite{Ruskai}.

\subsection{The inverse of a universal dynamical map}

Let us come back to the scenario of dynamical maps. Given a UDM, $\mathcal{E}_{(t_1,t_0)}$, one can wonder whether there is another UDM describing the evolution in the time reversed sense $\mathcal{E}_{(t_0,t_1)}$ in such a way that
\[
\mathcal{E}_{(t_0,t_1)}\mathcal{E}_{(t_1,t_0)}=\mathcal{E}_{(t_1,t_0)}^{-1}\mathcal{E}_{(t_1,t_0)}=\mathds{1}.
\]
Note, first of all, that it is quite easy to write a UDM which is not bijective (for instance, take $\{|\psi_n\rangle\}$ to be an orthonormal basis of the space, then define: $\mathcal{E}(\rho)=\sum_n|\phi\rangle\langle \psi_n|\rho|\psi_n\rangle\langle\phi|=\tr(\rho)|\phi\rangle\langle\phi|$, being $|\phi\rangle$ an arbitrary fixed state), and so the inverse does not always exist. But even if it is mathematically possible to invert a map, we are going to see that in general the inverse is not another UDM.

\begin{theorem}A UDM, $\mathcal{E}_{(t_1,t_0)}$, can be inverted by another UDM if and only if it is unitary $\mathcal{E}_{(t_1,t_0)}=\mathcal{U}_{(t_1,t_0)}$. \label{WignerUDM}
\end{theorem}
\begin{proof}
The proof is easy (assuming Wigner's theorem) but lengthy. Let us remove the temporal references to make the notation less cumbersome, so we write $\mathcal{E}\equiv\mathcal{E}_{(t_1,t_0)}$.
Since we have shown that UDMs are contractions on $\mathfrak{B}$,
\[
\|\mathcal{E}(\sigma)\|_1\leq\|\sigma\|_1, \quad \sigma\in\mathfrak{B}.
\]
Let us suppose that there exists a UDM, $\mathcal{E}^{-1}$, which is the inverse of $\mathcal{E}$; then by applying the above inequality twice
\[
\|\sigma\|_1=\|\mathcal{E}^{-1}\mathcal{E}(\sigma)\|_1\leq\|\mathcal{E}(\sigma)\|_1\leq\|\sigma\|_1,
\]
and we conclude that
\begin{equation}\label{Wigner1}
\|\mathcal{E}(\sigma)\|_1= \|\sigma\|_1, \quad \forall\sigma\in\mathfrak{B}.
\end{equation}
Next, note that such an invertible $\mathcal{E}$ would connect pure stares with pure states. Indeed, suppose that $|\psi\rangle\langle\psi|$ is a pure state, and $\mathcal{E}(|\psi\rangle\langle\psi|)$ is not, then we can write
\[
\mathcal{E}(|\psi\rangle\langle\psi|)=p\rho_1+(1-p)\rho_2, \quad 1>p>0,
\]
for some $\rho_1\neq\rho_2\neq\mathcal{E}(|\psi\rangle\langle\psi|)$. But after taking the inverse this equation yields
\[
|\psi\rangle\langle\psi|=p\mathcal{E}^{-1}(\rho_1)+(1-p)\mathcal{E}^{-1}(\rho_2),
\]
and $\mathcal{E}^{-1}$ is a (bijective) UDM, $\mathcal{E}^{-1}(\rho_1)\neq\mathcal{E}^{-1}(\rho_2)$ are quantum states; so such a decomposition is impossible as $|\psi\rangle\langle\psi|$ is pure by hypothesis; and a pure state cannot be decomposed as a non-trivial convex combination of another two states.

Now consider $\sigma=\frac{1}{2}(|\psi_1\rangle\langle\psi_1|-|\psi_2\rangle\langle\psi_2|)$ for two arbitrary pure states $|\psi_1\rangle\langle\psi_1|$ and $|\psi_2\rangle\langle\psi_2|$. We can calculate its eigenvalues and eigenvectors by restring the computation to the two dimensional subspace spanned by $|\psi_1\rangle$ and $|\psi_2\rangle$; because any other state $|\psi\rangle$ can be decomposed as a linear combination of $|\psi_1\rangle$, $|\psi_2\rangle$ and $|\psi_\bot\rangle$, such that $\sigma|\psi_\bot\rangle=0$. Thus the eigenvalue equation reads
\[
\sigma(\alpha|\psi_1\rangle+\beta|\psi_2\rangle)=\lambda(\alpha|\psi_1\rangle+\beta|\psi_2\rangle),
\]
for some $\lambda$, $\alpha$ and $\beta$. By taking the scalar product with respect to $\langle\psi_1|$ and $\langle\psi_2|$, and solving the two simultaneous equations one finds
\[
\lambda=\pm\frac{1}{2}\sqrt{1-|\langle\psi_2|\psi_1\rangle|^2},
\]
so the trace norm $\|\sigma\|_1=\sum_j|\lambda_j|$ turns out to be
\begin{equation}\label{Wigner2}
\|\sigma\|_1=\sqrt{1-|\langle\psi_2|\psi_1\rangle|^2}.
\end{equation}
However, since $\mathcal{E}$ connects pure states with pure states,
\[
\mathcal{E}(\sigma)=\frac{1}{2}\left[\mathcal{E}(|\psi_1\rangle\langle\psi_1|)-\mathcal{E}(|\psi_2\rangle\langle\psi_2|)\right]=\frac{1}{2}(|\tilde{\psi}_1\rangle\langle\tilde{\psi}_1|-|\tilde{\psi}_2\rangle\langle\tilde{\psi}_2|),
\]
and therefore from equations (\ref{Wigner1}) and (\ref{Wigner2}) we conclude that
\[
|\langle\psi_2|\psi_1\rangle|=|\langle\tilde{\psi}_2|\tilde{\psi}_1\rangle|.
\]
Invoking Wigner's theorem \cite{GP90,Peres}, the transformation $\mathcal{E}$ have to be implemented by
\[
\mathcal{E}=V\rho V^\dagger,
\]
where $V$ is a unitary or anti-unitary operator. Finally since the effect of an anti-unitary operator involves the complex conjugation \cite{GP90}, that is equivalent to the transposition of the elements of a trace-class self-adjoint operator, which is not a completely positive map; we conclude that $V$ has to be unitary, $\mathcal{E}_{(t_1,t_0)}=\mathcal{U}_{(t_1,t_0)}$.

\end{proof}

\noindent From another perspective, the connection between the failure of an inverse for a UDM to exist just denotes the irreversibility of the universal open quantum systems dynamics.

\subsection{Temporal continuity. Markovian evolutions}\label{sectionTemporalContinuity}

One important problem of the dynamical maps given by UDM is related with their continuity in time. This can be formulated as follows: assume that we know the evolution of a system between $t_0$ and $t_1$, $\mathcal{E}_{(t_1,t_0)}$, and between $t_1$ and a posterior time $t_2$, $\mathcal{E}_{(t_2,t_1)}$. From the continuity of time, one could expect that the evolution between $t_0$ and $t_2$ was given by the composition of applications, $\mathcal{E}_{(t_2,t_0)}=\mathcal{E}_{(t_2,t_1)}\mathcal{E}_{(t_1,t_0)}$. However, the exact meaning of this equation has to be analyzed carefully.

Indeed, let the total initial state be $\rho(t_0)=\rho_A(t_0)\otimes\rho_B(t_0)$, and the unitary evolution operators $U(t_2,t_1)$, $U(t_1,t_0)$ and $U(t_2,t_0)=U(t_2,t_1)U(t_1,t_0)$. Then the evolution of the subsystem $A$ between $t_0$ and any subsequent $t$ is clearly given by,
\begin{eqnarray*}
\begin{split}
\mathcal{E}_{(t_1,t_0)}[\rho_A(t_0)]=\tr_B[U(t_1,t_0)\rho_A(t_0)\otimes\rho_B(t_0)&U^\dagger(t_1,t_0)]\\
=&\sum_\alpha K_{\alpha}\left(t_1,t_0\right)\rho_A(t_0)K_{\alpha}^\dagger\left(t_1,t_0\right),
\end{split}\\
\end{eqnarray*}
\text{and}
\begin{eqnarray*}
\begin{split}
\mathcal{E}_{(t_2,t_0)}[\rho_A(t_0)]=\tr_B[U(t_2,t_0)\rho_A(t_0)\otimes\rho_B(t_0)&U^\dagger(t_2,t_0)]\\
=&\sum_\alpha K_{\alpha}\left(t_2,t_0\right)\rho_A(t_0)K_{\alpha}^\dagger\left(t_2,t_0\right).
\end{split}
\end{eqnarray*}
The problem arises with $\mathcal{E}_{(t_2,t_1)}$:
\[
\begin{split}
\mathcal{E}_{(t_2,t_1)}[\rho_A(t_1)]=\tr_B[U(t_2,t_1)\rho(t_1)&U^\dagger(t_2,t_1)]\\
=&\sum_\alpha K_{\alpha}\left(t_2,t_1,\rho_A\right)\rho_A(t_1)K_{\alpha}^\dagger\left(t_2,t_1,\rho_A\right),
\end{split}
\]
being $\rho(t_1)=U(t_1,t_0)[\rho_A(t_0)\otimes\rho_B(t_0)]U^\dagger(t_1,t_0)$. Since $\rho(t_1)$ is not a tensor product in general, and it depends on what initial states were taken for both subsystems, $\mathcal{E}_{(t_2,t_1)}$ has not the form of an UDM (see figure \ref{fig2Rev}).

\begin{figure}[h]
\centering
\includegraphics[width=0.9\textwidth]{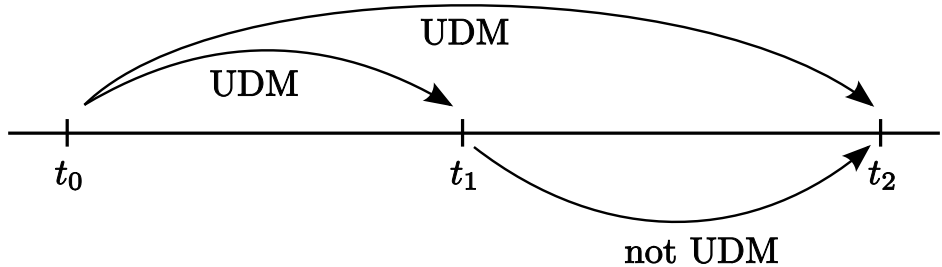}
\caption{General situation describing the dynamics of open quantum systems in terms of a UDM. At time $t_0$ the total state is assumed to be factorized, $\rho(t_0)=\rho_A(t_0)\otimes\rho_B(t_0)$, and the dynamical maps from this point are UDMs. However it is not possible to define a UDM from some intermediate time $t_1>t_0$, because generally at this time the total state can no longer be written as a tensor product.}\label{fig2Rev}
\end{figure}

Let us to illustrate this situation from another point of view. If the application $\mathcal{E}_{(t_1,t_0)}$ is bijective it seems that we can overcome the problem by defining
\begin{equation}\label{inversetrick}
\mathcal{E}_{(t_2,t_1)}=\mathcal{E}_{(t_2,t_0)}\mathcal{E}^{-1}_{(t_1,t_0)},
\end{equation}
which trivially satisfies the relation $\mathcal{E}_{(t_2,t_0)}=\mathcal{E}_{(t_2,t_1)}\mathcal{E}_{(t_1,t_0)}$. However, as we have just seen, $\mathcal{E}^{-1}_{(t_1,t_0)}$ is not completely positive unless it is unitary, so in general the evolution $\mathcal{E}_{(t_2,t_1)}$ is still not a UDM \cite{compositionCP}.

The same idea can also be used to define dynamical maps beyond UDM \cite{canonicalDM}. That is, given the global system in an arbitrary general state $\rho(t_0)$, if this is not a tensor product, we know that at least in some past time $t_{\text{initial}}$, before the subsystems started interacting, the state was factorized $\rho(t_{\text{initial}})=\rho_A(t_{\text{initial}})\otimes\rho_B(t_{\text{initial}})$. If we manage to find a unitary operation $U^{-1}(t_0,t_{\text{initial}})$ which maps $\rho(t_0)$ to the state $\rho(t_{\text{initial}})=\rho_A(t_{\text{initial}})\otimes\rho_B(t_{\text{initial}})$ for every reduced state $\rho_A(t_0)$ compatible with $\rho(t_0)$, having $\rho_B(t_{\text{initial}})$ the same for all $\rho_A(t_0)$, we can describe any posterior evolution from $t_0$ to $t_1$ as the composition of the inverse map and the UDM from $t_{\text{initial}}$ to $t_1$ as in (\ref{inversetrick}). However to find such an unitary operator can be complicated, and it will not always exist.

On the other hand, as a difference with the dynamics of closed systems, it is clear that these difficulties arising from the continuity of time for UDMs make it impossible to formulate the general dynamics of open quantum systems by means of differential equations which generate contractive families on $\mathfrak{B}$ (i.e. families of UDMs). However, sometimes one can write down a differential form for the evolution from a fixed time origin $t_0$, which is only a valid UDM to describe the evolution from that time $t_0$. In other words, the differential equation generates an eventually contractive family from $t_0$ on $\mathfrak{B}$ (see definition \ref{eventuallycontractive}).

To illustrate this, consider a bijective UDM, $\mathcal{E}_{(t,t_0)}$, then by differentiation (without the aim to be very rigorous here)
\begin{equation}\label{TCLForm}
\frac{d\rho_A(t)}{dt}=\frac{d\mathcal{E}_{(t,t_0)}}{dt}[\rho_A(t_0)]=\frac{d\mathcal{E}_{(t,t_0)}}{dt}\mathcal{E}^{-1}_{(t,t_0)}[\rho_A(t)]=\mathcal{L}_t[\rho_A(t)],
\end{equation}
where $\mathcal{L}_t=\frac{d\mathcal{E}_{(t,t_0)}}{dt}\mathcal{E}^{-1}_{(t,t_0)}$. This is the ultimate idea behind the ``time convolutionless'' method which will be briefly described in section \ref{sectionTCL}. If the initial condition is given at $t_0$, then, by construction, the solution to the above equation is $\rho_A(t)=\mathcal{E}_{(t,t_0)}[\rho_A(t_0)]$. Nevertheless, there is not certainty that the evolution from any other initial time is a UDM. Indeed, as we may guess, it is not difficult to write solutions given the initial condition $\rho(t_1)$ $(t_1>t_0)$ as (\ref{inversetrick}).

\begin{definition}\label{quantumMarkov}
We will say that a quantum systems undergoes a \emph{Markovian} evolution if it is described by a contractive evolution family on $\mathfrak{B}$ and thus we recover the composition law for the UDMs
\begin{equation}\label{CompLaw}
\mathcal{E}_{(t_2,t_0)}=\mathcal{E}_{(t_2,t_1)}\mathcal{E}_{(t_1,t_0)},
\end{equation}
which is the quantum analogue to the Chapman-Kolmogorov equation in classical Markovian process (see next section \ref{sectionMarkovClassic}). This property is sometimes also referred as ``divisibility condition''.
\end{definition}

Typically, the dynamics of an open quantum system is not Markovian, because as we have already said both systems $A$ and $B$ develop correlations $\rho_{corr}$ during their evolution and so the reduced dynamical map is not a UDM for some intermediate starting time $t_1$. However if the term which goes with $\rho_{corr}$ does not affect so much the dynamics, a Markovian model can be a good approximate description of time evolution. We will analyze this deeply in section \ref{sectionMicrosMarkov}.

\section{Quantum Markov process: mathematical structure}\label{sectionMarkov}

Before studying under which conditions can an open quantum system be approximately described by a Markovian evolution, we examine in this section the structure and properties of quantum Markovian process.

\subsection{Classical Markovian processes} \label{sectionMarkovClassic}

To motivate the adopted definition of a quantum Markov process, let us recall the definition of classical Markov processes. More detailed explanations can be found in \cite{Gardiner97,BrPe02,Norris97,Ethier05}; here we just sketch the most interesting properties for our purposes without getting too concern with mathematical rigor.

On a probability space, a \emph{stochastic process} is a family of random variables $\{X(t),t\in I\subset\mathds{R}\}$. Roughly speaking, a stochastic process is specified by a random variable $X$ depending on a parameter $t$ which usually represents time. Assume that $I$ is a countable set labeled by $n$, then a stochastic process is a \emph{Markov process} if the probability that the random variable $X$ takes a value $x_n$ at any  arbitrary time $t_n$, is conditioned only by the value $x_{n-1}$ that it took at time $t_{n-1}$, and does not depend on the values at earlier times. This condition is formulated in terms of the conditional probabilities as follows:
\begin{equation}
p(x_n,t_n|x_{n-1},t_{n-1};\ldots;x_0,t_0)=p(x_n,t_n|x_{n-1}),\quad \forall t_n\in I.
\end{equation}
This property is sometimes expressed in the following way: a Markov process does not have memory of the history of past values of $X$. They are so-named after the Russian mathematician A. Markov.

Consider now this process on a continuous interval $I$, then $t_{n-1}$ is time infinitesimally close to $t_n$. So for any $t>t'$ just from the definition of conditional probability
\[
p(x,t;x',t')=p(x,t|x',t')p(x',t'),
\]
where $p(x,t;x',t')$ is the joint probability that random variable $X$ takes value $x$ at time $t$ and $x'$ at time $t'$. By integrating with respect to $x'$ we find the relation between unconditional probabilities
\[
p(x,t)=\int dx'p(x,t|x',t')p(x',t'),
\]
and let us write it as
\begin{equation} \label{kernel}
p(x,t)=\int dx'K(x,t|x',t')p(x',t') \quad\text{with } K(x,t|x',t')=p(x,t|x',t').
\end{equation}
The Markov process is called homogeneous if $K(x,t|x',t')$ is only a function of the difference between the two time parameters involved $K(x,t|x',t')=K_{t-t'}(x|x')$.

Next, take the joint probability for any three consecutive times $t_3>t_2>t_1$ and apply again the definition of conditional probability twice
\begin{eqnarray*}
p(x_3,t_3;x_2,t_2;x_1,t_1)&=&p(x_3,t_3|x_2,t_2;x_1,t_1)p(x_2,t_2;x_1;t_1)\\
&=&p(x_3,t_3|x_2,t_2;x_1,t_1)p(x_2,t_2|x_1,t_1)p(x_1,t_1).
\end{eqnarray*}
The Markov condition implies that $p(x_3,t_3|x_2,t_2;x_1,t_1)=p(x_3,t_3|x_2,t_2)$, then by integrating over $x_2$ and dividing both sides by $p(x_1,t_1)$ one gets
\[
p(x_3,t_3|x_1,t_1)=\int dx_2 p(x_3,t_3|x_2,t_2)p(x_2,t_2|x_1,t_1),
\]
which is called \emph{Chapman-Kolmogorov equation}. Observe that with the notation of (\ref{kernel}) this is
\begin{equation}\label{Chap-Kolmo}
K(x_3,t_3|x_1,t_1)=\int dx_2K(x_3,t_3|x_2,t_2)K(x_2,t_2|x_1,t_1).
\end{equation}
Thinking of $K(x,t|x',t')$ as propagators of the time evolution between $t'$ and $t$ (see Eq. \ref{kernel}), equation (\ref{Chap-Kolmo}) express that they form evolution families with that composition law. Moreover, since $K(x,t|x',t')$ are conditional probabilities they connect any probability $p(x',t')$ with another probability $p(x,t)$ so in this sense they are universal (preserve the positivity of $p(x',t')$ and the normalization).

We can now draw a clear parallelism with a quantum setting. In the definition of quantum Markov process (definition \ref{quantumMarkov}), the role played by the probabilities $p(x,t)$ is played by the density operators $\rho(t)$ and the role of conditional probabilities $p(x,t|x',t')$ is played by UDMs $\mathcal{E}_{(t,t')}$, such that (\ref{CompLaw}) is the quantum analog to the classic Chapman-Kolmogorov equation (\ref{Chap-Kolmo}).

\subsection{Quantum Markov evolution as a differential equation}

For positive $\epsilon$, consider the difference
\[
\rho(t+\epsilon)-\rho(t)=[\mathcal{E}_{(t+\epsilon,0)}-\mathcal{E}_{(t,0)}]\rho(0)=[\mathcal{E}_{(t+\epsilon,t)}-\mathds{1}]\mathcal{E}_{(t,0)}[\rho(0)]=[\mathcal{E}_{(t+\epsilon,t)}-\mathds{1}]\rho(t),
\]
provided that the limit $\epsilon\rightarrow0$ is well defined (we assume that time evolution is smooth enough). We can obtain a linear differential equation for $\rho(t)$ (called \emph{master equation})
\begin{equation}
\frac{d\rho(t)}{dt}=\lim_{\epsilon\rightarrow0}\frac{\rho(t+\epsilon)-\rho(t)}{\epsilon}=\lim_{\epsilon\rightarrow0}\frac{[\mathcal{E}_{(t+\epsilon,t)}-\mathds{1}]}{\epsilon}\rho(t)=\mathcal{L}_t\rho(t),
\end{equation}
where by definition the \emph{generator} of the evolution is
\[
\mathcal{L}_t=\lim_{\epsilon\rightarrow0}\frac{[\mathcal{E}_{(t+\epsilon,t)}-\mathds{1}]}{\epsilon}.
\]

If some quantum evolution is Markovian according to our definition \ref{quantumMarkov}, we may wonder what is the form of these differential equations. The answer is given by the following theorem, which is the main result of this section.

\begin{theorem} \label{TeoremaMARKOV} A differential equation is a Markovian master equation if and only if it can be written in the form
\begin{equation}\label{diffMarkov}
\frac{d\rho(t)}{dt}=-i[H(t),\rho(t)]+\sum_k\gamma_k(t)\left[V_k(t)\rho(t)V_k^\dagger(t)-\frac{1}{2}\{V_k^\dagger(t)V_k(t),\rho(t)\}\right],
\end{equation}
where $H(t)$ and $V_k(t)$ are time-dependent operators, with $H(t)$ self-adjoint, and $\gamma_k(t)\geq0$ for every $k$ and time $t$.
\end{theorem}
\begin{proof} Let us show the necessity first (for this we follow a generalization of the argument given in \cite{AlickiLendi87} and \cite{BrPe02} for time-independent generators). If $\mathcal{E}_{(t_2,t_1)}$ is a UDM for any $t_2\geq t_1$ it admits a decomposition of the form
\[
\mathcal{E}_{(t_2,t_1)}[\rho]=\sum_\alpha K_\alpha(t_2,t_1)\rho K_\alpha^\dagger(t_2,t_1).
\]
Actually, let $\{F_j, j=1,\ldots,N^2\}$ be a complete orthonormal basis with respect to the Hilbert-Schmidt inner product $(F_j,F_k)_{\mathrm{HS}}=\tr(F_j^\dagger F_k)=\delta_{jk}$, in such a way that we chose $F_{N^2}=\mathds{1}/\sqrt{N}$ and the rest of operators are then traceless. Now the expansion of the Kraus operators in this basis yields
\[
\mathcal{E}_{(t_2,t_1)}[\rho]=\sum_{j,k} c_{jk}(t_2,t_1)F_j\rho F_k^\dagger,
\]
where the elements
\[
c_{jk}(t_2,t_1)=\sum_\alpha \left(F_j,K_\alpha(t_2,t_1)\right)_{\mathrm{HS}}\left(F_k,K_\alpha(t_2,t_1)\right)_{\mathrm{HS}}^\ast
\]
form a positive semidefinite matrix $\mathbf{c}(t_2,t_1)$. Indeed, for any $N^2$-dimensional vector $\mathbf{v}$ we have that
\[
\left(\mathbf{v},\mathbf{c}(t_2,t_1)\mathbf{v}\right)=\sum_{j,k} v_j^\ast c_{jk}(t_2,t_1)v_k=\sum_{\alpha}\left|\left(\sum_{k}v_k\left(F_k,K_\alpha(t_2,t_1)\right)_{\mathrm{HS}}\right)\right|^2\geq0.
\]
Since this holds for any $t_2\geq t_1$, let us take $t_1=t$ and $t_2=t+\epsilon$, then the generator reads
\begin{eqnarray}
\mathcal{L}_t(\rho)&=&\lim_{\epsilon\rightarrow0}\sum_{j,k}\frac{c_{jk}(t+\epsilon,t)F_j\rho F_k^\dagger-\mathds{1}}{\epsilon}=\lim_{\epsilon\rightarrow0}\left[\frac{1}{N}\frac{c_{N^2N^2}(t+\epsilon,t)-N}{\epsilon}\rho\right.\nonumber\\
&+&\frac{1}{\sqrt{N}}\sum_{j=1}^{N^2-1}\left(\frac{c_{jN^2}(t+\epsilon,t)}{\epsilon}F_j\rho+\frac{c_{N^2j}(t+\epsilon,t)}{\epsilon}\rho F_j^\dagger\right)\nonumber\\
&+&\left.\sum_{j,k=1}^{N^2-1}\frac{c_{jk}(t+\epsilon,t)}{\epsilon}F_j\rho F_k^\dagger\right].\label{PruebaGenerator1}
\end{eqnarray}
Now, we define the time-dependent coefficients $a_{jk}(t)$ by
\begin{eqnarray*}
a_{N^2N^2}(t)&=&\lim_{\epsilon\rightarrow0}\frac{c_{N^2N^2}(t+\epsilon,t)-N}{\epsilon},\\
a_{jN^2}(t)&=&\lim_{\epsilon\rightarrow0}\frac{c_{jN^2}(t+\epsilon,t)}{\epsilon}, \quad j=1,\ldots,N^2-1,\\
a_{jk}(t)&=&\lim_{\epsilon\rightarrow0}\frac{c_{jk}(t+\epsilon,t)}{\epsilon}, \quad j,k=1,\ldots,N^2-1,\\
\end{eqnarray*}
the following operators,
\begin{eqnarray*}
F(t)&=&\frac{1}{\sqrt{N}}\sum_{j=1}^{N^2-1}a_{jN^2}(t)F_j,\\
G(t)&=&\frac{a_{N^2N^2}}{2N}\mathds{1}+\frac{1}{2}\left[F^\dagger(t)+F(t)\right],
\end{eqnarray*}
and
\[
H(t)=\frac{1}{2i}\left[F^\dagger(t)-F(t)\right],
\]
which is self-adjoint. In terms of these operators the generator (\ref{PruebaGenerator1}) can be written like
\[
\mathcal{L}_t(\rho)=-i[H(t),\rho]+\{G(t),\rho\}+\sum_{j,k=1}^{N^2-1}a_{jk}(t)F_j\rho F_k^\dagger.
\]
Since the UDMs preserve trace for every density matrix $\rho$,
\[
0=\tr\left[\mathcal{L}_t(\rho)\right]=\tr\left\{\left[2G(t)+\sum_{j,k=1}^{N^2-1}a_{jk}(t)F_k^\dagger F_j\right]\rho\right\},
\]
and therefore we conclude that
\[
G(t)=-\frac{1}{2}\sum_{j,k=1}^{N^2-1}a_{jk}(t)F_k^\dagger F_j.
\]
Note that this is nothing but the condition $\sum_\alpha K_\alpha^\dagger(t_2,t_1)K_\alpha(t_2,t_1)=\mathds{1}$ expressed in differential way. Thus the generator adopts the form
\[
\mathcal{L}_t(\rho)=-i[H(t),\rho]+\sum_{j,k=1}^{N^2-1}a_{jk}(t)\left[F_j\rho F_k^\dagger-\frac{1}{2}\{F_k^\dagger F_j,\rho\}\right].
\]
Finally note that because of the positive semidefiniteness of $c_{jk}(t+\epsilon,t)$, the matrix $a_{jk}(t),\ j,k=1,\ldots,N^2-1$ is also positive semidefinite, so for any $t$ it can be diagonalized by mean of a unitary matrix $\mathbf{u}(t)$; that is,
\[
\sum_{j,k}u_{mj}(t)a_{jk}(t)u^\ast_{nk}(t)=\gamma_m(t)\delta_{mn},
\]
where each eigenvalue is positive $\gamma_m(t)\geq0$. Introducing a new set of operators $V_k(t)$,
\[
V_k(t)=\sum_{j=1}^{N^2-1}u^\ast_{kj}(t)F_j, \quad F_j=\sum_{k=1}^{N^2-1}u_{kj}(t)V_k(t),
\]
we obtain the desired result
\[
\mathcal{L}_t(\rho)=-i[H(t),\rho]+\sum_k\gamma_k(t)\left[V_k(t)\rho V_k^\dagger(t)-\frac{1}{2}\{V_k^\dagger(t)V_k(t),\rho\}\right],
\]
with $\gamma_k(t)\geq0, \forall k,t$.

Assume now that some dynamics is given by a differential equation as (\ref{diffMarkov}). Since it is a first order linear equation, there exists a continuous family of propagators satisfying
\begin{eqnarray*}
\mathcal{E}_{(t_2,t_0)}&=&\mathcal{E}_{(t_2,t_1)}\mathcal{E}_{(t_1,t_0)},\\
\mathcal{E}_{(t,t)}&=&\mathds{1},
\end{eqnarray*}
for every time $t_2\geq t_1\geq t_0$. It is clear that these propagators are trace preserving (as for any $\rho$, $\tr(\mathcal{L}_t\rho)=0$). So we have to prove that for any times $t_2\geq t_1$ the dynamical maps $\mathcal{E}_{(t_2,t_1)}$ are UDMs, or equivalently completely positive maps. For that we use the approximation formulas introduced in section \ref{sectionMates}; firstly the time-splitting formula (\ref{time-splitting}) allow us to write the propagator like
\begin{equation}\label{PruebaGenerator2}
\mathcal{E}_{(t_2,t_1)}=\lim_{\max|t'_{j+1}-t'_j|\rightarrow0}\prod_{j={n-1}}^{0}e^{\mathcal{L}_{t'_j}(t'_{j+1}-t'_j)},
\end{equation}
where $t_2=t'_n\geq t'_{n-1}\geq\ldots\geq t'_{0}=t_1$. Next, take any element of this product
\[
e^{\mathcal{L}_{t'_\ell}\tau_\ell}, \quad \tau_\ell=t'_{\ell+1}-t'_\ell\geq0,
\]
and since the constants $\gamma_k(t)$ are positive for every time, we can write this generator evaluated at the instant $t'_\ell$ as
\[
\mathcal{L}_{t'_\ell}(\cdot)=-i[H(t'_\ell),(\cdot)]+\sum_k\left[\tilde{V}_k(t'_\ell)(\cdot)\tilde{V}_k^\dagger(t'_\ell)-\frac{1}{2}\{\tilde{V}_k^\dagger(t'_\ell)\tilde{V}_k(t'_\ell),(\cdot)\}\right],
\]
where $\tilde{V}_k(t'_\ell)=\sqrt{\gamma_k(t'_\ell)} V_k(t'_\ell)$. Now we split this generator in several parts,
\[
\mathcal{L}_{t'_\ell}=\sum_{k=0}^M\tilde{\mathcal{L}}_k,
\]
where
\begin{eqnarray}
\tilde{\mathcal{L}}_0(\cdot)&=&-i[H(t'_\ell),(\cdot)]-\sum_m^M\frac{1}{2}\{\tilde{V}_m^\dagger(t'_\ell)\tilde{V}_m(t'_\ell),(\cdot)\},\\
\tilde{\mathcal{L}}_k(\cdot)&=&\tilde{V}_k(t'_\ell)(\cdot)\tilde{V}_k^\dagger(t'_\ell), \quad k=1,\ldots,M,
\end{eqnarray}
and here $M$ is the upper limit in the sum over $k$ in (\ref{diffMarkov}). On one hand, differentiating with respect to $\tau$ one can easily check that
\[
e^{\tilde{\mathcal{L}}_0\tau}\rho=e^{\left\{\left[-iH(t'_\ell)-\sum_k\frac{1}{2}\tilde{V}_k^\dagger(t'_\ell)\tilde{V}_k(t'_\ell)\right]\tau\right\}}\rho e^{\left\{\left[iH(t'_\ell)-\sum_k\frac{1}{2}\tilde{V}_k^\dagger(t'_\ell)\tilde{V}_k(t'_\ell)\right]\tau\right\}}
\]
which has the form $O\rho O^\dagger$, so it is completely positive for all $\tau$. On the other hand, for $\tau\geq0$
\[
e^{\tilde{\mathcal{L}}_k\tau}\rho=\sum_{m=0}^{\infty}\frac{\tau^m}{m!}\tilde{V}_k^m(t'_\ell)\rho\tilde{V}_k^{\dagger m}(t'_\ell),
\]
has also the Kraus form $\sum_m O_m\rho O^\dagger_m$, so it is completely positive for all $k$. Now everything is ready to invoke the Lie-Trotter product formula (\ref{Lie-TrotterGeneral}), since
\begin{equation*}
e^{\mathcal{L}_{t'_\ell}\tau}=e^{\left(\sum_{k=0}^M\tilde{\mathcal{L}}_k\right)\tau}=\lim_{n\rightarrow\infty}\left(\prod_{k=1}^N e^{\frac{\tilde{\mathcal{L}}_k}{n}}\right)^n
\end{equation*}
is the composition (infinite product) of completely positive maps $e^{\tilde{\mathcal{L}}_k\tau/n}$, it is completely positive for all $\tau$. Finally, because this is true for every instantaneous $t'_\ell$ the ``time-splitting'' formula (\ref{PruebaGenerator2}) asserts that $\mathcal{E}_{(t_2,t_1)}$ is completely positive for any $t_2$ and $t_1$, because it is just another composition of completely positive maps \cite{SufficiencyGernot}.

\mbox{}
\end{proof}

This result can be taken as a global consequence of the works of A. Kossakowski \cite{Kossakowski1,Kossakowski2,GoKoSh76} and G. Lindblad \cite{Lindblad76}, who analyzed this problem for the case of time-homogeneous equations, that is, when the UDMs depend only on the difference $\tau=t_2-t_1$, $\mathcal{E}_{(t_2,t_1)}=\mathcal{E}_{\tau}$ and so they form a one-parameter semigroup (not a group because the inverses are generally not a UDM, see theorem \ref{WignerUDM}). As we know, under the continuity assumption, $\mathcal{E}_{\tau}=e^{\mathcal{L}\tau}$ can be defined and the generator $\mathcal{L}$ is time-independent, moreover it has the form
\begin{equation}\label{generatorMarkov}
\mathcal{L}\rho(t)=-i[H,\rho(t)]+\sum_k\gamma_k\left[V_k\rho(t)V_k^\dagger-\frac{1}{2}\{V_k^\dagger V_k,\rho(t)\}\right].
\end{equation}
This is the result proven by Gorini, Kossakowski and Sudarshan \cite{GoKoSh76} for finite dimensional semigroups and by Lindblad \cite{Lindblad76} for infinite dimensional systems with bounded generators, that is, for uniformly continuous semigroups \cite{EngelNagel}.

\subsection{Kossakowski conditions}

It is probably worthwhile to dedicate a small section to revise the original idea of Kossakowski's seminal work \cite{Kossakowski1,Kossakowski2,GoKoSh76}, as it has clear analogies with some results in classical Markov processes. The route that A. Kossakowski and collaborators followed to arrive at (\ref{generatorMarkov}) was to derive some analogue to the classical Markovian conditions on the generator $\mathcal{L}$. That is, for classical Markov process in finite dimensional systems, a matrix $Q$ is the generator of an evolution stochastic matrix $e^{Q\tau}$ for any $\tau>0$ if only if the following conditions in its elements are satisfied:
\begin{eqnarray}
Q_{ii}\leq0,\ \forall i, \label{KolmogorovConditions1}\\
Q_{ij}\geq0,\ \forall i\neq j, \label{KolmogorovConditions2}\\
\sum_j Q_{ij}=0,\ \forall i \label{KolmogorovConditions3}.
\end{eqnarray}
The proof can be found in \cite{Norris97}. To define similar conditions in the quantum case, one uses the Lumer-Phillips theorem \ref{Lumer-PhillipsTheo}, as $\mathcal{E}_{\tau}=e^{\mathcal{L}\tau}$ has to be a contraction semigroup in the Banach space of self-adjoint matrices with respect to the trace norm. For that aim, consider the next definition.

\begin{definition}
Let $\sigma\in\mathfrak{B}$, with spectral decomposition $\sigma=\sum_j\sigma_jP_j$. We define the linear operator \emph{sign} of $\sigma$ as
\[
\sgn(\sigma)=\sum_j\sgn(\sigma_j)P_j,
\]
where $\sgn(\sigma_j)$ is the sign function of each eigenvalue, that is
\[
\sgn(\sigma_j)=\left\{\begin{array}{ll}
1,& \sigma_j>0\\
0,& \sigma_j=0\\
-1,& \sigma_j<0.
\end{array}
\right.
\]
\end{definition}
\begin{definition}
We define the following product of two elements $\sigma,\rho\in\mathfrak{B}$,
\begin{equation}\label{KossakowskiProduct}
[\sigma,\rho]_K=\Vert \sigma\Vert\tr[\sgn(\sigma)\rho].
\end{equation}
\end{definition}
\begin{proposition}
The product (\ref{KossakowskiProduct}) is a semi-inner product. In addition it is a real number.
\end{proposition}
\begin{proof}
The proof is quite straightforward. It is evident that for the product (\ref{KossakowskiProduct}), (\ref{sip1}) and (\ref{sip2}) hold. Finally (\ref{sip3}) also follows easily. By introducing the spectral decomposition of $\sigma=\sum_j\sigma_jP_j^\sigma$,
\begin{eqnarray*}
|[\sigma,\rho]_K|=\Vert \sigma\Vert\left|\sum_j\sgn(\sigma_j)\tr(P_j^\sigma\rho)\right|&\leq&\Vert \sigma\Vert\sum_j\lvert\sgn(\sigma_j)\rvert\lvert\tr(P_j^\sigma \rho)\rvert\\
&\leq& \Vert \sigma\Vert\sum_j\lvert\tr(P_j^\sigma \rho)\rvert;
\end{eqnarray*}
and we also use the spectral decomposition of $\rho=\sum_j\rho_jP_j^\rho$, to obtain
\begin{eqnarray*}
\sum_j\lvert\tr(P_j^\sigma\rho)\rvert&=&\sum_j\left\lvert\sum_k\rho_k\tr(P_j^\sigma P_k^\rho)\right\rvert\leq\sum_{k,j}\lvert \rho_k\rvert\lvert\tr(P_j^\rho P_k^\sigma)\rvert\\
&=&\sum_{k,j}|\rho_k|\tr(P_j^\sigma P_k^\rho)=\sum_{k}|\rho_k|\tr(P_k^\rho)=\sum_{k}|\rho_k|=\Vert \rho\Vert.
\end{eqnarray*}
The reality of $[\sigma,\rho]_K$ is clear from the fact that both $\sgn(\sigma)$ and $\rho$ are self-adjoint.

\end{proof}

\noindent Now we can formulate the equivalent to the conditions (\ref{KolmogorovConditions1}), (\ref{KolmogorovConditions2}) and (\ref{KolmogorovConditions3}) in the quantum case.
\begin{theorem}[Kossakowski Conditions]
A linear operator $\mathcal{L}$ on $\mathfrak{B}$ (of finite dimension $N$) is the generator of a trace preserving contraction semigroup if and only if for every resolution of the identity $\mathfrak{P}=(P_1,P_2,\ldots,P_N)$, $\mathds{1}=\sum_jP_j$, the relations
\begin{eqnarray}
A_{ii}(\mathfrak{P})&\leq&0\quad (i=1,2,\ldots,N),\label{-Koss}\\
A_{ij}(\mathfrak{P})&\geq&0\quad (i\ne j=1,2,\ldots,N),\label{+Koss}\\
\sum_{i=1}^N A_{ij}(\mathfrak{P})&=&0\quad(j=1,2,\ldots,N)\label{trpKoss}
\end{eqnarray}
are satisfied, where
\[
A_{ij}(\mathfrak{P})=\tr[P_i\mathcal{L}(P_j)].
\]
\end{theorem}
\begin{proof}
The condition (\ref{trpKoss}) is just the trace preserving requirement $\tr[\mathcal{L}(\sigma)]=0$, $\sigma\in\mathfrak{B}$. We have
\begin{eqnarray*}
\sum_{i=1}^N A_{ij}(\mathfrak{P})=\sum_{i}^N\tr[P_i\mathcal{L}(P_j)]=\tr[\mathcal{L}(P_j)]=0, \quad \forall P_j\in\mathfrak{P}\\
\Leftrightarrow \tr[\mathcal{L}(\sigma)]=0,\quad \forall \sigma\in\mathfrak{B}.
\end{eqnarray*}

On the other hand, the theorem of Lumer-Phillips \ref{Lumer-PhillipsTheo} asserts that
\begin{equation}\label{LPK}
[\sigma,\mathcal{L}(\sigma)]_K\leq0;
\end{equation}
taking $\sigma$ to be some projector $P$ we obtain
\begin{equation}\label{LPK1}
\tr[P,\mathcal{L}(P)]\leq0,
\end{equation}
so (\ref{-Koss}) is necessary. Furthermore (\ref{LPK}) can be rewritten as
\begin{eqnarray}\label{LPK2}
[\sigma,\mathcal{L}(\sigma)]_K&=&\sum_{k,j}\sgn(\sigma_k)\sigma_j\tr[P_k\mathcal{L}(P_j)]=\sum_j|\sigma_j|\tr[P_j\mathcal{L}(P_j)]\nonumber\\
&+&\sum_{k\ne j}\sgn(\sigma_k)\sigma_j\tr[P_k\mathcal{L}(P_j)]\leq0,
\end{eqnarray}
where we have used the spectral decomposition $\sigma=\sum_j\sigma_jP_j$ again. By splitting (\ref{trpKoss}) in a similar way
\[
\tr[P_j\mathcal{L}(P_j)]=-\sum_{k\ne j}\tr[P_k\mathcal{L}(P_j)],\quad \text{for each }j=1,\ldots,N.
\]
Here note that on the right hand side the sum is just over the index $k$ (not over $j$). Inserting this result in (\ref{LPK2}) yields
\begin{equation}\label{LPK3}
\sum_{k\ne j}\alpha_{kj}\tr[P_k\mathcal{L}(P_j)]\geq0,
\end{equation}
where the sum is over the two indexes, with
\[
\alpha_{kj}=|\sigma_j|[1-\sgn(\sigma_k)\sgn(\sigma_j)]\geq0.
\]
Thus for trace-preserving semigroups the condition (\ref{LPK2}) is equivalent to (\ref{-Koss}) and (\ref{LPK3}). Therefore if (\ref{-Koss}), (\ref{+Koss}) and (\ref{trpKoss}) are satisfied $\mathcal{L}$ is the generator of a trace preserving contraction semigroup because of (\ref{LPK1}) and (\ref{LPK3}). It only remains to prove the necessity of (\ref{+Koss}), that can be seen directly from (\ref{LPK3}). However, it is also possible to show it from the definition of generator:
\[
\tr[P_k\mathcal{L}(P_j)]=\lim_{\tau\rightarrow0}\frac{\tr[P_k(\mathcal{E}_\tau-\mathds{1})P_j]}{\tau}=\lim_{\tau\rightarrow0}\frac{\tr[P_k\mathcal{E}_\tau (P_j)]}{\tau}\geq0,
\]
since $\mathcal{E}_\tau (P_j)\geq0$ as $\mathcal{E}_\tau$ is a trace preserving and completely positive map.

\end{proof}

Finally from this result, one can prove (\ref{generatorMarkov}). For that one realizes (similarly to the necessity proof of theorem \ref{TeoremaMARKOV}) that a trace and self-adjoint preserving generator can be written as
\[
\mathcal{L}\rho=-i[H,\rho]+\sum_{j,k=1}^{N^2-1}a_{jk}\left[F_j\rho F_k^\dagger-\frac{1}{2}\{F_k^\dagger F_j,\rho\}\right],
\]
where $H=H^\dagger$ and the coefficients $a_{jk}$ form a Hermitian matrix. $\mathcal{E}_{\tau}$ is completely positive (not just positive) if and only if $\mathcal{L}\otimes\mathds{1}$ (in the place of $\mathcal{L}$) satisfies the conditions (\ref{-Koss}), (\ref{+Koss}) and (\ref{trpKoss}). It is easy to see \cite{GoKoSh76} this implies that $a_{jk}$ is positive semidefinite, and then one obtains (\ref{generatorMarkov}) by diagonalization. The matrix $a_{jk}$ is sometimes referred to as Kossakowski matrix \cite{Benatti05,Alexandra08}.

\subsection{Steady states of homogeneous Markov processes}\label{sectionSteady}

In this section we give a brief introduction to the steady state properties of homogeneous Markov processes, i.e. completely positive semigroups. This is a complicated topic and still an active area of research, so we only expose here some simple but important results.

\begin{definition} A semigroup $\mathcal{E}_\tau$ is relaxing if there exist a unique (steady) state $\rho_\mathrm{ss}$ such that $\mathcal{E}_\tau\left(\rho_\mathrm{ss}\right)=\rho_\mathrm{ss}$ for all $\tau$ and
\[
\lim_{\tau\rightarrow\infty}\mathcal{E}_\tau\left(\rho\right)=\rho_\mathrm{ss},
\]
for every initial state $\rho$.
\end{definition}

Of course if one state $\rho_\mathrm{ss}$ is a steady state of a semigroup $\mathcal{E}_\tau$, it is an eigenoperator of the operator $\mathcal{E}_\tau$ with eigenvalue 1,
\[
\mathcal{E}_\tau\left(\rho_\mathrm{ss}\right)=e^{\mathcal{L}\tau}\rho_\mathrm{ss}=\rho_\mathrm{ss},
\]
and, by differentiating, it is an eigenoperator of the generator with zero eigenvalue,
\begin{equation}\label{steadystate1}
\mathcal{L}(\rho_\mathrm{ss})=0.
\end{equation}
If this equation admits more than one solution then the semigroup is obviously not relaxing, and the final state will depend on the initial state, or might even oscillate inside of the set of convex combinations of solutions. However, the uniqueness of the solution of equation (\ref{steadystate1}) should not presumably be a guarantee for the semigroup to be relaxing. First of all let us point out the next result.

\begin{theorem} \label{theoexiststeady}A completely positive semigroup $\mathcal{E}_\tau=e^{\mathcal{L}\tau}$ on a finite dimensional Banach space $\mathfrak{B}$ has always at least one steady state, this is, there exists at least a state $\rho_\mathrm{ss}$ which is a solution of (\ref{steadystate1}).
\end{theorem}
\begin{proof} We prove this by construction. Take any initial state $\rho(0)$; we calculate the average state of its evolution by (this is sometimes called ``ergodic average'')
\begin{equation}\label{averagedstate}
\bar{\rho}=\lim_{T\rightarrow\infty} \frac{1}{T}\int_0^Te^{\mathcal{L}\tau'}\rho(0)d\tau'.
\end{equation}
It is simple to show that the integral converges just by taking norms
\[
\lim_{T\rightarrow\infty} \left\|\frac{1}{T}\int_0^Te^{\mathcal{L}\tau}\rho(0)d\tau\right\|\leq\lim_{T\rightarrow\infty} \frac{1}{T}\int_0^T\|e^{\mathcal{L}\tau}\|\|\rho(0)\|d\tau\leq1.
\]
Because $e^{\mathcal{L}\tau}$ is completely positive, it is obvious that $\bar{\rho}$ is a state. Moreover, $\bar{\rho}$ is clearly a steady state
\[
e^{\mathcal{L}\tau}\bar{\rho}=\bar{\rho}.
\]
\end{proof}

The easiest way to show the steady state properties of a finite dimensional semigroup is by writing the generator and the exponential in matrix form. In order to do that, let us consider again the Hilbert-Schmidt inner product on $\mathfrak{B}$ and a set of mutually orthonormal basis $\{F_j, j=1,\ldots,N^2\}$, that is $(F_j,F_k)_{\mathrm{HS}}=\tr(F_j^\dagger F_k)=\delta_{jk}$. Then the matrix elements of the generator $\mathcal{L}$ in that basis will be
\[
\mathcal{L}_{jk}=(F_j,\mathcal{L}F_k)_{\mathrm{HS}}=\tr[F_j^\dagger \mathcal{L}(F_k)].
\]
To calculate the exponential of $\mathcal{L}_{jk}$ one uses the Jordan block form (in general $\mathcal{L}_{jk}$ is not diagonalizable). That is, we look for a basis such that $\mathcal{L}_{jk}$ can be written as a block-diagonal matrix, $\mathbf{S}^{-1}\boldsymbol{\mathcal{L}}\mathbf{S}=\boldsymbol{\mathcal{L}}^{(1)}\oplus\boldsymbol{\mathcal{L}}^{(2)}\oplus\ldots\oplus\boldsymbol{\mathcal{L}}^{(M)}$, of $M$ blocks where $\mathbf{S}$ is the matrix for the change of basis, and each block has the form
\[
\boldsymbol{\mathcal{L}}^{(i)}=\left[\begin{array}{ccccc}
\lambda_i      & 1         & 0       &  \ldots & 0         \\
0              & \lambda_i & 1       &  \ddots & \vdots    \\
\vdots         & \ddots    & \ddots  &  \ddots & 0         \\
\vdots         &           & \ddots  &  \ddots & 1         \\
0              & \ldots    & \ldots  &  0      & \lambda_i
\end{array}\right].
\]
Here $\lambda_i$ denotes some eigenvalue of $\boldsymbol{\mathcal{L}}$ and the size of the block coincides with its algebraic multiplicity $k_i$. It is a well-know result (see for example references \cite{EngelNagel,HornJonson}) that the exponential of $\boldsymbol{\mathcal{L}}$ can be explicitly computed as
\begin{eqnarray}
e^{\boldsymbol{\mathcal{L}}\tau}&=&\mathbf{S}e^{\boldsymbol{\mathcal{L}}^{(1)}\tau}\oplus e^{\boldsymbol{\mathcal{L}}^{(2)}\tau}\oplus\ldots\oplus e^{\boldsymbol{\mathcal{L}}^{(M)}\tau}\mathbf{S}^{-1}\nonumber\\
&=&\mathbf{S}\left(e^{\lambda_1\tau}e^{\textbf{N}_1\tau}\right)\oplus \left(e^{\lambda_2\tau}e^{\textbf{N}_2\tau}\right)\oplus\ldots\oplus \left(e^{\lambda_M\tau}e^{\textbf{N}_M\tau}\right)\mathbf{S}^{-1},\label{JordanExponential}
\end{eqnarray}
where the matrices $\textbf{N}_i=\boldsymbol{\mathcal{L}}^{(i)}-\lambda_i\mathbf{1}$ are nilpotent matrices, $\left(\textbf{N}_i\right)^{k_i}=0$, and thus their exponentials are easy to deal with. If $\boldsymbol{\mathcal{L}}$ is a diagonalizable matrix, $\textbf{N}_i=0$ for every eigenvalue.

Now we move on to the principal result of this section.
\begin{theorem} A completely positive semigroup is relaxing if and only if the zero eigenvalue of $\mathcal{L}$ is non-degenerate and the rest of eigenvalues have negative real part. \label{theoremSteady}
\end{theorem}

\begin{proof} For the ``only if'' part. The condition of non-degenerate zero eigenvalue of $\mathcal{L}$ is trivial, otherwise the semigroup has more than one steady state and is not relaxing by definition. But even if this is true, it is necessary that the rest of eigenvalues have negative real part. Indeed, consider some eigenvalue with purely imaginary part $\mathcal{L}(\sigma)=i\phi\sigma$, where $\sigma$ is the associated eigenvector. By using the fact that $e^{\mathcal{L}\tau}$ is trace preserving
\[
e^{\mathcal{L}\tau}(\sigma)=e^{i\phi\tau}\sigma\Rightarrow \tr\left[e^{\mathcal{L}\tau}(\sigma)\right]=\tr(\sigma)=e^{i\phi\tau}\tr(\sigma)\Leftrightarrow\tr(\sigma)=0.
\]
So we can split $\sigma$ as a difference of positive operators $\sigma=\sigma_+-\sigma_-$, given by its spectral decomposition as in the proof of theorem \ref{theoCPTcontractions}. Since $\tr(\sigma)=0$, $\tr(\sigma_+)=\tr(\sigma_-)=c$, and we can write $\sigma/c=\rho_+-\rho_-$, where $\rho_\pm$ are now quantum states. The triangle inequality with the steady state $\rho_\mathrm{ss}$ reads
\begin{eqnarray*}
\|e^{\mathcal{L}\tau}(\sigma/c)\|&=&\|e^{\mathcal{L}\tau}(\rho_+-\rho_-)\|\leq\|e^{\mathcal{L}\tau}(\rho_+-\rho_\mathrm{ss})\|+\|e^{\mathcal{L}\tau}(\rho_--\rho_\mathrm{ss})\|\\
&=&\|e^{\mathcal{L}\tau}(\rho_+)-\rho_\mathrm{ss}\|+\|e^{\mathcal{L}\tau}(\rho_-)-\rho_\mathrm{ss}\|.
\end{eqnarray*}
Thus if the semigroup is relaxing we obtain $\lim_{\tau\rightarrow\infty}\|e^{\mathcal{L}\tau}(\sigma/c)\|=0$, but this is impossible as $e^{\mathcal{L}\tau}(\sigma)=e^{i\phi\tau}\sigma$ and
\[
\|e^{\mathcal{L}\tau}(\rho_+-\rho_-)\|=\|e^{\mathcal{L}\tau}(\sigma/c)\|=|e^{i\phi\tau}|\|(\sigma/c)\|=\|\sigma/c\|>0,\quad \forall \tau.
\]

Conversely, let us denote $\lambda_1=0$ the non-degenerate zero eigenvalue of $\mathcal{L}$, since the remaining ones have negative real part, according to (\ref{JordanExponential}) in the limit $\tau\rightarrow\infty$ the matrix form of the semigroup will be
\begin{eqnarray}
\boldsymbol{\mathcal{E}}_{\infty}&=&\lim_{\tau\rightarrow\infty}e^{\boldsymbol{\mathcal{L}}\tau}=\lim_{\tau\rightarrow\infty}\mathbf{S}\left(e^{\lambda_1\tau}\right)\oplus \left(e^{\lambda_2\tau}e^{\textbf{N}_2\tau}\right)\oplus\ldots\oplus \left(e^{\lambda_M\tau}e^{\textbf{N}_M\tau}\right)\mathbf{S}^{-1}\nonumber\\
&=&\mathbf{S}\left(1\oplus0\oplus\ldots\oplus 0\right)\mathbf{S}^{-1}=\mathbf{S}\mathbf{D}_1\mathbf{S}^{-1},
\end{eqnarray}
where
\[
\mathbf{D}_1=\left(\begin{array}{cccc}
1      & 0      & \cdots & 0      \\
0      & 0      & \cdots & 0      \\
\vdots & \vdots & \ddots & \vdots \\
0      & 0      & \cdots & 0      \\
\end{array}\right)
\]
coincides with the projector on the eigenspace associated with the eigenvalue $1$. From here it is quite simple to see that the application of $\boldsymbol{\mathcal{E}}_{\infty}$ on any initial state $\rho$ written as a vector in the basis of $\{F_j, j=1,\ldots,N^2\}$ will give us just the steady state. More explicitly, if $\rho_\mathrm{ss}$ is the steady state, its vector representation will be
\[
\mathbf{v}_\mathrm{ss}=\left[(F_1,\rho_\mathrm{ss})_\mathrm{HS},\ldots,(F_{N^2},\rho_\mathrm{ss})_\mathrm{HS}\right]^\mathrm{t},
\]
and the condition $e^{\mathcal{L}\tau}\rho_\mathrm{ss}=\rho_\mathrm{ss}$ is of course expressed in this basis as
\begin{equation}
e^{\boldsymbol{\mathcal{L}}\tau}\mathbf{v}_\mathrm{ss}=\mathbf{v}_\mathrm{ss} \label{steadystateMatrix}.
\end{equation}
On the other hand, as we already have said, $\mathbf{D}_1$ is the projector on the eigenspace associated with the eigenvalue $1$, which is expanded by $\mathbf{v}_\mathrm{ss}$, so we can write $\mathbf{D}_1=\mathbf{v}_\mathrm{ss}(\mathbf{v}_\mathrm{ss},\cdot)$ or $\mathbf{D}_1=|\mathbf{v}_\mathrm{ss}\rangle\langle\mathbf{v}_\mathrm{ss}|$ in Dirac notation.

As a particular case ($\tau\rightarrow\infty$) of (\ref{steadystateMatrix}), we have $\boldsymbol{\mathcal{E}}_{\infty}\mathbf{v}_\mathrm{ss}=\mathbf{v}_\mathrm{ss}$, by making a little bit of algebra with this condition
\begin{eqnarray}
\boldsymbol{\mathcal{E}}_{\infty}\mathbf{v}_\mathrm{ss}=\mathbf{S}\mathbf{D}_1\mathbf{S}^{-1}\mathbf{v}_\mathrm{ss}=\mathbf{v}_\mathrm{ss}&\Rightarrow & \mathbf{S}\mathbf{v}_\mathrm{ss}(\mathbf{v}_\mathrm{ss},\mathbf{S}^{-1}\mathbf{v}_\mathrm{ss})=\mathbf{v}_\mathrm{ss}\nonumber \\ &\Rightarrow& \mathbf{S}\mathbf{v}_\mathrm{ss}=\tfrac{1}{(\mathbf{v}_\mathrm{ss},\mathbf{S}^{-1}\mathbf{v}_\mathrm{ss})}\mathbf{v}_\mathrm{ss} \label{steadystate2},
\end{eqnarray}
therefore $\mathbf{v}_\mathrm{ss}$ is eigenvector of $\mathbf{S}$ with eigenvalue $\tfrac{1}{(\mathbf{v}_\mathrm{ss},\mathbf{S}^{-1}\mathbf{v}_\mathrm{ss})}$.

Finally consider any initial state $\rho$ written as a vector $\mathbf{v}_\rho$ in the orthonormal basis $\{F_j, j=1,\ldots,N^2\}$. We have
\[
\lim_{\tau\rightarrow\infty}e^{\boldsymbol{\mathcal{L}}\tau}\mathbf{v}_\rho=\boldsymbol{\mathcal{E}}_{\infty}\mathbf{v}_\rho=\mathbf{S}\mathbf{D}_1\mathbf{S}^{-1}\mathbf{v}_\rho=\mathbf{S}\mathbf{v}_\mathrm{ss}(\mathbf{v}_\mathrm{ss},\mathbf{S}^{-1}\mathbf{v}_\rho),
\]
thus using (\ref{steadystate2}),
\[
\boldsymbol{\mathcal{E}}_{\infty}\mathbf{v}_\rho=\frac{(\mathbf{v}_\mathrm{ss},\mathbf{S}^{-1}\mathbf{v}_\rho)}{(\mathbf{v}_\mathrm{ss},\mathbf{S}^{-1}\mathbf{v}_\mathrm{ss})}\mathbf{v}_\mathrm{ss}= C(\rho)\mathbf{v}_\mathrm{ss}, \quad C(\rho)\equiv\frac{(\mathbf{v}_\mathrm{ss},\mathbf{S}^{-1}\mathbf{v}_\rho)}{(\mathbf{v}_\mathrm{ss},\mathbf{S}^{-1}\mathbf{v}_\mathrm{ss})},
\]
and in the notation of states and operators this equation reads
\[
\lim_{\tau\rightarrow\infty}e^{\mathcal{L}\tau}\rho=C(\rho)\rho_{\mathrm{ss}}, \quad \text{for all }\rho.
\]
One can compute $C(\rho)$ explicitly and find that $C(\rho)=1$ for every $\rho$, however this follows from the condition that $e^{\mathcal{L}\tau}$ is a trace preserving semigroup.

\end{proof}

This theorem can be taken as a weaker version of the result by Schirmer and Wang \cite{Schirmer-Wang}, which assert that if the steady state is unique then the semigroup is relaxing (there are not purely imaginary eigenvalues). A particular simple case is if the unique steady state is pure $\rho_{ss}=|\psi\rangle\langle\psi|$. Then since a pure state cannot be expressed as a nontrivial convex combination of other states, the equation (\ref{averagedstate}) implies that the semigroup is relaxing. This was also stated for an arbitrary UDM in theorem 8 of reference \cite{Bougarth}.

On the other hand, it is desirable to have some condition on the structure of the generator which assures that the semigroup is relaxing as well as a characterization of the steady states; there are several results on this topic (see \cite{Spohn80,Schirmer-Wang,Davies70,Spohn76,Spohn77,Evans77,Frigerio77,Frigerio78,Frigerio-Spohn,Baumgartner1,Baumgartner2,Ticozzi1,Ticozzi2,Ticozzi3,Ticozzi4} and references therein). Here we will only present an important result given by Spohn in \cite{Spohn77}.

\begin{theorem}[Spohn] Consider a completely positive semigroup, $\mathcal{E}_\tau=e^{\mathcal{L}\tau}$, in $\mathfrak{B}$ with generator \label{theoremSpohn}
\[
\mathcal{L}\rho=-i[H,\rho]+\sum_{k\in I}\gamma_k\left[V_k\rho V_k^\dagger-\frac{1}{2}\{V_k^\dagger V_k,\rho\}\right],
\]
for some set of indexes $I$. Provided that the set $\{V_k,k\in I\}$ is self-adjoint (this is, the adjoint of every element of the set is inside of the set) and the only operators commuting with all of them are proportional to the identity, which is expressed as $\{V_k,k\in I\}'=c\mathds{1}$, the semigroup $\mathcal{E}_\tau$ is relaxing.
\end{theorem}
\begin{proof} Since $\{V_k,k\in I\}$ is self-adjoint we can consider a self-adjoint orthonormal basis with respect to the Hilbert-Schmidt product $\{F_j, j=1,\ldots,N^2\}$, such that, as in the proof of theorem \ref{TeoremaMARKOV}, we chose $F_{N^2}=\mathds{1}/\sqrt{N}$, and the remaining elements are traceless. Expanding
\begin{equation}\label{steadystate3}
V_k=\sum_{m=1}^{N^2-1}v_{km}F_m+\frac{v_{k,N^2}}{\sqrt{N}}\mathds{1},
\end{equation}
we rewrite the generator as
\begin{equation}\label{steadystate4}
\mathcal{L}\rho=-i[H+\bar{H},\rho]+\sum_{m,n=1}^{N^2-1}a_{mn}\left[F_m\rho F_n-\frac{1}{2}\{F_n F_m,\rho\}\right],
\end{equation}
where $\bar{H}$ is the self-adjoint operator
\[
\bar{H}=\sum_{k\in I}\frac{\gamma_k}{2i\sqrt{N}}\sum_{m=1}^{N^2-1}\left(v_{k,N^2}v^\ast_{km}-v^\ast_{k,N^2}v_{km}\right)F_m,
\]
and the Hermitian matrix is $a_{mn}=\sum_{k\in I}\gamma_kv_{km}v^\ast_{kn}$. Now note that we can trivially chose $a_{mn}$ to be strictly positive; for that it is enough to diagonalize again $a_{mn}$ and ignore the members in the sum of (\ref{steadystate4}) which goes with zero eigenvalue as they do not contribute, and write again the less dimensional matrix $a_{mn}$ in terms of the numbers of elements of the basis strictly necessary, say $p\leq{N^2-1}$,
\[
\mathcal{L}\rho=-i[H+\bar{H},\rho]+\sum_{m,n=1}^{p}a_{mn}\left[F_m\rho F_n-\frac{1}{2}\{F_n F_m,\rho\}\right].
\]
Next, consider some $\alpha>0$ to be smaller than the smallest eigenvalue of the strictly positive $a_{mn}$, and split the generator in two pieces $\mathcal{L}=\mathcal{L}_1+\mathcal{L}_2$ with
\[
\mathcal{L}_2=\alpha\sum_{m=1}^p\left[F_m\rho F_m-\frac{1}{2}\{F_m F_m,\rho\}\right].
\]
It is clear that $\mathcal{L}_2$ is the generator of a completely positive semigroup, and the same for $\mathcal{L}_1$ since $a_{mn}-\alpha\delta_{mn}$ is still a positive semidefinite matrix. Next, for some $\sigma\in\mathfrak{B}$ we compute
\begin{eqnarray*}
(\sigma,\mathcal{L}_2(\sigma))_\mathrm{HS}&=&\tr[\sigma\mathcal{L}_2(\sigma)]=\alpha\sum_{m=1}^p\tr[(\sigma F_m)^2-\sigma^2F_m^2]\\
&=&\frac{\alpha}{2}\sum_{m=1}^p\tr\left([F_m,\sigma][F_m,\sigma]\right)=-\frac{\alpha}{2}\sum_{m=1}^p\tr\left([F_m,\sigma]^\dagger[F_m,\sigma]\right)\leq0,
\end{eqnarray*}
so $\mathcal{L}_2$ has only real eigenvalues smaller or equal to zero. Moreover for the eigenvalue $0$ we have $\tr[\sigma\mathcal{L}_2(\sigma)]=0$ which requires $[F_m,\sigma]=0$ for $m=1,...,p$. But, because of the equation  (\ref{steadystate3}), it implies that $[V_k,\sigma]=0$ for $k\in I$ and therefore $\sigma=c\mathds{1}$ by presupposition.

Now, let $\tilde{\mathcal{L}}$, $\tilde{\mathcal{L}}_1$ and $\tilde{\mathcal{L}}_2$ denote the restriction of $\mathcal{L}$, $\mathcal{L}_1$ and $\mathcal{L}_2$ to the subspace generated by $\{F_j, j=1,\ldots,N^2-1\}$, this is the subspace of traceless self-adjoint operators. Then since $\mathcal{L}_1$ is the generator of a contracting semigroup all the eigenvalues of $\tilde{\mathcal{L}}_1$ have real part smaller or equal to zero. On the other hand from the foregoing we conclude that the spectrum of $\tilde{\mathcal{L}}_2$ is composed only of strictly negative real numbers, this implies that $\tilde{\mathcal{L}}=\tilde{\mathcal{L}}_1+\tilde{\mathcal{L}}_2$ has only eigenvalues with strictly negative real part. Indeed, writing $\tilde{\mathcal{L}}$, $\tilde{\mathcal{L}}_1$ and $\tilde{\mathcal{L}}_2$  in its matrix representation, $\tilde{\boldsymbol{\mathcal{L}}}$, $\tilde{\boldsymbol{\mathcal{L}}}_1$ and $\tilde{\boldsymbol{\mathcal{L}}}_2$, we can split them in their Hermitian and anti-Hermitian parts, e.g.
\[
\tilde{\boldsymbol{\mathcal{L}}}=\tilde{\boldsymbol{\mathcal{L}}}_A+i\tilde{\boldsymbol{\mathcal{L}}}_B,
\]
where $\tilde{\boldsymbol{\mathcal{L}}}_A=\left(\tilde{\boldsymbol{\mathcal{L}}}+\tilde{\boldsymbol{\mathcal{L}}}^\dagger\right)/2$ and $\tilde{\boldsymbol{\mathcal{L}}}_B=\left(\tilde{\boldsymbol{\mathcal{L}}}-\tilde{\boldsymbol{\mathcal{L}}}^\dagger\right)/2i$ are Hermitian matrices. Similarly for $\tilde{\boldsymbol{\mathcal{L}}}_1$ and $\tilde{\boldsymbol{\mathcal{L}}}_2$. It is quite evident that the real part of the spectrum of $\tilde{\mathcal{L}}$ is strictly negative if and only if $\tilde{\boldsymbol{\mathcal{L}}}_A$ is a negative definite matrix, $\left(\mathbf{v}_\sigma,\tilde{\boldsymbol{\mathcal{L}}}_A\mathbf{v}_\sigma\right)<0$ for any $\sigma\in\mathrm{Span}\{F_j, j=1,\ldots,N^2-1\}$. But, since $\tilde{\boldsymbol{\mathcal{L}}}_A=\tilde{\boldsymbol{\mathcal{L}}}_{A1}+\tilde{\boldsymbol{\mathcal{L}}}_{A2}$ where
$\tilde{\boldsymbol{\mathcal{L}}}_{A1}=\left(\tilde{\boldsymbol{\mathcal{L}}}_1+\tilde{\boldsymbol{\mathcal{L}}}_1^\dagger\right)/2$ and $\tilde{\boldsymbol{\mathcal{L}}}_{A2}=\left(\tilde{\boldsymbol{\mathcal{L}}}_2+\tilde{\boldsymbol{\mathcal{L}}}_2^\dagger\right)/2$, we obtain
\[
\left(\mathbf{v}_\sigma,\tilde{\boldsymbol{\mathcal{L}}}_A\mathbf{v}_\sigma\right)=\left(\mathbf{v}_\sigma,\tilde{\boldsymbol{\mathcal{L}}}_{A1}\mathbf{v}_\sigma\right)+\left(\mathbf{v}_\sigma,\tilde{\boldsymbol{\mathcal{L}}}_{A2}\mathbf{v}_\sigma\right)<0,
\]
as $\left(\mathbf{v}_\sigma,\tilde{\boldsymbol{\mathcal{L}}}_{A1}\mathbf{v}_\sigma\right)\leq0$ and $\left(\mathbf{v}_\sigma,\tilde{\boldsymbol{\mathcal{L}}}_{A2}\mathbf{v}_\sigma\right)<0$. Thus all the eigenvalues of $\tilde{\mathcal{L}}$ have strictly negative real part.

Finally since the eigenvalues of $\mathcal{L}$ are just the ones of $\tilde{\mathcal{L}}$ plus one (this is obvious because the dimension of the subspace of traceless operators is just one less than the whole state dimension $N^2$), and this extra eigenvalue has to be equal to 0 because theorem \ref{theoexiststeady} asserts that there exists always at least one steady state. Therefore $\mathcal{L}$ fulfills the conditions of theorem \ref{theoremSteady} and the semigroup $e^{\mathcal{L}\tau}$ is relaxing.

\end{proof}

For checking if some semigroup satisfies the conditions of this theorem it is sometimes useful to explore if the operators $\{V_k,k\in I\}$ correspond to some irreducible representation of a group. Then we may invoke the Schur's lemma \cite{JansenBoon}.

\section{Microscopic description: Markovian case} \label{sectionMicrosMarkov}

As we have already pointed out, the dynamics of an open quantum system is notably different from that of a closed system; both from practical and fundamental points of view. Essential questions as the existence of universally valid dynamics or the continuity in time of these, become non-trivial matters for the evolution of open quantum systems.

Generally it would be necessary to describe the evolution in the extended closed system, and trace out the state at the end of the evolution. Unfortunately, many times, one of the systems is out of our control, and/or it has an infinite number of degrees of freedom, such that the task of calculating the associated Kraus operators can be difficult and impractical. That is the case for example of a finite $N$-dimensional system coupled with an infinite dimensional one; according to Eq. (\ref{Kraus}), we would obtain, in principle, infinite terms in the sum. However it can be easily shown that any completely positive map on a $N$-dimensional system can be written with at most $N^2$ Kraus operators (of course; since the dimension of $\mathfrak{B}^\ast$ is $N^2$, so the use of the infinite terms is unnecessary.

For that reason other techniques apart from the ``brute force'' strategy will become very useful and they will be the subject of this section.

FirstLY, we analyze some conditions under for a UDM to be approximately be approximately represented by a quantum Markov process. We will start from a microscopic Hamiltonian and apply some approximations leading to UDMs which fulfill the divisibility property (\ref{CompLaw}). Note that, in principle, as we have already said, the reduced evolution of a quantum system is never exactly Markovian. This is because, as discussed in section \ref{sectionTemporalContinuity}, there is always a privileged time, say $t_0$, where the both systems $A$ and $B$ start to interact being then the global state a product $\rho=\rho_A\otimes\rho_B$ and defining a UDM for $\rho_A$. Since generally the global state can no longer be factorized, and
(generally) the dynamical map from there is not a UDM, the divisibility property (\ref{CompLaw}) is violated.

\subsection{Nakajima-Zwanzig equation}\label{sectionNaka-Zwan}

Let assume that $A$ is our open system while $B$ plays the role of the environment; the Hamiltonian of the whole system is given by
\begin{equation}\label{totalH}
H=H_A+H_B+V,
\end{equation}
where $H_A$ and $H_B$ are the Hamiltonians which govern the local dynamics of $A$ and $B$ respectively, and $V$ stands for the interaction between both systems.

We will follow the approach initroduced by Nakajima \cite{Nakajima} and Zwanzig \cite{Zwanzig} using projection operators, which is also explained in \cite{Haake,Fain02,BrPe02} among others. In this method we define in the combined Hilbert space ``system plus environment'' $\hilbert=\hilbert_A\otimes\hilbert_B$, two orthogonal projection operators, $\mathcal{P}$ and $\mathcal{Q}$, $\mathcal{P}^2=\mathcal{P}$, $\mathcal{Q}^2=\mathcal{Q}$ and $\mathcal{P}\mathcal{Q}=\mathcal{Q}\mathcal{P}=0$, given by
\begin{eqnarray}
\mathcal{P}\rho&=&\tr_B(\rho)\otimes\rho_B,\\
\mathcal{Q}\rho&=&(\mathds{1}-\mathcal{P})\rho,
\end{eqnarray}
where $\rho\in\hilbert$ is the whole state of the system and environment and $\rho_B\in\hilbert_B$ is a fixed state of the environment. In fact, we choose $\rho_B$ to be the physical initial state of the environment. Strictly speaking this is not necessary but the method becomes much more complicated if one chooses another state. In addition, we assume that the system $B$ is initially in thermal equilibrium, this is
\[
\rho_B=\rho_\mathrm{th}=e^{-\beta H_B}[\tr\left(e^{-\beta H_B}\right)]^{-1},
\]
where $\beta=1/T$. Note that $\mathcal{P}\rho$ give all necessary information about the reduced system state $\rho_A$; so the knowledge of the dynamics of $\mathcal{P}\rho$ implies that one knows the time evolution of the reduced system.

If we start from the von Neumann equation (\ref{Liouville}) for the whole system
\begin{equation}\label{LiouvilleABV}
\frac{d\rho(t)}{dt}=-i[H_A+H_B+V,\rho(t)],
\end{equation}
and take projection operators we get
\begin{eqnarray}
\frac{d}{dt}\mathcal{P}\rho(t)&=&-i\mathcal{P}[H_A+H_B+V,\rho(t)],\label{SchrP} \\
\frac{d}{dt}\mathcal{Q}\rho(t)&=&-i\mathcal{Q}[H_A+H_B+V,\rho(t)]. \label{SchrQ}
\end{eqnarray}
As usual in perturbation theory we shall work in interaction picture with respect to the free Hamiltonians
\[
\tilde{\rho}(t)=e^{i(H_A+H_B)t}\rho(t)e^{-i(H_A+H_B)t}.
\]

Since
\begin{eqnarray*}
\tr_B[\tilde{\rho}(t)]&=&\tr_B\left[e^{i(H_A+H_B)t}\rho(t)e^{-i(H_A+H_B)t}\right]\\
&=&e^{iH_At}\tr_B\left[e^{iH_Bt}\rho(t)e^{-iH_Bt}\right]e^{iH_At}=e^{iH_At}\tr_B[\rho(t)]e^{iH_At}=\tilde{\rho}_A(t).
\end{eqnarray*}
the projection operators are preserved under this operation, and equations (\ref{SchrP}) and (\ref{SchrQ}) can be written in interaction picture just like
\begin{eqnarray}
\frac{d}{dt}\mathcal{P}\tilde{\rho}(t)&=&-i\mathcal{P}[\tilde{V}(t),\tilde{\rho}(t)], \label{InterP}\\
\frac{d}{dt}\mathcal{Q}\tilde{\rho}(t)&=&-i\mathcal{Q}[\tilde{V}(t),\tilde{\rho}(t)]. \label{InterQ}
\end{eqnarray}

Let us use the notation $\mathcal{V}(t)\cdot\equiv-i[\tilde{V}(t),\cdot]$; then by introducing the identity $\mathds{1}=\mathcal{P}+\mathcal{Q}$ between $\mathcal{V}(t)$ and $\tilde{\rho}(t)$ the equations (\ref{InterP}) and (\ref{InterQ}) may be rewritten as
\begin{eqnarray}
\frac{d}{dt}\mathcal{P}\tilde{\rho}(t)&=&\mathcal{P}\mathcal{V}(t)\mathcal{P}\tilde{\rho}(t)+\mathcal{P}\mathcal{V}(t)\mathcal{Q}\tilde{\rho}(t),\label{InterP2}\\
\frac{d}{dt}\mathcal{Q}\tilde{\rho}(t)&=&\mathcal{Q}\mathcal{V}(t)\mathcal{P}\tilde{\rho}(t)+\mathcal{Q}\mathcal{V}(t)\mathcal{Q}\tilde{\rho}(t).
\end{eqnarray}
The formal integration of the first of these equations gives
\begin{equation}\label{projectorPInt1}
\mathcal{P}\tilde{\rho}(t)=\mathcal{P}\rho(0)+\int_{0}^{t}ds\mathcal{P}\mathcal{V}(s)\mathcal{P}\tilde{\rho}(s)+\int_{0}^{t}ds\mathcal{P}\mathcal{V}(s)\mathcal{Q}\tilde{\rho}(s),
\end{equation}
where we have set $0$ as the origin of time without losing generality. On the other hand, the solution of the second equation can be written formally as
\begin{equation}\label{Qformalsolution}
\mathcal{Q}\tilde{\rho}(t)=\mathcal{G}(t,0)\mathcal{Q}\tilde{\rho}(0)+\int_{0}^tds\mathcal{G}(t,s)\mathcal{Q}\mathcal{V}(s)\mathcal{P}\tilde{\rho}(s).
\end{equation}
This is nothing but the operational version of variation of parameters formula for ordinary differential equations (see for example \cite{Chicone,Ince}), where the solution to the homogeneous equation,
\[
\frac{d}{dt}\mathcal{Q}\tilde{\rho}(t)=\mathcal{Q}\mathcal{V}(t)\mathcal{Q}\tilde{\rho}(t),
\]
is given by the propagator
\[
\mathcal{G}(t,s)=\mathcal{T}e^{\int_s^tdt'\mathcal{Q}\mathcal{V}(t')}.
\]
Inserting (\ref{Qformalsolution}) in (\ref{projectorPInt1}) yields
\begin{eqnarray}
\mathcal{P}\tilde{\rho}(t)&=&\mathcal{P}\rho(0)+\int_{0}^{t}ds\mathcal{P}\mathcal{V}(s)\mathcal{P}\tilde{\rho}(s)+\int_{0}^{t}ds\mathcal{P}\mathcal{V}(s)\mathcal{G}(s,0)\mathcal{Q}\tilde{\rho}(0) \nonumber \\
&+&\int_{0}^{t}ds\int_{0}^sdu\mathcal{P}\mathcal{V}(s)\mathcal{G}(s,u)\mathcal{Q}\mathcal{V}(u)\mathcal{P}\tilde{\rho}(u).\label{projectorPInt2}
\end{eqnarray}
This is the integrated version of the so-called \emph{generalized Nakajima-Zwanzig equation}
\begin{eqnarray}
\frac{d}{dt}\mathcal{P}\tilde{\rho}(t)&=&\mathcal{P}\mathcal{V}(t)\mathcal{P}\tilde{\rho}(t)+\mathcal{P}\mathcal{V}(t)\mathcal{G}(t,0)\mathcal{Q}\tilde{\rho}(0) \nonumber \\
&+&\int_{0}^tdu\mathcal{P}\mathcal{V}(t)\mathcal{G}(t,u)\mathcal{Q}\mathcal{V}(u)\mathcal{P}\tilde{\rho}(u).
\end{eqnarray}

Next, two assumptions are usually made. First one is that $\mathcal{P}\mathcal{V}(t)\mathcal{P}=0$, this means,
\[
\mathcal{P}\mathcal{V}(t)\mathcal{P}\rho=-i\tr_B[\tilde{V}(t),\rho_A\otimes\rho_B]\otimes\rho_B=-i\left[\tr_B\left(\tilde{V}(t)\rho_B\right),\rho_A\right]\otimes\rho_B=0,
\]
for every $\rho$ and then every $\rho_A$, which implies $\tr_B\left(\tilde{V}(t)\rho_B\right)=0$. If this is not fulfilled (which is not the case in the most practical situations), provided that $\rho_B$ commutes with $H_B$ one can always define a new interaction Hamiltonian with a shifted origin of the energy for $A$ \cite{Cohen92,ModernCohen08}. The change
\begin{equation}\label{ceroprimermomento}
V'=V-\tr_B[V\rho_B]\otimes\mathds{1}, \quad \text{and}\quad H_A'=H_A+\tr_B[V\rho_B]\otimes\mathds{1},
\end{equation}
makes the total Hamiltonian to be the same and
\begin{eqnarray*}
\tr_B\left[\tilde{V}'(t)\rho_B\right]&=&\tr_B\left[e^{i(H_A'+H_B)t}V'e^{-i(H_A'+H_B)t}\rho_B\right]\\
&=&e^{iH_A't}\tr_B[e^{iH_Bt}Ve^{-iH_Bt}\rho_B]e^{-iH_A't}\\
&-&e^{iH_A't}\tr_B[V\rho_B]e^{-iH_A't}\tr_B[e^{iH_Bt}\rho_Be^{-iH_Bt}]\\
&=&e^{iH_A't}\tr_B[V\rho_B]e^{-iH_A't}-e^{iH_A't}\tr_B[V\rho_B]e^{-iH_A't}=0,
\end{eqnarray*}
as we wished.

The second assumption consists in accepting that initially $\rho(0)=\rho_A(0)\otimes\rho_B(0)$, which is the necessary assumption to get a UDM. Of course some skepticism may arise thinking that if the system $B$ is out of our control, there is not guaranty that this assumption is fulfilled. Nonetheless the control on the system $A$ is enough to assure this condition, for that it will be enough to prepare a pure state in $A$, for instance by making a projective measurement as argued in \cite{Whitney08}, to assure, by theorem \ref{rhoApuro} that at the initial time the global system is in a product state.

These two assumptions make to vanish the second and the third term in equation (\ref{projectorPInt2}), and the integro-differential equation yields
\begin{equation}\label{Naka-Zwan}
\frac{d}{dt}\mathcal{P}\tilde{\rho}(t)=\int_{0}^tdu\mathcal{K}(t,u)\mathcal{P}\tilde{\rho}(u),
\end{equation}
where
\[
\mathcal{K}(t,u)=\mathcal{P}\mathcal{V}(t)\mathcal{G}(t,u)\mathcal{Q}\mathcal{V}(u).
\]
If this kernel is homogeneous $\mathcal{K}(t,u)=\mathcal{K}(t-u)$, which is the case for stationary states $\rho_B$ (as exemplified in next section), this equation can be formally solved by a Laplace transformation \cite{Daviesbook76}, but this usually turns out into an intractable problem in practice.

In order to transform this integro-differential equation in an Markovian master equation one would wish that the integral on the right hand side of (\ref{Naka-Zwan}) turns into a generator of the form of (\ref{diffMarkov}). For that aim, the first intuitive guess is to try that $\mathcal{K}(t,u)$ behaves as a delta function with respect to $\tilde{\rho}(u)$. For that to be true the typical variation time $\tau_A$ of $\tilde{\rho}(u)$ (which is only due to the interaction with $B$ because we have removed the rest of the dynamics by taking integration picture) has to be much larger than some time $\tau_B$, which characterizes the speed at which the kernel $\mathcal{K}(t,u)$ is decreasing when $|t-u|\gg1$. Of course this kind of approximations intrinsically involve some assumptions on the size of $B$ and the ``strength'' of the Hamiltonians, and there is in principle no guaranty that the resulting differential equation has the form of (\ref{diffMarkov}). However, several works \cite{Davies1,Davies2,GoKo76,FrGo76,FrNoVe76} studied carefully two limiting procedures where $\tau_A/\tau_B\rightarrow\infty$ so that this turns out to be the case. These two limits are

\begin{itemize}
\item \emph{The weak coupling limit}. Let us rewrite $V\rightarrow\alpha V$, where $\alpha$ accounts for the strength of the interaction. Since for $\alpha$ small the variation of $\tilde{\rho}(u)$ is going to be slow, in the limit $\alpha\rightarrow0$ one would get $\tau_A\rightarrow\infty$, and so $\tau_A/\tau_B\rightarrow\infty$ for every fixed $\tau_B$. Of course for $\alpha=0$ there is not interaction, and so we will need to rescale the time parameter appropriately before taking the limit. This is discussed in more detail in next section \ref{sectionweakcoupling}.
\item \emph{The singular coupling limit}. This corresponds to somehow the opposite case; here we get $\tau_B\rightarrow0$ by making $\mathcal{K}(t,u)$ to approach a delta function. We will discuss this limit in the forthcoming section \ref{sectionsingularcoupling}.
\end{itemize}

We should stress that one can construct Markovian models \textit{ad hoc} which try to describe a system from a microscopic picture without involving these procedures; see for example \cite{Ziman1,Ziman2}. In this sense, these two cases are sufficient conditions to get a Markovian evolution, but not necessary.

\subsection{Weak coupling limit}\label{sectionweakcoupling}
As we said above, if we assume the weak coupling limit, the change of $\tilde{\rho}(u)$ is negligible within the typical time $\tau_B$ where the kernel $\mathcal{K}(t,u)$ is varying. To perform this limit, consider again Eq. (\ref{projectorPInt2}) under the assumptions of $\mathcal{P}\mathcal{V}(t)\mathcal{P}=0$ and $\rho(0)=\rho_A(0)\otimes\rho_B(0)$:
\begin{equation}\label{weakcoupling1}
\mathcal{P}\tilde{\rho}(t)=\mathcal{P}\rho(0)+\alpha^2\int_{0}^{t}ds\int_{0}^sdu\mathcal{P}\mathcal{V}(s)\mathcal{G}(s,u)\mathcal{Q}\mathcal{V}(u)\mathcal{P}\tilde{\rho}(u),
\end{equation}
where we have redefined $V\rightarrow\alpha V$ (so $\mathcal{V}\rightarrow\alpha \mathcal{V}$). Since in the interaction picture the whole state satisfies a von Neumann equation of the form
\begin{equation}\label{Von-NeumannInt}
\frac{d}{dt}\tilde{\rho}(t)=-i\alpha[\tilde{V}(t),\tilde{\rho}(t)],
\end{equation}
we can connect the whole state at times $s$ and $u$ (note $u\leq s$) by the unitary solution
\begin{eqnarray*}
\tilde{\rho}(s)=\tilde{U}(s,u)\tilde{\rho}(u)\tilde{U}^\dagger(s,u)&=&\tilde{\mathcal{U}}(s,u)\tilde{\rho}(u)\\
&\Rightarrow& \tilde{\rho}(u)=\tilde{U}^\dagger(s,u)\tilde{\rho}(u)\tilde{U}(s,u)\equiv\tilde{\mathcal{U}}(u,s)\tilde{\rho}(s).
\end{eqnarray*}
By introducing this operator in the last term of (\ref{weakcoupling1}) we find
\[
\mathcal{P}\tilde{\rho}(t)=\mathcal{P}\rho(0)+\alpha^2\int_{0}^{t}ds\int_{0}^sdu\mathcal{P}\mathcal{V}(s)\mathcal{G}(s,u)\mathcal{Q}\mathcal{V}(u)\mathcal{P}\tilde{\mathcal{U}}(u,s)\tilde{\rho}(s).
\]
Now the term inside the double integral admits an expansion in powers of $\alpha$ given by the time-ordered series of
\begin{eqnarray}
\mathcal{G}(s,u)&=&\mathcal{T}e^{\alpha\int_s^tdt'\mathcal{Q}\mathcal{V}(t')}\quad \text{and}, \nonumber\\
\tilde{\mathcal{U}}(u,s)&=&\left[\tilde{\mathcal{U}}(s,u)\right]^\dagger=\left[\mathcal{T}e^{\alpha\int_u^sdt'\mathcal{V}(t')}\right]^\dagger=\mathcal{T}^{\star}e^{\alpha\int_u^sdt'\mathcal{V}^\dagger(t')} \label{backwardU},
\end{eqnarray}
where the adjoint is taken over every element of the expansion and $\mathcal{T}^{\star}$ denotes the antichronological time-ordering operator. This is defined similarly as $\mathcal{T}$ in (\ref{timeordering}) but ordering in the opposite sense, this is, the operators evaluated at later times appear first in the product. Thus, at lowest order we get
\begin{eqnarray*}
\mathcal{P}\tilde{\rho}(t)&=&\mathcal{P}\rho(0)+\alpha^2\int_{0}^{t}ds\int_{0}^sdu\mathcal{P}\mathcal{V}(s)\mathcal{Q}\mathcal{V}(u)\mathcal{P}\tilde{\rho}(s)+\mathcal{O}(\alpha^3)\\
&=&\mathcal{P}\rho(0)+\alpha^2\int_{0}^{t}ds\int_{0}^sdu\mathcal{P}\mathcal{V}(s)\mathcal{V}(u)\mathcal{P}\tilde{\rho}(s)+\mathcal{O}(\alpha^3).
\end{eqnarray*}
where we have used again that $\mathcal{P}\mathcal{V}(t)\mathcal{P}=0$. After a change of variable $u\rightarrow s-u$, writing explicitly the projector $\mathcal{P}$ and the operators $\mathcal{V}$, we conclude that
\begin{equation}\label{weakcoupling2}
\tilde{\rho}_A(t)=\rho_A(0)-\alpha^2\int_{0}^{t}ds\int_{0}^sdu\tr_B\left[\tilde{V}(s),[\tilde{V}(s-u),\tilde{\rho}_A(s)\otimes\rho_B]\right]+\mathcal{O}(\alpha^3).
\end{equation}

\subsubsection{Decomposition of $V$} \label{sectionDecompositionV}
In order to continue with the derivation, we write the interaction Hamiltonian as
\begin{equation}\label{Vdesc}
V=\sum_{k}A_{k}\otimes B_{k},
\end{equation}
with $A_{k}^\dagger=A_{k}$ and $B_{k}^\dagger=B_{k}$. Note that since $V^\dagger=V$, one can always chose $A_{k}$ and $B_{k}$ to be self-adjoint. Indeed, assume that $V$ is a sum of products of any general operators $X_{k}$ and $Y_{k}$, it will be written then like
\begin{eqnarray}
V=\sum_{k=1}^nX_{k}^\dagger\otimes Y_{k}+X_{k}\otimes Y_{k}^\dagger. \label{VdescNoH}
\end{eqnarray}
Any operator can be decomposed as
\[
X_{k}=X^{(a)}_{k}+iX^{(b)}_{k},
\]
where
\[
X^{(a)}_{k}=\frac{(X_{k}^\dagger+X_{k})}{2},\quad X^{(b)}_{k}=\frac{i(X_{k}^\dagger-X_{k})}{2}
\]
are self-adjoint. This, and the same kind of decomposition for $Y_{k}=Y^{(a)}_{k}+iY^{(b)}_{k}$ inserted in equation (\ref{VdescNoH}) gives
\begin{eqnarray*}
V&=&\sum_{k=1}^n\left(X^{(a)}_{k}-iX^{(b)}_{k}\right)\otimes \left(Y^{(a)}_{k}+iY^{(b)}_{k}\right)\\
&+&\left(X^{(a)}_{k}+iX^{(b)}_{k}\right)\otimes \left(Y^{(a)}_{k}-iY^{(b)}_{k}\right)\\
&=&2\sum_{k=1}^nX^{(a)}_k\otimes Y^{(a)}_{k}+X^{(b)}_k\otimes Y^{(b)}_{k},
\end{eqnarray*}
so one can make the choice
\begin{eqnarray*}
A_{k}&=&2X^{(a)}_k, \quad B_{k}=2Y^{(a)}_k, \quad \text{for }k=1,\ldots,n\\
A_{k}&=&2X^{(b)}_{k-n}, \quad B_{k}=2Y^{(b)}_{k-n}, \quad \text{for }k=n+1,\dots,2n
\end{eqnarray*}
to arrive at (\ref{Vdesc}).

Now we focus on the operators $A_{k}$, let us assume for the moment that the spectrum of $H_A$ is discrete, and let $|\psi _\varepsilon\rangle$ be an eigenstate associate with the eigenvalue $\varepsilon$. We define
\[
A_{k}(\omega)=\sum_{\varepsilon'-\varepsilon=\omega}|\psi_\varepsilon\rangle\langle\psi_\varepsilon|A_{k}|\psi_{\varepsilon'}\rangle\langle\psi_{\varepsilon'}|,
\]
where the sum is over every $\varepsilon'$ and $\varepsilon$ such that their difference is $\omega$. The operator so defined is an eigenoperator of $\mathcal{L}^{\mathcal{A}}=i[H_A,\cdot]$ with eigenvalue $-i\omega$; indeed,
\[
\mathcal{L}^{\mathcal{A}}A_{k}(\omega)=i[H_A,A_{k}(\omega)]=i\sum_{\varepsilon'-\varepsilon=\omega}(\varepsilon-\varepsilon')|\psi_\varepsilon\rangle\langle\psi_\varepsilon|A_{k}|\psi_{\varepsilon'}\rangle\langle\psi_{\varepsilon'}|=-i\omega A_{k}(\omega).
\]
As a result,
\[
\mathcal{L}^{\mathcal{A}}A_{k}^\dagger(\omega)=i\omega A_{k}^\dagger(\omega),
\]
with
\[
A_{k}^\dagger(\omega)=A_{k}(-\omega),
\]
and this holds for every $k$. In addition, it is straightforward to verify that
\begin{equation}\label{commutinglamd}
[H_A,A_{k}^\dagger(\omega) A_\ell(\omega)]=0.
\end{equation}

On the other hand, to write $A_{k}(\omega)$ and $A_{k}^\dagger(\omega)$ in the interaction picture with respect to $H_A$ is quite easy:
\begin{eqnarray*}
e^{\mathcal{L}^\mathcal{A} t}A_{k}(\omega)&=&e^{iH_A t}A_{k}(\omega)e^{-iH_A t}=e^{-i\omega t}A_{k}(\omega),\\
e^{\mathcal{L}^\mathcal{A} t}A_{k}^\dagger(\omega)&=&e^{iH_A t}A_{k}^\dagger(\omega)e^{-iH_A t}=e^{i\omega t}A_{k}^\dagger(\omega).
\end{eqnarray*}
Moreover, note that
\[
\sum_\omega A_k(\omega)=\sum_{\varepsilon',\varepsilon}|\psi_\varepsilon\rangle\langle\psi_\varepsilon|A_{k}|\psi_{\varepsilon'}\rangle\langle\psi_{\varepsilon'}|=A_k,
\]
and similarly that
\[
\sum_\omega A_k^\dagger(\omega)=A_k.
\]
Thus, by introducing these two decompositions for $A_k$ in (\ref{Vdesc}) and taking the interaction picture, one gets
\begin{equation}\label{Veigendesc}
\tilde{V}(t)=\sum_{\omega,k} e^{-i\omega t}A_k(\omega)\otimes \tilde{B}_k(t)=\sum_{\omega,k} e^{i\omega t}A_k^\dagger(\omega)\otimes \tilde{B}^\dagger_k(t).
\end{equation}

Now, coming back to the equation (\ref{weakcoupling2}) the use of the above first equality with $\tilde{V}(s-u)$ and the second one with $\tilde{V}(s)$ yields
\begin{eqnarray}
\tilde{\rho}_A(t)=\rho_A(0)&+&\alpha^2\int_{0}^{t}ds
\sum_{\omega,\omega'}\sum_{k,\ell}e^{i(\omega'-\omega)s}\Gamma_{k\ell}^s(\omega)[A_\ell(\omega)\tilde{\rho}_A(s),A_k^\dagger(\omega')]\nonumber\\
&+&e^{i(\omega-\omega')s}\Gamma_{\ell k}^{s\ast}(\omega)[A_\ell(\omega'),\tilde{\rho}_A(s)A_k^\dagger(\omega)]+\mathcal{O}(\alpha^3).
\end{eqnarray}
Here we have introduced the quantities
\begin{eqnarray}\label{leftFourier}
\Gamma_{k\ell}^s(\omega)&=&\int_0^s du e^{i\omega u}\mathrm{Tr}\left[\tilde{B}_k(s)\tilde{B}_\ell(s-u)\rho_B\right]\nonumber\\
&=&\int_0^s du e^{i\omega u}\mathrm{Tr}\left[\tilde{B}_k(u) B_\ell\rho_B\right],
\end{eqnarray}
where the last step is justified because $\rho_B$ commutes with $e^{iH_Bt}$. Next, we take a rescaled time $\tau=\alpha^2 t$ and $\sigma=\alpha^2 s$, obtaining
\begin{eqnarray*}
\tilde{\rho}_A(t)&=&\tilde{\rho}_A(\tau/\alpha^2)\\
&=&\rho_A(0)+\int_{0}^{\tau}d\sigma
\sum_{\omega,\omega'}\sum_{k,\ell}e^{i(\omega'-\omega)\sigma/\alpha^2}\Gamma_{k\ell}^{\sigma/\alpha^2}(\omega)[A_\ell(\omega)\tilde{\rho}_A(\sigma/\alpha^2),A_k^\dagger(\omega')]\\
&+&e^{i(\omega-\omega')\sigma/\alpha^2}\Gamma_{\ell k}^{\sigma/\alpha^2\ast}(\omega)[A_\ell(\omega'),\tilde{\rho}_A(\sigma/\alpha^2)A_k^\dagger(\omega)]+\mathcal{O}(\alpha).
\end{eqnarray*}
From this equation we see that in the limit $\alpha\rightarrow0$ (keeping finite $\tau$ and $\sigma$) the time scale of change of $\tilde{\rho}_A(\tau/\alpha^2)$ is just $\tau$. This is merely because the dependency with $\alpha^2$ of $\tilde{\rho}_A(\tau/\alpha^2)$ enters only in the free evolution which we remove by taking the interaction picture. The remaining time scale due to the interaction with the environment is $\tau$, as expressed in the above equation (the integral runs from 0 to $\tau$), so we write
\[
\lim_{\alpha\rightarrow0}\tilde{\rho}_A(t)=\lim_{\alpha\rightarrow0}e^{iH_A\tau/\alpha^2}\rho_A(\tau/\alpha^2)e^{-iH_A\tau/\alpha^2}\equiv\tilde{\rho}_A(\tau),
\]
and thus
\begin{eqnarray}
\tilde{\rho}_A(\tau)&=&\rho_A(0)+\lim_{\alpha\rightarrow0}\int_{0}^{\tau}d\sigma
\sum_{\omega,\omega'}\sum_{k,\ell}e^{i(\omega'-\omega)\sigma/\alpha^2}\Gamma_{k\ell}^{\sigma/\alpha^2}(\omega)[A_\ell(\omega)\tilde{\rho}_A(\sigma/\alpha^2),A_k^\dagger(\omega')]\nonumber\\
&+&e^{i(\omega-\omega')\sigma/\alpha^2}\Gamma_{\ell k}^{\sigma/\alpha^2\ast}(\omega)[A_\ell(\omega'),\tilde{\rho}_A(\sigma/\alpha^2)A_k^\dagger(\omega)]+\mathcal{O}(\alpha)\nonumber\\
&=&\rho_A(0)+\lim_{\alpha\rightarrow0}\int_{0}^{\tau}d\sigma
\sum_{\omega,\omega'}\sum_{k,\ell}e^{i(\omega'-\omega)\sigma/\alpha^2}\Gamma_{k\ell}^{\sigma/\alpha^2}(\omega)[A_\ell(\omega)\tilde{\rho}_A(\sigma),A_k^\dagger(\omega')]\nonumber\\
&+&e^{i(\omega-\omega')\sigma/\alpha^2}\Gamma_{\ell k}^{\sigma/\alpha^2\ast}(\omega)[A_\ell(\omega'),\tilde{\rho}_A(\sigma)A_k^\dagger(\omega)]. \label{weakcoupling2.1}
\end{eqnarray}

Before going further in the derivation, a mathematical result which will be particularly useful from now on is the following.
\begin{proposition}[Riemann-Lebesgue lemma]\label{propAlfredo} Let $f(t)$ be integrable in $[a,b]$ then
\[
\lim_{x\rightarrow\infty}\int_a^b e^{ix t}f(t)dt=0.
\]
\end{proposition}
\begin{proof} There are several methods to prove this, here we define $g(t)=f(t)[\theta(t-a)-\theta(t-b)]$ and so
\[
\int_a^b e^{ix t}f(t)dt=\int_{-\infty}^{\infty} e^{ix t}g(t)dt=\hat{g}(x),
\]
where $\hat{g}(x)$ denotes the Fourier transform of $g(t)$. Now it is evident that $g(t)$ is square integrable, then as the Fourier transform preserve norms \cite{reedsimon2} (this is sometimes referred to as Parseval's theorem):
\[
\int_{-\infty}^{\infty} |g(t)|^2dt=\int_{-\infty}^{\infty} |\hat{g}(x)|^2dx,
\]
which implies that $\hat{g}(x)$ is also square integrable and then $\lim_{x\rightarrow\infty}\hat{g}(x)=0$.

\end{proof}

Because of this result (applied on the Banach space $\mathfrak{B}$, to be concise) in the limit $\alpha\rightarrow0$ the oscillatory factors in (\ref{weakcoupling2.1}) with $\omega'\neq\omega$ are going to vanish. This so-called \emph{secular approximation} is done in the same spirit that well known rotating wave approximation in the context of quantum optics \cite{Fain02,BrPe02,GardinerZoller04,Benatti05,ModernCohen08}. We obtain
\begin{eqnarray}
\tilde{\rho}_A(\tau)=\rho_A(0)&+&\int_{0}^{\tau}d\sigma
\sum_{\omega}\sum_{k,\ell}\Gamma_{k\ell}^{\infty}(\omega)[A_\ell(\omega)\tilde{\rho}_A(\sigma),A_k^\dagger(\omega)]\nonumber\\
&+&\Gamma_{\ell k}^{\infty\ast}(\omega)[A_\ell(\omega),\tilde{\rho}_A(\sigma)A_k^\dagger(\omega)]\nonumber\\
&\equiv&\rho_A(0)+\int_{0}^{\tau}d\sigma\tilde{\mathcal{L}}\left[\tilde{\rho}_A(\sigma)\right].\label{weakcoupling3}
\end{eqnarray}
Here the coefficients are given by the one-sided Fourier transform
\[
\Gamma_{k\ell}^{\infty}(\omega)=\int_0^\infty du e^{i\omega u}\mathrm{Tr}\left[\tilde{B}_k(u) B_\ell\rho_B\right].
\]
\subsubsection{Correlation functions}\label{sectionCorrFunc}
In order to make convergent the above integrals, the environmental correlation functions $\mathrm{Tr}\left[\tilde{B}_k(u) B_\ell\rho_B\right]$ have to decrease as $u$ is increasing. However, if we use the corresponding eigenoperator decomposition of $B_k$ for the Hamiltonian $H_B$, we find
\[
\mathrm{Tr}\left[\tilde{B}_k(u) B_\ell\rho_B\right]=\sum_{\omega} e^{-i\omega u} \mathrm{Tr}\left[B_k(\omega) B_\ell\rho_B\right],
\]
which shows that the correlation functions are periodic in $u$ and their integral extended to infinity will diverge. The only possible loophole to break the periodicity is the assumption that the system $B$ has infinite degrees of freedom, in such a way that $H_B$ has a continuous spectrum. Then everything is essentially the same but the sums are substituted by integrals, i.e. the decomposition in eigenoperators will be
\begin{equation}\label{Beigendecomp}
B_k=\int_{-a}^{a}d\omega B_k(\omega),
\end{equation}
where $a$ is the maximum eigenfrequency (it can be infinite). So the correlation functions become
\[
\mathrm{Tr}\left[\tilde{B}_k(u) B_\ell\rho_B\right]=\int_{-a}^{a}d\omega e^{-i\omega u} \mathrm{Tr}\left[B_k(\omega) B_\ell\rho_B\right],
\]
and of course now they are not periodic as
\[
\lim_{u\rightarrow\infty}\mathrm{Tr}\left[\tilde{B}_k(u) B_\ell\rho_B\right]=0,
\]
because of the proposition (\ref{propAlfredo}).

\subsubsection{Davies' theorem}
We have just seen that a necessary condition for the weak limit to exist is that the environmental system $B$ has infinite degrees of freedom. However this is not a sufficient condition because the decreasing behavior of the correlation functions may be not fast enough for the one-sided Fourier transform to exist. A sufficient condition was pointed out by Davies.
\begin{theorem}[Davies] If there exist some $\epsilon>0$ such that
\begin{equation}\label{Daviescondition}
\int_0^\infty dt \left|\mathrm{Tr}\left[\tilde{B}_k(t) B_\ell\rho_B\right]\right|(1-t)^\epsilon<\infty,
\end{equation}
then the weak coupling is strictly well defined (for a bounded interaction $V$) and
\[
\lim_{\alpha\rightarrow0, \tau=\alpha^2t}\|\tilde{\mathcal{E}}_{(\tau,0)}\left[\rho_A(0)\right]-e^{\tilde{\mathcal{L}}\tau}\left[\rho_A(0)\right]\|_1=0
\]
for all $\rho_A(0)$ and uniformly in every finite interval $[0,\tau_0]$, where $\tilde{\mathcal{E}}_{(\tau,0)}$ denotes the exact reduced dynamics in the interaction picture.
\end{theorem}
\begin{proof} To prove the condition (\ref{Daviescondition}) is beyond the scope of this work, it was developed in \cite{Davies1,Davies2}. Under (\ref{Daviescondition}) the convergency to the above result is clear from equation (\ref{weakcoupling3}) which is the integrated version of
\[
\frac{d\tilde{\rho}_A(\tau)}{d\tau}=\tilde{\mathcal{L}}\tilde{\rho}_A(\tau)\Rightarrow \tilde{\rho}_A(\tau)= e^{\tilde{\mathcal{L}}\tau}\left[\tilde{\rho}_A(0)\right].
\]
\end{proof}

On one hand, note that the method works because there are no eigenfrequencies, $\omega$, arbitrary close to one another, as a result of the discreteness of the spectrum of $H_A$, otherwise the secular approximation is jeopardized. For the infinite dimensional case where this is not satisfied, recently D. Taj and F. Rossi \cite{Taj} have proposed another approximation method which substitutes the secular approximation and also leads to a completely positive semigroup.

On the other hand, an important consequence of the secular approximation is that it guaranties that the generator $\tilde{\mathcal{L}}$ has the correct form of a Markovian homogenous master equation (theorem \ref{TeoremaMARKOV}). To see this, let us decompose the matrices $\Gamma_{k\ell}^{\infty}(\omega)$ in sum of Hermitian and anti-Hermitian parts
\begin{equation}\label{Gammadescomposition}
\Gamma_{k\ell}^{\infty}(\omega)=\frac{1}{2}\gamma_{k\ell}(\omega)+iS_{k\ell}(\omega),
\end{equation}
where the coefficients
\[
S_{k\ell}(\omega)=\frac{1}{2i}[\Gamma_{k\ell}^{\infty}(\omega)-\Gamma_{\ell k}^{\infty\ast}(\omega)],
\]
and
\begin{eqnarray}
\gamma_{k \ell}(\omega)&=&\Gamma_{k \ell}^{\infty}(\omega)+\Gamma_{\ell k}^{\infty\ast}(\omega)\nonumber \\
&=&\int_{0}^{\infty}due^{i\omega u}\mathrm{Tr}\left[\tilde{B}_k(u) B_\ell\rho_\mathrm{th}\right]+\int_{0}^{\infty}due^{-i\omega u}\mathrm{Tr}\left[\tilde{B}_k(-u) B_\ell\rho_\mathrm{th}\right]\nonumber \\
&=&\int_{-\infty}^{\infty}due^{i\omega u}\mathrm{Tr}\left[\tilde{B}_k(u) B_\ell\rho_\mathrm{th}\right]\label{gammapequenha},
\end{eqnarray}
form Hermitian matrices. In terms of these quantities, the generator $\tilde{\mathcal{L}}$ can be written as
\begin{eqnarray*}
\tilde{\mathcal{L}}\left[\tilde{\rho}_A(\tau)\right]=&-&i[H_{\mathrm{LS}},\tilde{\rho}_A(\tau)]\\
&+&\sum_{\omega}\sum_{k,\ell}\gamma_{k\ell}(\omega)\left[ A_\ell(\omega)\tilde{\rho}_A(\tau)A_k^\dagger(\omega)-\frac{1}{2}\{A_k^\dagger(\omega)A_\ell(\omega),\tilde{\rho}_A(\tau)\}\right],
\end{eqnarray*}
where the Hamiltonian part is given by
\begin{equation}\label{Hlambshift}
H_{\mathrm{LS}}=\sum_\omega\sum_{k,\ell}S_{k\ell}(\omega)A_k^\dagger(\omega)A_\ell(\omega).
\end{equation}
In order to prove that $\tilde{\mathcal{L}}$ has the form of a generator of a Markovian homogenous evolution, it is required just to show that the matrix $\gamma_{k\ell}(\omega)$ is positive semidefinite for all $\omega$. This is a consequence of the Bochner's theorem \cite{reedsimon1,reedsimon2}, which asserts that the Fourier transform of a function of ``positive type'' is a positive quantity. Note that a function $f(t)$ is of positive type if for any $t_m$ and $t_n$ the matrix constructed as $f_{mn}=f(t_m-t_n)$ is positive semidefinite.

For any vector $\mathbf{v}$ we have
\begin{eqnarray*}
(\mathbf{v},\boldsymbol{\gamma}\mathbf{v})&=&\sum_{k,\ell}v_k^\ast\gamma_{k\ell}(\omega)v_\ell=\int_{-\infty}^{\infty}due^{i\omega u}\sum_{k,\ell}v_k^\ast\mathrm{Tr}\left[\tilde{B}_k(u) B_\ell\rho_\mathrm{th}\right]v_\ell\\
&=&\int_{-\infty}^{\infty}due^{i\omega u}\sum_{k,\ell}v_k^\ast\mathrm{Tr}\left[e^{iH_Bu}B_k e^{-iH_Bu} B_\ell\rho_\mathrm{th}\right]v_\ell\\
&=&\int_{-\infty}^{\infty}due^{i\omega u}\mathrm{Tr}\left[e^{iH_Bu}C^\dagger e^{-iH_Bu} C\rho_\mathrm{th}\right],
\end{eqnarray*}
where $C=\sum_k B_k v_k$.

\begin{proposition} The function $f(t)=\mathrm{Tr}\left[e^{iH_Bt}C^\dagger e^{-iH_Bt} C\rho_\mathrm{th}\right]$ is of ``positive type'' and then the matrix $\gamma_{k\ell}$ is positive semidefinite $(\mathbf{v},\boldsymbol{\gamma}\mathbf{v})\geq0$.
\end{proposition}
\begin{proof} The positivity of $(\mathbf{v},\boldsymbol{\gamma}\mathbf{v})\geq0$ follows from the Bochner's theorem if $f(t)$ is of ``positive type'' because it is its Fourier transform. For the proof of the first statement, let us take the trace in the eigenbasis of $H_B$ (remember that the spectrum is taken to be continuous):
\[
f(t)=\int d\varepsilon \langle \varphi_\varepsilon| e^{iH_Bt}C^\dagger e^{-iH_Bt} C\rho_\mathrm{th}|\varphi_\varepsilon\rangle
=\int d\varepsilon e^{i\varepsilon t} p_{\varepsilon} \langle \varphi_\varepsilon| C^\dagger e^{-iH_Bt} C|\varphi_\varepsilon\rangle,
\]
here we have used that
\[
\rho_\mathrm{th}|\varphi_\varepsilon\rangle=\frac{e^{-\beta\varepsilon}}{\tr\left[e^{(-\beta H_B)}\right]}|\varphi_\varepsilon\rangle\equiv p_\varepsilon|\varphi_\varepsilon\rangle,
\]
with $p_\varepsilon\geq0$ since $\rho_\mathrm{th}$ is a positive operator. By introducing the identity $\mathds{1}=\int d\varepsilon' |\varphi_{\varepsilon'}\rangle\langle \varphi_{\varepsilon'}|$,
\begin{eqnarray*}
f(t)&=&\int d\varepsilon d\varepsilon' e^{i\varepsilon t} p_{\varepsilon} \langle \varphi_\varepsilon| C^\dagger e^{-iH_Bt}|\varphi_{\varepsilon'}\rangle\langle \varphi_{\varepsilon'}| C|\varphi_\varepsilon\rangle\\
&=&\int d\varepsilon d\varepsilon' e^{i(\varepsilon-\varepsilon') t} p_{\varepsilon} \left|\langle\varphi_{\varepsilon'}| C|\varphi_\varepsilon\rangle\right|^2.
\end{eqnarray*}
Finally we take another arbitrary vector $\mathbf{w}$ and the inner product
\begin{eqnarray*}
(\mathbf{w},\mathbf{f}\mathbf{w})&=&\sum_{m,n}w_m^\ast f(t_m-t_n)w_n\\
&=&\int d\varepsilon d\varepsilon' \sum_{m,n}w_m^\ast e^{i(\varepsilon-\varepsilon') (t_m-t_n)}w_n p_{\varepsilon} \left|\langle\varphi_{\varepsilon'}| C|\varphi_\varepsilon\rangle\right|^2\\
&=&\int d\varepsilon d\varepsilon' \left|\sum_{m}w_m^\ast e^{i(\varepsilon-\varepsilon') t_m}\right|^2 p_{\varepsilon} \left|\langle\varphi_{\varepsilon'}| C|\varphi_\varepsilon\rangle\right|^2\geq0,
\end{eqnarray*}
because the integrand is positive on the whole domain of integration.

\end{proof}

As a conclusion of these results, if for a small, but finite $\alpha$, we substitute the exact reduced dynamics $\tilde{\mathcal{E}}_{(t,0)}$ by the Markovian one $e^{\tilde{\mathcal{L}}\tau}=e^{\alpha^2\tilde{\mathcal{L}}t}$, the error which we make is bounded and tends to zero as $\alpha$ decrease provided that the assumptions of the Davies' theorem are fulfilled.

Of course one can write $\alpha^2\tilde{\mathcal{L}}$ and immediately go back to Schr\"odinger picture by using the unitary operator $e^{iH_At}$ to obtain
\begin{eqnarray}\label{weakcoupling4}
\frac{d\rho_A(t)}{dt}&=&-i[H_A+\alpha^2H_{\mathrm{LS}},\rho_A(t)] \\
&+&\alpha^2\sum_{\omega}\sum_{k,\ell}\gamma_{k\ell}(\omega)\left[ A_\ell(\omega)\rho_A(t)A_k^\dagger(\omega)-\frac{1}{2}\{A_k^\dagger(\omega)A_\ell(\omega),\rho_A(t)\}\right]\nonumber.
\end{eqnarray}
Note that because of (\ref{commutinglamd}), the Hamiltonian $H_{\mathrm{LS}}$ commutes with $H_A$, so it just produces a shift in the energy levels of the system $A$. Thus this is sometimes called Lamb shift Hamiltonian. However, note that $H_{\mathrm{LS}}$ is not only influenced by an environment in the vacuum state (zero temperature) but also by thermal fluctuations.

\subsubsection{Explicit expressions for $\gamma_{k\ell}$ and $S_{k\ell}$} \label{sectionexplictgamma}
We can give explicit expressions for $\gamma_{k\ell}$ and $S_{k\ell}$ by using the decomposition of the operators $B_k$ in eigenoperators. For the decay rates
\begin{eqnarray}\label{explicitgamma}
\gamma_{k\ell}(\omega)&=&\int_{-\infty}^{\infty}due^{i\omega u}\mathrm{Tr}\left[\tilde{B}_k(u) B_\ell\rho_\mathrm{th}\right]\nonumber \\
&=&\int_{-a}^a d\omega'\int_{-\infty}^{\infty}due^{i(\omega-\omega') u}\mathrm{Tr}\left[B_k(\omega') B_\ell\rho_\mathrm{th}\right]\nonumber \\
&=&2\pi\int_{-a}^a d\omega'\delta(\omega-\omega')\mathrm{Tr}\left[B_k(\omega') B_\ell\rho_\mathrm{th}\right]\nonumber \\
&=&2\pi\mathrm{Tr}\left[B_k(\omega) B_\ell\rho_\mathrm{th}\right],
\end{eqnarray}
provided that $\omega\in(-a,a)$. So $\gamma_{k\ell}(\omega)$ is just proportional to the correlation function of the eigenoperator $B_k(\omega)$ with the same frequency $\omega$, unless this frequency is outside of the spectrum of $[H_B,\cdot]$, then $\gamma_{k\ell}(\omega)$ will be zero.

To get an explicit expression for the shifts $S_{k\ell}$ is a little bit more complicated. From (\ref{Gammadescomposition}) we have
\[
S_{k\ell}(\omega)=-i\Gamma_{k\ell}^{\infty}(\omega)+\frac{i}{2}\gamma_{k \ell}(\omega).
\]
The substitution of the decomposition in eigenoperators in $\Gamma_{k,\ell}^{\infty}(\omega)$ and the previous result (\ref{explicitgamma}) for $\gamma_{k,\ell}(\omega)$ yields
\begin{equation}
S_{k\ell}(\omega)=-i\int_{-a}^a d\omega' \int_0^\infty du e^{i(\omega-\omega') u}\mathrm{Tr}\left[B_k(\omega') B_\ell\rho_B\right]+\pi i\mathrm{Tr}\left[B_k(\omega) B_\ell\rho_\mathrm{th}\right]. \label{PrincipalValue1}
\end{equation}
Note that the integral $\int_0^\infty du e^{i(\omega-\omega') u}$ in not well-defined. Similarly it happens with $\int_{-\infty}^\infty du e^{i(\omega-\omega') u}$, but we know this integral is proportional to the delta function. The mathematical theory which provides a meaningful interpretation to these integrals is the theory of distributions or generalized functions. Here we will not dwell on it, references \cite{reedsimon1,Boccara} provide introductions and further results can be found in \cite{Schwartz,Vladimirov}. For us it is enough to show that these integrals can have a meaning if we define them as a limit of ordinary functions. For instance if we define
\[
\int_{-\infty}^\infty du e^{i(\omega-\omega') u}=\lim_{\epsilon\rightarrow 0^+} \int_{0}^\infty du e^{i(\omega-\omega') u-\epsilon u}+\lim_{\epsilon\rightarrow 0^-} \int_{-\infty}^0 du e^{i(\omega-\omega') u+\epsilon u},
\]
by making the elementary integrals,
\begin{eqnarray*}
\int_{-\infty}^\infty du e^{i(\omega-\omega') u}&=&\lim_{\epsilon\rightarrow 0^+}\frac{i}{(\omega-\omega')+\epsilon i}+\lim_{\epsilon\rightarrow 0^-} \frac{i}{\epsilon i-(\omega-\omega')}\\
&=&\lim_{\epsilon\rightarrow 0}\frac{i}{(\omega-\omega')+\epsilon i}+\frac{i}{\epsilon i-(\omega-\omega')}\\
&=&\lim_{\epsilon\rightarrow 0}\frac{2 \epsilon}{\epsilon^2 + (\omega-\omega')^2}\equiv2\pi \delta(\omega-\omega').
\end{eqnarray*}
We identify the well-known family of functions $f_\epsilon(\omega-\omega')=\frac{\epsilon}{\pi[\epsilon^2 + (\omega-\omega')^2]}$ which tend to the delta function as $\epsilon\rightarrow0$ \cite{Arfken} because for any (regular enough) integrable function $g(\omega)$
\[
\lim_{\epsilon\rightarrow 0}\int_a^b d\omega' f_\epsilon(\omega-\omega')g(\omega')=g(\omega), \quad (a<\omega<b).
\]
So, if we do the same procedure but just with the one-sided Fourier transform $\int_0^\infty du e^{i(\omega-\omega') u}$, we obtain in equation (\ref{PrincipalValue1})
\[
S_{k\ell}(\omega)=\lim_{\epsilon\rightarrow 0^+}\int_{-a}^a d\omega' \frac{1}{(\omega-\omega')+\epsilon i} \mathrm{Tr}\left[B_k(\omega') B_\ell\rho_B\right]+\pi i\mathrm{Tr}\left[B_k(\omega) B_\ell\rho_\mathrm{th}\right].
\]
Multiplying the integral by the complex conjugate of the denominator we obtain
\begin{eqnarray*}
\lim_{\epsilon\rightarrow 0^+}\int_{-a}^a d\omega' \frac{\mathrm{Tr}\left[B_k(\omega') B_\ell\rho_B\right]}{(\omega-\omega')+\epsilon i} &=&\lim_{\epsilon\rightarrow 0^+}\int_{-a}^a d\omega' \tfrac{(\omega-\omega')}{(\omega-\omega')^2+\epsilon^2} \mathrm{Tr}\left[B_k(\omega') B_\ell\rho_B\right]\\
&-&\lim_{\epsilon\rightarrow 0^+}i\int_{-a}^a d\omega' \tfrac{\epsilon}{(\omega-\omega')^2+\epsilon^2} \mathrm{Tr}\left[B_k(\omega') B_\ell\rho_B\right]\\
&=&\lim_{\epsilon\rightarrow 0^+}\int_{-a}^a d\omega' \tfrac{(\omega-\omega')}{(\omega-\omega')^2+\epsilon^2} \mathrm{Tr}\left[B_k(\omega') B_\ell\rho_B\right]\\
&-&i\pi\mathrm{Tr}\left[B_k(\omega) B_\ell\rho_B\right].
\end{eqnarray*}
Therefore the last term is canceled in the shifts, which read
\[
S_{k\ell}(\omega)=\lim_{\epsilon\rightarrow 0^+}\int_{-a}^a d\omega' \tfrac{(\omega-\omega')}{(\omega-\omega')^2+\epsilon^2} \mathrm{Tr}\left[B_k(\omega') B_\ell\rho_B\right].
\]
Now we decompose the integration interval $(-a,a)$ in three subintervals $(-a,\omega-\delta)$, $(\omega-\delta,\omega+\delta)$ and $(\omega+\delta,a)$ and take the limit $\delta\rightarrow0$
\begin{eqnarray*}
S_{k\ell}(\omega)&=&\lim_{\epsilon,\delta\rightarrow 0^+}\int_{-a}^{\omega-\delta} d\omega' \tfrac{(\omega-\omega')}{(\omega-\omega')^2+\epsilon^2} \mathrm{Tr}\left[B_k(\omega') B_\ell\rho_B\right]\\
&+&\lim_{\epsilon,\delta\rightarrow 0^+}\int_{\omega+\delta}^{a} d\omega' \tfrac{(\omega-\omega')}{(\omega-\omega')^2+\epsilon^2} \mathrm{Tr}\left[B_k(\omega') B_\ell\rho_B\right]\\
&+&\lim_{\epsilon,\delta\rightarrow 0^+}\int_{\omega-\delta}^{\omega+\delta} d\omega' \tfrac{(\omega-\omega')}{(\omega-\omega')^2+\epsilon^2} \mathrm{Tr}\left[B_k(\omega') B_\ell\rho_B\right].
\end{eqnarray*}
The last integral just goes to zero because for any $\epsilon$ as small as we want the integrand is an integrable function without singularities. Whereas, since $\lim_{\epsilon\rightarrow 0^+}\tfrac{(\omega-\omega')}{(\omega-\omega')^2+\epsilon^2}=\tfrac{1}{(\omega-\omega')}$, the limiting process in $\delta$ is just the definition of the Cauchy principal value of the integral \cite{Arfken}:
\begin{eqnarray}\label{explicitshift}
S_{k\ell}(\omega)&=&\lim_{\delta\rightarrow 0^+}\left\{\int_{-a}^{\omega-\delta} d\omega' \frac{\mathrm{Tr}\left[B_k(\omega') B_\ell\rho_B\right]}{(\omega-\omega')}+\int_{\omega+\delta}^{a} d\omega' \frac{\mathrm{Tr}\left[B_k(\omega') B_\ell\rho_B\right]}{(\omega-\omega')}\right\}\nonumber\\
&=&\mathrm{P.V.}\int_{-a}^{a} d\omega' \frac{\mathrm{Tr}\left[B_k(\omega') B_\ell\rho_B\right]}{(\omega-\omega')}.
\end{eqnarray}
In theory of distributions, this result is sometimes referred as Sochozki's formulae \cite{Vladimirov}.

\subsubsection{Typical examples}\label{sectionExamples}
It is not worth to dwell on examples of the application of the above results, as in the books are plenty of them. Just for sake of comparison, we sketch some of the most common cases.
\begin{itemize}
\item \underline{Two-level system damped by a bath of harmonic oscillators}. The total Hamiltonian is
\begin{eqnarray*}
H&=&H_{\mathrm{sys}}+H_{\mathrm{bath}}+V\\
&\equiv&\frac{\omega_0}{2}\sigma_z+\int_0^{\omega_\textrm{max}}d\omega a^\dagger_\omega a_\omega+\int_0^{\omega_\textrm{max}}d\omega h(\omega)\left(\sigma_+a_\omega+\sigma_-a_\omega^\dagger\right),
\end{eqnarray*}
where $\sigma_z$ is the corresponding Pauli matrix, $\sigma_\pm=(\sigma_x\pm i\sigma_y)/2$, and $a_\omega$ accounts for the continuous bosonic operators $[a_\omega,a^\dagger_{\omega'}]=\delta(\omega-\omega')$; where $\omega$ is the frequency of each mode. The upper limit $\omega_\textrm{max}$ is the maximum frequency present in the bath of harmonic oscillators, it could be infinity provided that the coupling function $h(\omega)$ decreases fast enough. We have taken the continuous limit from the beginning, however it is common to do it in the course of the computation.

The interaction Hamiltonian may be written in the form $V=A_1\otimes B_1+A_2\otimes B_2$, where it is easy to check that the system operators and their eigendecompositions are given by
\begin{eqnarray*}
A_1&=&\sigma_x=\sigma_++\sigma_-, \quad A_1(\omega_0)=\sigma_-,\ A_1(-\omega_0)=\sigma_+,\\
A_2&=&\sigma_y=i(\sigma_+-\sigma_-), \quad A_2(\omega_0)=-i\sigma_-,\ A_2(-\omega_0)=i\sigma_+.
\end{eqnarray*}
Similarly for the bath operators,
\[
B_1=\int_0^{\omega_\textrm{max}}d\omega h(\omega)\tfrac{(a_\omega+a_\omega^\dagger)}{2}=\int_{-\omega_\textrm{max}}^{\omega_\textrm{max}}d\omega B_1(\omega),
\]
with $B_1(\omega)=\tfrac{h(\omega)a_\omega}{2}$,  $B_1(-\omega)=\tfrac{h(\omega)a_\omega^\dagger}{2}$, and
\[
B_2=\int_0^{\omega_\textrm{max}}d\omega h(\omega)\tfrac{i(a_\omega^\dagger-a_\omega)}{2}=\int_{-\omega_\textrm{max}}^{\omega_\textrm{max}}d\omega B_2(\omega),
\]
with $B_2(\omega)=\tfrac{-ih(\omega)a_\omega}{2}$, $B_2(-\omega)=\tfrac{ih(\omega)a_\omega^\dagger}{2}$.

The decay rates are given by $\gamma_{k,\ell}(\omega_0)=2\pi\tr[B_k(\omega_0)B_{\ell}\rho_B]$, particularly
\begin{eqnarray*}
\gamma_{11}(\omega_0)&=&2\pi\tr[B_1(\omega_0)B_{1}\rho_B]\\
&=&\frac{\pi}{2}h(\omega_0)\int_0^{\omega_\textrm{max}}d\omega h(\omega)\tr[a_{\omega_0}(a_\omega+a_\omega^\dagger)\rho_B]\\
&=&\frac{\pi}{2}h^2(\omega_0)\tr[a_{\omega_0}a_{\omega_0}^\dagger\rho_B]=\frac{\pi}{2}J(\omega_0)[\bar{n}(\omega_0)+1].
\end{eqnarray*}
Here $J(\omega_0)=h^2(\omega_0)$ is the so-called spectral density of the bath (which accounts for the strength of the coupling per frequency) at frequency $\omega_0$ and $\bar{n}(\omega_0)$ is the mean number of bosons in the thermal state $\rho_B$ of the bath with frequency $\omega_0$; i.e. the expected number of particles of the Bose-Einstein statistics
\[
\bar{n}(\omega_0)=\left[e^{(\omega_0/T)}-1\right]^{-1},
\]
where $T$ is the temperature of the bath. For the rest of the elements one obtains essentially the same
\[
\boldsymbol{\gamma}(\omega_0)=\frac{\pi}{2}J(\omega_0)[\bar{n}(\omega_0)+1]\left(
\begin{array}{cc}
1 & i\\
-i  & 1
\end{array}\right).
\]
Similarly for $\gamma_{k,\ell}(-\omega_0)$,
\[
\boldsymbol{\gamma}(-\omega_0)=\frac{\pi}{2}J(\omega_0)\bar{n}(\omega_0)\left(
\begin{array}{cc}
1 & i\\
-i  & 1
\end{array}\right).
\]
Let us make the analog for the shifts $S_{k\ell}$; the first element is given by
\begin{eqnarray*}
S_{11}(\omega_0)&=&\mathrm{P.V.}\int_{-\omega_\textrm{max}}^{\omega_\textrm{max}} d\omega' \frac{\mathrm{Tr}\left[B_1(\omega') B_1\rho_B\right]}{(\omega_0-\omega')}\\
&=&\mathrm{P.V.}\int_{0}^{\omega_\textrm{max}} d\omega' \frac{\mathrm{Tr}\left[B_k(\omega') B_\ell\rho_B\right]}{(\omega_0-\omega')}\\
&+&\mathrm{P.V.}\int_{0}^{\omega_\textrm{max}} d\omega' \frac{\mathrm{Tr}\left[B_k(-\omega') B_\ell\rho_B\right]}{(\omega_0+\omega')}\\
&=&\frac{1}{4}\mathrm{P.V.}\int_{0}^{\omega_\textrm{max}} d\omega' J(\omega')\left[\frac{\bar{n}(\omega')+1}{(\omega_0-\omega')}+\frac{\bar{n}(\omega')}{(\omega_0+\omega')}\right],
\end{eqnarray*}
and the rest of the elements are
\begin{eqnarray*}
S_{22}(\omega_0)&=&S_{11}(\omega_0),\\
S_{12}(\omega_0)&=&S_{21}^{\ast}(\omega_0)=\frac{i}{4}\mathrm{P.V.}\int_{0}^{\omega_\textrm{max}} d\omega' J(\omega')\left[\frac{\bar{n}(\omega')+1}{(\omega_0-\omega')}-\frac{\bar{n}(\omega')}{(\omega_0+\omega')}\right].\\
\end{eqnarray*}
For the shifts with $-\omega_0$:
\begin{eqnarray*}
S_{11}(-\omega_0)&=&S_{22}(-\omega_0)=\frac{-1}{4}\mathrm{P.V.}\int_{0}^{\omega_\textrm{max}} d\omega' J(\omega')\left[\frac{\bar{n}(\omega')+1}{(\omega_0+\omega')}+\frac{\bar{n}(\omega')}{(\omega_0-\omega')}\right],\\
S_{12}(-\omega_0)&=&S_{21}^{\ast}(-\omega_0)=\frac{-i}{4}\mathrm{P.V.}\int_{0}^{\omega_\textrm{max}} d\omega' J(\omega')\left[\frac{\bar{n}(\omega')+1}{(\omega_0+\omega')}+\frac{\bar{n}(\omega')}{(\omega_0-\omega')}\right].\\
\end{eqnarray*}
Therefore according to (\ref{Hlambshift}) the shift Hamiltonian is
\begin{eqnarray*}
H_{\mathrm{LS}}&=&\mathrm{P.V.}\int_{0}^{\omega_\textrm{max}} d\omega' \frac{J(\omega')\bar{n}(\omega')}{(\omega_0-\omega')}[\sigma_+,\sigma_-]+\mathrm{P.V.}\int_{0}^{\omega_\textrm{max}} d\omega' \frac{J(\omega')}{(\omega_0-\omega')}\sigma_+\sigma_-\\
&=&\mathrm{P.V.}\int_{0}^{\omega_\textrm{max}} d\omega' \frac{J(\omega')\bar{n}(\omega')}{(\omega_0-\omega')}\sigma_z+\mathrm{P.V.}\int_{0}^{\omega_\textrm{max}} d\omega' \frac{J(\omega')}{(\omega_0-\omega')}\frac{(\sigma_z+\mathds{1})}{2}\\
&\equiv&\left(\Delta'+\frac{\Delta}{2}\right)\sigma_z,
\end{eqnarray*}
where $\Delta'$ is the integral depending on $\bar{n}(\omega')$ which leads to a shift due to the presence of the thermal field (Stark-like shift), and $\Delta$ is the second integral independent of $\bar{n}(\omega')$ which lead to a Lamb-like shift effect. Of course the identity $\mathds{1}$ does not contribute to the dynamics as it commutes with any $\rho$, so we can forget it.

For the total equation we immediately obtain
\begin{eqnarray*}
\frac{d\rho_{\mathrm{sys}}(t)}{dt}&=&-i\left[\left(\frac{\omega_0+\Delta}{2}+\Delta'\right)\sigma_z,\rho_{\mathrm{sys}}(t)\right]\\
&+&\Gamma[\bar{n}(\omega_0)+1]\left[ \sigma_-\rho_{\mathrm{sys}}(t)\sigma_+-\frac{1}{2}\{\sigma_+\sigma_-,\rho_{\mathrm{sys}}(t)\}\right]\\
&+&\Gamma\bar{n}(\omega_0)\left[ \sigma_+\rho_{\mathrm{sys}}(t)\sigma_--\frac{1}{2}\{\sigma_-\sigma_+,\rho_{\mathrm{sys}}(t)\}\right],
\end{eqnarray*}
where $\Gamma=2\pi J(\omega_0)$.
\item \underline{Harmonic oscillator damped by a bath of harmonic oscillators}. This case is actually quite similar to the previous one, but the two-level system is now substituted by another harmonic oscillator:
\begin{eqnarray*}
H&=&H_{\mathrm{sys}}+H_{\mathrm{bath}}+V\\
&\equiv&\omega_0a^\dagger a+\int_0^{\omega_\textrm{max}}d\omega a^\dagger_\omega a_\omega+\int_0^{\omega_\textrm{max}}d\omega h(\omega)\left(a^\dagger a_\omega+aa_\omega^\dagger\right).
\end{eqnarray*}
The bath operators $B_1$ and $B_2$ are the same as above and the system operators are now
\begin{eqnarray*}
A_1&=&a^\dagger+a, \quad A_1(\omega_0)=a,\ A_1(-\omega_0)=a^\dagger,\\
A_2&=&i(a^\dagger-a), \quad A_2(\omega_0)=-ia,\ A_2(-\omega_0)=ia^\dagger.
\end{eqnarray*}
Thus by mimicking the previous steps we obtain
\begin{eqnarray*}
\frac{d\rho_{\mathrm{sys}}(t)}{dt}&=&-i\left[\left(\omega_0+\Delta\right)\sigma_z,\rho_{\mathrm{sys}}(t)\right]\\
&+&\Gamma[\bar{n}(\omega_0)+1]\left[ a\rho_{\mathrm{sys}}(t)a^\dagger-\frac{1}{2}\{a^\dagger a,\rho_{\mathrm{sys}}(t)\}\right]\\
&+&\Gamma\bar{n}(\omega_0)\left[a^\dagger \rho_{\mathrm{sys}}(t)a-\frac{1}{2}\{a a^\dagger,\rho_{\mathrm{sys}}(t)\}\right].
\end{eqnarray*}
Note in this case that, on one hand, the Stark-like shift $\Delta'$ does not contribute because of the commutation rule of $a$ and $a^\dagger$; and, on the other hand, the dimension of the system is infinite and the interaction with the environment is mediated by an unbounded operator. However as long as domain problems do not arise one expects that this is the correct approximation to the dynamics (and actually it is, see for instance \cite{Rivas09}).
\item \underline{Pure dephasing of a two-level system}. In this case the interaction commutes with the system Hamiltonian and so the populations of the initial system density matrix remain invariant,
\begin{eqnarray*}
H&=&H_{\mathrm{sys}}+H_{\mathrm{bath}}+V\\
&\equiv&\frac{\omega_0}{2}\sigma_z+\int_0^{\omega_\textrm{max}}d\omega a^\dagger_\omega a_\omega+\int_0^{\omega_\textrm{max}}d\omega h(\omega)\sigma_z\left(a_\omega+a_\omega^\dagger\right).
\end{eqnarray*}
This is an interesting case to explore because it involves some subtleties, but it is exactly solvable (see for example chapter 4 of \cite{BrPe02}, \cite{Zueco}, and references therein). The interaction Hamiltonian is decomposed as $V=A\otimes B$, where $A=A(\omega=0)=\sigma_z$ ($\sigma_z$ is eigenoperator of $[\sigma_z,\cdot]$ with zero eigenvalue!) and
\[
B=\int_{-\omega_\textrm{max}}^{\omega_\textrm{max}}d\omega B(\omega), \text{ with}\left\{
\begin{array}{l}
B(\omega)=h(\omega)a_\omega, \\
B(-\omega)=h(\omega) a_\omega^\dagger,
\end{array}
\right. \text{for }\omega>0.
\]
So some problem arises when we calculate the decay rates because $B(\omega=0)$ is not well-defined. The natural solution to this is is to understand $\omega=0$ as a limit in such a way that
\[
\gamma(0)=2\pi\lim_{\omega\rightarrow0}\tr[B(\omega)B\rho_B],
\]
but then there is a further problem. Depending on which side we approach to 0 the function $\tr[B(\omega)B\rho_B]$ takes the value $2\pi J(\omega)[\bar{n}(\omega)+1]$ for $\omega>0$ or $2\pi J(|\omega|)\bar{n}(|\omega|)$ for $\omega<0$. Thus the limit in principle would not be well-defined, except in the high temperature regime where $\bar{n}(|\omega|)\simeq \bar{n}(|\omega|)+1$. However both limits give the same result provided that one assumes $\lim_{\omega\rightarrow0}J(|\omega|)=0$. This is a natural physical assumption because a mode with 0 frequency does not have any energy and influence on the system. Thus if $\lim_{\omega\rightarrow0}J(|\omega|)=0$ we find
\[
\gamma(0)=2\pi\lim_{\omega\rightarrow0} J(|\omega|)\bar{n}(|\omega|).
\]
Actually, since $\bar{n}(|\omega|)$ goes to infinity as $1/|\omega|$ when $\omega$ decreases, if $J(|\omega|)$ does not tend to zero linearly in $\omega$ (this is called Ohmic-type spectral density \cite{Weiss08}) this constant $\gamma(0)$ will be either 0 or infinity. This fact restricts quite a lot the spectral densities which can be treated for the pure dephasing problem in the weak coupling framework with significant results (which is not a major problem because this model is exactly solvable as we have already pointed out). On the other hand, shifts do not occur because $H_{\mathrm{LS}}\propto\sigma_z\sigma_z=\mathds{1}$, and finally the master equations is
\[
\frac{d\rho_{\mathrm{sys}}(t)}{dt}=-i\left[\frac{\omega_0}{2}\sigma_z,\rho_{\mathrm{sys}}(t)\right]+\gamma(0)\left[\sigma_z\rho_{\mathrm{sys}}(t)\sigma_z-\rho_{\mathrm{sys}}(t)\right].
\]

\end{itemize}

\subsubsection{Some remarks on the secular approximation and positivity preserving requirement}
As we have seen, the weak coupling limit provides a rigorous mathematical procedure to derive Markovian master equations. Nonetheless, in solid state or chemical physics the interactions are typically stronger than for electromagnetic environments, and it is sometimes preferred to use the perturbative treatment without performing the secular approximation. However, this method jeopardizes the positivity of the dynamics \cite{DumckeandSpohn} (see also \cite{Whitney08,ZhaoChen02}) and several tricks have been proposed to heal this drawback. For example, one possibility is taking in consideration just the subset of the whole possible states of the system which remain positive in the evolution \cite{Silbey}, or describing the evolution by the inclusion of a ``slippage'' operator \cite{Gaspard}. These proposals can be useful just in some situations; in this regard it seems they do not work well for multipartite systems \cite{BenattiSlipagge1,BenattiSlipagge2}.

We have already explained in detail the requirement of the complete positivity for a universal dynamics. The option to proceed without it may be useful but risky, and actually the risk seems to be quite high as the positivity of the density matrix supports the whole consistency of the quantum theory. We believe the secular approximation should be used as it is mathematically correct in the weak coupling limit. In case of the failure of a completely positive semigroup to describe correctly the dynamics of the physical system, it is probably more appropriate the use of non-Markovian methods as the ``dynamical coarse graining'' (see section \ref{sectionDCG}) which preserves complete positivity.

\subsubsection{Failure of the assumption of factorized dynamics} \label{sectionFactorizedRev}

In some derivations of the weak coupling limit it is common to use the following argument. One makes the formal integration of the von Neumann equation in the interaction picture (\ref{Von-NeumannInt})
\[
\frac{d}{dt}\tilde{\rho}(t)=-i\alpha[\tilde{V}(t),\tilde{\rho}(t)],
\]
to give
\[
\tilde{\rho}(t)=\rho(0)-i\alpha\int_0^tds[\tilde{V}(s),\tilde{\rho}(s)].
\]
By iterating this equation twice one gets
\[
\tilde{\rho}(t)=\rho(0)-i\alpha\int_0^tds[\tilde{V}(t),\rho(0)]-\alpha^2\int_0^tds\int_0^sdu[\tilde{V}(s),[\tilde{V}(u),\tilde{\rho}(u)]].
\]
After taking partial trace, under the usual aforementioned assumptions of $\rho(0)=\rho_A(0)\otimes\rho_B(0)$, and $\tr_B\left[\tilde{V}(t)\rho_B(0)\right]=0$, this equation is simplified to
\[
\tilde{\rho}_A(t)=\rho_A(0)-\alpha^2\int_0^tds\int_0^sdu\tr_B[\tilde{V}(s),[\tilde{V}(u),\tilde{\rho}(u)]].
\]
Now, one argues that the system $B$ is typically much larger than $A$ and so, provided that the interaction is weak, the state of $B$ remains unperturbed by the presence of the $A$, and the whole state can be approximated as $\tilde{\rho}(u)\approx\tilde{\rho}_A(u)\otimes\rho_B(u)$. By making this substitution in the above equation one arrives at equation (\ref{weakcoupling2}) (modulo some small details which one can consult in most of the textbooks of the references), and henceforth proceed to derive the Markovian master equation without any mention to projection operators, etc.

This method can be considered as an effective model in order to obtain equation (\ref{weakcoupling2}), however it does not make any sense to assume that the physical state of the system factorizes for all times. Otherwise how would it be possible the interaction between both subsystems?. Moreover, why not making the substitution $\tilde{\rho}(u)\approx\tilde{\rho}_A(u)\otimes\rho_B(u)$ directly in the von Neumann equation (\ref{Von-NeumannInt}) instead of in its integrated and iterated version?.

The validity of the substitution $\tilde{\rho}(u)\approx\tilde{\rho}_A(u)\otimes\rho_B(u)$ is justified because of the factorized initial condition and the perturbative approach of the weak coupling limit. Of course we pointed out
previously
that the weak coupling approach makes sense only if the coupling is small and the environment has infinite degrees of freedom. This fits in the above argument about the size of $B$, but it can dismissed even by numerical simulations (see \cite{Rivas09}) that the real total state $\tilde{\rho}(u)$ is close to $\tilde{\rho}_A(u)\otimes\rho_B(u)$. It should be therefore be considered as an ansatz to arrive at the correct equation rather than a physical requirement on the evolution.

\subsubsection{Steady state properties} \label{sectionSteadyweak}
In this section we shall analyze the steady state properties of the Markovian master equation obtained in the weak coupling limit (\ref{weakcoupling4}). The most important result is that the thermal state of the system
\[
\rho_A^{\mathrm{th}}=\frac{e^{(-H_A/T)}}{\tr[e^{(-H_A/T)}]},
\]
with the same temperature T as the bath is a steady state. To see this one notes that the bath correlation functions satisfies the so-called Kubo-Martin-Schwinger (KMS) condition \cite{KMSpapers}, namely
\begin{eqnarray}\label{KMS}
\langle \tilde{B}_k(u) B_\ell \rangle&=&\tr\left[\rho_\mathrm{th}\tilde{B}_k(u) B_\ell\right]=\frac{1}{\tr\left(e^{-\beta H_A}\right)}\tr\left[e^{iH_B(u+i\beta)}B_k e^{-iH_Bu} B_\ell\right]\nonumber \\
&=&\frac{1}{\tr\left(e^{-\beta H_A}\right)}\tr\left[B_\ell e^{iH_B(u+i\beta)}B_k e^{-iH_B(u+i\beta)}e^{-\beta H_B}\right]\nonumber\\
&=&\tr\left[B_\ell e^{iH_B(u+i\beta)}B_k e^{-iH_B(u+i\beta)} \rho_\mathrm{th}\right]=\langle B_\ell \tilde{B}_k(u+i\beta) \rangle.
\end{eqnarray}
By decomposing $\tilde{B}_k$ in eigenoperators $B(\omega')$ of $[H_B,\cdot]$, we obtain, similarly to (\ref{explicitgamma}), the following relation between their Fourier transforms
\begin{eqnarray}
\gamma_{k\ell}(\omega)&=&\int_{-\infty}^{\infty}due^{i\omega u}\mathrm{Tr}\left[\tilde{B}_k(u) B_\ell\rho_\mathrm{th}\right]\nonumber \\
&=&\int_{-\infty}^{\infty}due^{i \omega u}\mathrm{Tr}\left[B_\ell \tilde{B}_k(u+i\beta)\rho_\mathrm{th}\right]\nonumber \\
&=&\int_{-a}^a d\omega'e^{\beta\omega'}\int_{-\infty}^{\infty}due^{i(\omega-\omega') u}\mathrm{Tr}\left[B_\ell B_k(\omega') \rho_\mathrm{th} \right]\nonumber \\
&=&2\pi\int_{-a}^a d\omega'e^{\beta\omega'}\delta(\omega-\omega')\mathrm{Tr}\left[B_\ell B_k(\omega') \rho_\mathrm{th}\right]\nonumber \\
&=&2\pi e^{\beta\omega}\mathrm{Tr}\left[B_\ell B_k(\omega)\rho_\mathrm{th}\right]=2\pi e^{\beta\omega}\left\{\mathrm{Tr}\left[B_k^\dagger(\omega) B_\ell \rho_\mathrm{th}\right]\right\}^\ast\nonumber\\
&=&2\pi e^{\beta\omega}\left\{\mathrm{Tr}\left[B_k(-\omega) B_\ell \rho_\mathrm{th}\right]\right\}^\ast\nonumber\\
&=&e^{\beta\omega}\gamma_{k\ell}^\ast(-\omega)=e^{\beta\omega}\gamma_{\ell k}(-\omega).
\end{eqnarray}
On the other hand, since $A(\omega')$ is an eigenoperator of $[H_A,\cdot]$, one easily obtains the relations
\begin{eqnarray}
\rho_A^{\mathrm{th}}A_k(\omega)&=&e^{\beta\omega}A_k(\omega)\rho_A^{\mathrm{th}},\\
\rho_A^{\mathrm{th}}A_k^\dagger(\omega)&=&e^{-\beta\omega}A_k^\dagger(\omega)\rho_A^{\mathrm{th}}.
\end{eqnarray}
Now one can check that $\frac{d\rho_A^{\mathrm{th}}}{dt}=0$ according to equation (\ref{weakcoupling4}). Indeed, it is obvious that $\rho_A^{\mathrm{th}}$ commutes with the Hamiltonian part, for the remaining one we have
\begin{eqnarray*}
\frac{d\rho_A^{\mathrm{th}}}{dt}&=&\sum_{\omega}\sum_{k,\ell}\gamma_{k\ell}(\omega)\left[ A_\ell(\omega)\rho_A^{\mathrm{th}}A_k^\dagger(\omega)-\frac{1}{2}\{A_k^\dagger(\omega)A_\ell(\omega),\rho_A^{\mathrm{th}}\}\right]\\
&=&\sum_{\omega}\sum_{k,\ell}\gamma_{k\ell}(\omega)\left[ e^{-\beta\omega}A_\ell(\omega)A_k^\dagger(\omega)\rho_A^{\mathrm{th}}-A_k^\dagger(\omega)A_\ell(\omega)\rho_A^{\mathrm{th}}\right]\\
&=&\sum_{\omega}\sum_{k,\ell}\left[ \gamma_{k\ell}(\omega)e^{-\beta\omega}A_\ell^\dagger(-\omega)A_k(-\omega)-\gamma_{k\ell}(\omega)A_k^\dagger(\omega)A_\ell(\omega)\right]\rho_A^{\mathrm{th}}\\
&=&\sum_{\omega}\sum_{k,\ell}\left[ \gamma_{\ell k}(-\omega)A_\ell^\dagger(-\omega)A_k(-\omega)-\gamma_{k\ell}(\omega)A_k^\dagger(\omega)A_\ell(\omega)\right]\rho_A^{\mathrm{th}}=0,
\end{eqnarray*}
as the sum in $\omega$ goes from $-\omega_{\textrm{max}}$ to $\omega_{\textrm{max}}$ for some maximum difference between energies $\omega_{\textrm{max}}$. This proves that $\rho_A^{\mathrm{th}}$ is a steady state. Of course this does not implies the convergency of any initial state to it (an example is the pure dephasing equation derived in section \ref{sectionExamples}). Sufficient conditions for that were given in the section \ref{sectionSteady}.

\subsection{Singular coupling limit} \label{sectionsingularcoupling}
As explained in section \ref{sectionNaka-Zwan}, the singular coupling limit is somehow the complementary situation to the weak coupling limit, since in this case we get $\tau_B\rightarrow0$ by making the correlation functions to approach a delta function. However, as a difference with the weak coupling limit, the singular coupling limit is not so useful in practice because as its own name denotes, it requires some ``singular'' situation. For that reason and for sake of comparison, we sketch here the ideas about this method from the concepts previously defined. For further details the reader will be referred to the literature.

There are several ways to implement the singular coupling limit, one consists in rescaling the total Hamiltonian as \cite{AlickiLendi87,BrPe02,Benatti05}
\[
H=H_A+\alpha^{-2}H_B+\alpha^{-1}V,
\]
and perform the ``singular limit'' by taking $\alpha\rightarrow0$. Starting from the equation (\ref{projectorPInt2}) under the assumptions of $\mathcal{P}\mathcal{V}(t)\mathcal{P}=0$ and $\rho(0)=\rho_A(0)\otimes\rho_B(0)$ we have
\begin{equation}\label{singularcoupling1}
\mathcal{P}\tilde{\rho}(t)=\mathcal{P}\rho(0)+\frac{1}{\alpha^2}\int_{0}^{t}ds\int_{0}^sdu\mathcal{P}\mathcal{V}(s)\mathcal{G}(s,u)\mathcal{Q}\mathcal{V}(u)\mathcal{P}\tilde{\rho}(u).
\end{equation}
For a moment, let us take a look to the expression at first order in the expansion of $\mathcal{G}(s,u)$:
\[
\mathcal{P}\tilde{\rho}(t)=\mathcal{P}\rho(0)+\frac{1}{\alpha^2}\int_{0}^{t}ds\int_{0}^sdu\mathcal{P}\mathcal{V}(s)\mathcal{V}(u)\mathcal{P}\tilde{\rho}(u)+\mathcal{O}(\alpha^{-3}),
\]
which we rewrite as
\[
\tilde{\rho}_A(t)=\rho_A(0)-\frac{1}{\alpha^2}\int_{0}^{t}ds\int_{0}^sdu\left[\tilde{V}(s)[\tilde{V}(u)\tilde{\rho}_A(u)\otimes\rho_B]\right]+\mathcal{O}(\alpha^{-3}).
\]
By introducing the general form of the interaction Eq. (\ref{Vdesc}),
\begin{eqnarray}
\tilde{\rho}_A(t)=\rho_A(0)&+&\int_{0}^{t}ds\int_{0}^sduC_{k\ell}(s,u)\left[\tilde{A}_\ell(u)\tilde{\rho}_A(u),\tilde{A}_k(s)\right]\nonumber\\
&+&C_{k\ell}^\ast(s,u)\left[\tilde{A}_k(s),\tilde{\rho}_A(u)\tilde{A}_\ell(u)\right]+\mathcal{O}(\alpha^{-3}),\label{singularcoupling2}
\end{eqnarray}
here the correlation functions are
\[
C_{k\ell}(s,u)=\frac{1}{\alpha^2}\tr[\tilde{B}_k(s)\tilde{B}_\ell(u)\rho_B].
\]
If we assume that $\rho_B$ is again some eigenstate of the free Hamiltonian we get
\begin{equation}\label{CorrSingular}
C_{k\ell}(s-u)=\frac{1}{\alpha^2}\tr[\tilde{B}_k(s-u)B_\ell\rho_B],
\end{equation}
and by using the eigendecomposition (\ref{Beigendecomp}) of $B_k$ and taking into account the factor $\alpha^{-2}$ in the free Hamiltonian $H_B$, we find
\[
C_{k\ell}(s-u)=\frac{1}{\alpha^2}\int_{-a}^ad\omega e^{-\frac{i\omega(s-u)}{\alpha^2}}\tr[B_k(\omega)B_\ell\rho_B].
\]
Because the proposition \ref{propAlfredo}, in the limit $\alpha\rightarrow0$ the above integral tends to zero as a square integrable function in $x=\alpha^{-2}(s-u)$. It implies that it tends to zero at least as $\alpha^2$ and so we expect this convergence to be faster than the prefactor $\frac{1}{\alpha^2}$ in the correlation functions. However in the peculiar case of $s=u$ the correlation functions go to infinity, so indeed in this limit they approach a delta function $C_{k\ell}(s-u)\propto\delta(s-u)$. Next, by making again the change of variable $u\rightarrow s-u$, equation (\ref{singularcoupling2}) reads
\begin{eqnarray*}
\tilde{\rho}_A(t)=\rho_A(0)&+&\int_{0}^{t}ds\int_{0}^sduC_{k\ell}(u)\left[\tilde{A}_\ell(s-u)\tilde{\rho}_A(s-u),\tilde{A}_k(s)\right]\\
&+&C_{k\ell}^\ast(u)\left[\tilde{A}_k(s),\tilde{\rho}_A(s-u)\tilde{A}_\ell(s-u)\right]+\mathcal{O}(\alpha^{-3}).
\end{eqnarray*}
Therefore in the limit of $\alpha\rightarrow0$, we can substitute $u$ by $0$ in the $\tilde{A}_\ell$ operators and extend the integration to infinite without introducing an unbounded error
\begin{eqnarray*}
\tilde{\rho}_A(t)=\rho_A(0)&+&\int_{0}^{t}ds\Gamma_{k\ell}\left[\tilde{A}_\ell(s)\tilde{\rho}_A(s),\tilde{A}_k(s)\right]\\
&+&\Gamma_{\ell k}^\ast\left[\tilde{A}_\ell(s),\tilde{\rho}_A(s)\tilde{A}_k(s)\right]+\mathcal{O}(\alpha^{-3}),
\end{eqnarray*}
where
\[
\Gamma_{k\ell}=\int_0^\infty duC_{k\ell}(u).
\]
We may split this matrix in Hermitian $\gamma_{k\ell}=\int_0^\infty duC_{k\ell}(u)$ and anti-Hermitian part $iS_{k\ell}$ as in equation (\ref{Gammadescomposition}), and differentiate to obtain
\begin{eqnarray*}
\frac{d\tilde{\rho}_A(t)}{dt}=&-&i[\tilde{H}_{\mathrm{LS}}(t),\tilde{\rho}(t)]\\
&+&\gamma_{k\ell}\left[\tilde{A}_\ell(t)\tilde{\rho}_A(t)\tilde{A}_k(t)-\frac{1}{2}\left\{\tilde{A}_k(t)\tilde{A}_\ell(t),\tilde{\rho}_A(t)\right\}\right]+\mathcal{O}(\alpha^{-3}),
\end{eqnarray*}
where $\tilde{H}_{\mathrm{LS}}(t)=\sum_{k,\ell}S_{k\ell}\tilde{A}_k(t)\tilde{A}_\ell(t)$. Since $\gamma_{k\ell}$ has the same expression as in (\ref{gammapequenha}) with $\omega=0$ it is always a positive semidefinite matrix and up to second order the evolution equation has the form of a Markovian master equation. However, from the former it is not clear what happens with the higher order terms; a complete treatment of them can be found in \cite{GoKo76,FrGo76,FrNoVe76}, where it is proven that they vanish. The idea behind this fact is that by using the property that for a Gaussian state (which $\rho_B$ is assumed to be) the correlation functions of any order can be written as a sum of products of two time correlation functions (\ref{CorrSingular}). Particularly odd order correlation functions vanish because of the condition $\mathcal{P}\mathcal{V}(t)\mathcal{P}=0$, and we obtain a product of deltas for even order. Actually the expansion is left only with those products in which some time arguments appear in ``overlapping order'' such as
\[
\delta(t_1-t_3)\delta(t_2-t_4),
\]
for $t_4\geq t_3\geq t_2\geq t_1$, and because the time-ordered integration in the expansion of $\mathcal{G}(s,u)$ they do not give any contribution.

Finally by coming back to Schr\"odinger picture we obtain that in the singular coupling limit
\begin{eqnarray*}
\frac{d \rho_A(t)}{dt}=&-&i[H_A+H_{\mathrm{LS}},\rho(t)]\\
&+&\gamma_{k\ell}\left[A_\ell\rho_A(t)A_k-\frac{1}{2}\left\{A_kA_\ell,\rho_A(t)\right\}\right],
\end{eqnarray*}
note that $H_{\mathrm{LS}}=\sum_{k,\ell}S_{k\ell}A_kA_\ell$ does not commute in general with $H_A$.

Despite the singular coupling limit is not very realistic, it is possible to find effective equations for the evolution in which self-adjoint operators $A_k$ also appear. For example under interactions with classic stochastic external fields \cite{GoKo76} or in the study of the continuous measurement processes \cite{BrPe02,CavesMilburn}.

\subsection{Extensions of the weak coupling limit}
Here we will discuss briefly some extensions of the weak coupling limit for interacting systems, the reader can consult further studies as \cite{Rivas09} for example.

\subsubsection{Weak coupling for multipartite systems}\label{sectionweakcomposite}
Consider two equal (for simplicity) quantum systems and their environment, such that the total Hamiltonian is given by
\[
H=H_{A1}+H_{A2}+H_E+\alpha V.
\]
Here $H_{A1}$ ($H_{A2}$) is the free Hamiltonian of the subsystem 1 (2), $H_E$ is the free Hamiltonian of the environment and $\alpha V$ denotes the interaction between the subsystems and the environment. This interaction term is assumed to have the form $V=\sum_kA_k^{(1)}\otimes B_k+A_k^{(2)}\otimes B_k$, where $A_k^{(1)}$ ($A_k^{(2)}$) are operators acting on the subsystem 1 (2), and $B_k$ on the environment. So that the weak coupling limit for the whole state of both subsystems $\rho_{12}$ leads to a Markovian master equation like
\begin{eqnarray}\label{weakcouplingcomposite1}
\frac{d\rho_{12}(t)}{dt}&=&-i[H_{A1}+H_{A2}+\alpha^2H_{\mathrm{LS}},\rho_{12}(t)] \\
&+&\alpha^2\sum_{\omega}\sum_{k,\ell}\gamma_{k\ell}(\omega)\left[ A_\ell(\omega)\rho_{12}(t)A_k^\dagger(\omega)-\frac{1}{2}\{A_k^\dagger(\omega)A_\ell(\omega),\rho_{12}(t)\}\right]\nonumber,
\end{eqnarray}
where $A_k(\omega)$ are linear combinations of $A_k^{(1)}$ and $A_k^{(2)}$ decomposed in eigenoperators of $[H_{A1}+H_{A2},\cdot]$. Of course the form of these operators depend on the form of the coupling. For instance, suppose that we have two identical environments acting locally on each subsystem $V=\sum_kA_k^{(1)}\otimes B_k^{(1)}+A_k^{(2)}\otimes B_k^{(2)}$, then we will obtain
\begin{eqnarray*}
\frac{d\rho_{12}(t)}{dt}&=&-i\sum_{j=1}^2[H_{A1}+H_{A2}+\alpha^2H_{\mathrm{LS}}^{(j)},\rho_{12}(t)]\nonumber \\
&+&\alpha^2\sum_{\omega}\sum_{k,\ell}\gamma_{k\ell}^{(j)}(\omega)\left[ A_\ell^{(j)}(\omega)\rho_{12}(t)A_k^{(j)\dagger}(\omega)-\frac{1}{2}\{A_k^{(j)\dagger}(\omega)A_\ell^{(j)}(\omega),\rho_{12}(t)\}\right].
\end{eqnarray*}

Consider now the case in which we perturb the free Hamiltonian of the subsystems by an interaction term between them of the form $\beta V_{12}$ (where $\beta$ accounts for the strength). Instead of using the interaction picture with respect to the free evolution $H_{A1}+H_{A2}+H_{B}$, the general strategy to deal with this problem is to consider the interaction picture including the interaction between subsystems $H_{A1}+H_{A2}+\beta V_{12}+H_{B}$, and then proceeds with the weak coupling method.

However if $\beta$ is small, one can follow a different route. First, taking the interaction picture with respect to $H_{A1}+H_{A2}+H_{B}$, the evolution equation for the global system is
\[ \frac{d}{dt}\tilde{\rho}(t)=-i\beta[\tilde{V}_{12}(t),\tilde{\rho}(t)]-i\alpha[\tilde{V}(t),\tilde{\rho}(t)]\equiv\beta \mathcal{V}_{12}(t)\tilde{\rho}(t)+\alpha\mathcal{V}(t)\tilde{\rho}(t).
\]
Similarly to section \ref{sectionNaka-Zwan}, we take projection operators
\begin{eqnarray}
\frac{d}{dt}\mathcal{P}\tilde{\rho}(t)&=\beta\mathcal{P}\mathcal{V}_{12}(t)\tilde{\rho}(t)+\alpha\mathcal{P}\mathcal{V}_{SB}(t)\tilde{\rho}(t),\label{projectorPcomposite}\\
\frac{d}{dt}\mathcal{Q}\tilde{\rho}(t)&=\beta\mathcal{Q}\mathcal{V}_{12}(t)\tilde{\rho}(t)+\alpha\mathcal{Q}\mathcal{V}_{SB}(t)\tilde{\rho}(t),
\end{eqnarray}
and find the formal solution to the second equation
\begin{eqnarray}\label{projectorQcomposite} \mathcal{Q}\tilde{\rho}(t)=\mathcal{G}(t,0)\mathcal{Q}\tilde{\rho}(0)+\beta\int_{0}^tds\mathcal{G}(t,s)\mathcal{Q}\mathcal{V}_{12}(s)\mathcal{P}\tilde{\rho}(s)\nonumber\\
+\alpha\int_{0}^tds\mathcal{G}(t,s)\mathcal{Q}\mathcal{V}(s)\mathcal{P}\tilde{\rho}(s),
\end{eqnarray}
where
\[
\mathcal{G}(t,s)=\mathcal{T}e^{\int_s^tdt'\mathcal{Q}[\beta\mathcal{V}_{12}(t')+\alpha\mathcal{V}(t')]}.
\]
Now the procedure is as follows, we introduce the identity $\mathds{1}=\mathcal{P}+\mathcal{Q}$ just in the last term of equation (\ref{projectorPcomposite}),
\[
\frac{d}{dt}\mathcal{P}\tilde{\rho}(t)=\beta\mathcal{P}\mathcal{V}_{12}(t)\tilde{\rho}(t)+\alpha\mathcal{P}\mathcal{V}(t)\mathcal{P}\tilde{\rho}(t)+\alpha\mathcal{P}\mathcal{V}(t)\mathcal{Q}\tilde{\rho}(t),
\]
whose formal integration yields
\begin{equation}\label{projectorPcompositeInt1}
\mathcal{P}\tilde{\rho}(t)=\mathcal{P}\rho(0)+\beta\int_{0}^{t}ds\mathcal{P}\mathcal{V}_{12}(s)\tilde{\rho}(s)+\alpha\int_{0}^{t}ds\mathcal{P}\mathcal{V}(s)\mathcal{Q}\tilde{\rho}(s),
\end{equation}
here we have used the condition $\mathcal{P}\mathcal{V}(t)\mathcal{P}=0$. Next we insert the formal solution (\ref{projectorQcomposite}) into the last term. By assuming again an initial factorized state ``subsystems$\otimes$enviroment'' ($\mathcal{Q}\rho(t_0)=0$) we find
\begin{eqnarray*}
\mathcal{P}\tilde{\rho}(t)=\mathcal{P}\rho(0)&+&\beta\int_{0}^{t}ds\mathcal{P}\mathcal{V}_{12}(s)\tilde{\rho}(s)\\
&+&\int_{0}^{t}ds\int_{0}^sdu\mathcal{K}_1(s,u)\mathcal{P}\tilde{\rho}(u)+\int_{0}^{t}ds\int_{0}^sdu\mathcal{K}_2(s,u)\mathcal{P}\tilde{\rho}(u),
\end{eqnarray*}
where the kernels are
\begin{eqnarray*}
\mathcal{K}_1(s,u)&=&\alpha\beta\mathcal{P}\mathcal{V}(s)\mathcal{G}(s,u)\mathcal{Q}\mathcal{V}_{12}(u)\mathcal{P}=0,\\
\mathcal{K}_2(s,u)&=&\alpha^2\mathcal{P}\mathcal{V}(s)\mathcal{G}(s,u)\mathcal{Q}\mathcal{V}(u)\mathcal{P}.
\end{eqnarray*}
The first one vanishes because $\mathcal{V}_{12}(s)$ commutes with $\mathcal{P}$
and $\mathcal{Q}\mathcal{P}=0$, and the second kernel up to second order in $\alpha$ and $\beta$ becomes
\[
\mathcal{K}_2(s,u)=\alpha^2\mathcal{P}\mathcal{V}(s)\mathcal{Q}\mathcal{V}(u)\mathcal{P}+\mathcal{O}(\alpha^3,\alpha^2\beta)=\alpha^2\mathcal{P}\mathcal{V}(s)\mathcal{V}(u)\mathcal{P}+\mathcal{O}(\alpha^3,\alpha^2\beta).
\]
Therefore the integrated equation of motion reads
\begin{eqnarray*}
\mathcal{P}\tilde{\rho}(t)=\mathcal{P}\rho(0)&+&\beta\int_{0}^{t}ds\mathcal{P}\mathcal{V}_{12}(s)\tilde{\rho}(s)\\
&+&\alpha^2\int_{0}^{t}ds\int_{0}^sdu\mathcal{P}\mathcal{V}(s)\mathcal{V}(u)\mathcal{P}\tilde{\rho}(u)+\mathcal{O}(\alpha^3,\alpha^2\beta).
\end{eqnarray*}
If we consider small intercoupling $\beta$ such that $\alpha\gtrsim\beta$, we can neglect the higher orders and the weak coupling procedure of this expression (c.f. section \ref{sectionweakcoupling}) will lead to the usual weak coupling generator with the Hamiltonian part $\beta V_{12}$ added. This is in Schr\"odinger picture
\begin{eqnarray}\label{weakcouplingcomposite2}
\frac{d\rho_{12}(t)}{dt}&=&-i[H_{A1}+H_{A2}+\beta V_{12}+\alpha^2H_{\mathrm{LS}},\rho_{12}(t)] \\
&+&\alpha^2\sum_{\omega}\sum_{k,\ell}\gamma_{k\ell}(\omega)\left[ A_\ell(\omega)\rho_{12}(t)A_k^\dagger(\omega)-\frac{1}{2}\{A_k^\dagger(\omega)A_\ell(\omega),\rho_{12}(t)\}\right]\nonumber.
\end{eqnarray}
Compare equations (\ref{weakcouplingcomposite1}) and (\ref{weakcouplingcomposite2}).

Surprisingly, this method may provide good results also for large $\beta$ under some conditions, see \cite{Rivas09}.

\subsubsection{Weak coupling under external driving}\label{sectionweakdriving}
Let us consider now a system $A$ subject to an external time-dependent perturbation $\beta H_\mathrm{ext}(t)$ and weakly coupled to some environment $B$, in such a way that the total Hamiltonian is
\[
H=H_A+\beta H_\mathrm{ext}(t)+H_B+\alpha V.
\]
In order to derive a Markovian master equation for this system we must take into account
of a couple of details. First, since the Hamiltonian is time-dependent the generator of the master equation will also be time-dependent,
\[
\frac{d\rho_A(t)}{dt}=\mathcal{L}_t\rho_A(t),
\]
whose solution defines a family of propagators $\mathcal{E}_{(t_2,t_1)}$ such that
\begin{eqnarray*}
\rho_S(t_2)&=&\mathcal{E}_{(t_2,t_1)}\rho_S(t_1),\\
\mathcal{E}_{(t_3,t_1)}&=&\mathcal{E}_{(t_3,t_2)}\mathcal{E}_{(t_2,t_1)}.
\end{eqnarray*}
This family may be contractive (i.e. Markovian inhomogeneous) or just eventually contractive (see section \ref{sectionEvolutionFamilies}). Secondly, there is an absence of rigorous methods to derive at a Markovian master equation (i.e. contractive family) in the weak coupling limit when the system Hamiltonian is time-dependent, with the exception of adiabatic regimes of external perturbations \cite{DaviesSpohn,AlickiHt}. Fortunately if the $H_\mathrm{ext}(t)$ is periodic in $t$, it is possible to obtain Markovian master equations, even though the complexity of the problem increases.

On one hand, if $\beta$ is small, in the spirit of the previous section 
we can expect that the introduction of $\beta H_\mathrm{ext}(t)$ just in the Hamiltonian part of the evolution would provide a good approximation to the motion,
\begin{eqnarray}\label{weakcouplingexternal1}
\frac{d\rho_A(t)}{dt}&=&-i[H_A+\beta H_\mathrm{ext}(t)+\alpha^2H_{\mathrm{LS}},\rho_A(t)] \\
&+&\alpha^2\sum_{\omega}\sum_{k,\ell}\gamma_{k\ell}(\omega)\left[ A_\ell(\omega)\rho_A(t)A_k^\dagger(\omega)-\frac{1}{2}\{A_k^\dagger(\omega)A_\ell(\omega),\rho_A(t)\}\right]\nonumber,
\end{eqnarray}
which has the form of a Markovian master equation (\ref{diffMarkov}).

On the other hand, for larger $\beta$ the only loophole seems to work in the interaction picture generated by the unitary propagator
\[
U(t_1,t_0)=\mathcal{T}e^{-i\int_{t_0}^{t_1}\left[H_A+\beta H_\mathrm{ext}(t')\right]dt'}.
\]
Taking $t_0=0$ without loss of generality, the time evolution equation for $\tilde{\rho}(t)=U^\dagger(t,0)\rho(t)U(t,0)$ is
\begin{equation}\label{ecPictureRara***}
\frac{d\tilde{\rho}(t)}{dt}=-i\alpha[\tilde{V}(t),\tilde{\rho}(t)].
\end{equation}
Following an analogous procedure which was used in section \ref{sectionweakcoupling} for time-independent generators, one immediately deals with the problem that it is not clear whether there exists something similar to the eigenoperator decomposition of $\tilde{V}(t)=U^\dagger(t,0)VU(t,0)$ as in (\ref{Veigendesc}). Note however that since the dependency of the operator $H_\mathrm{ext}(t)$ with $t$ is periodical, by differentiation of $U(t,0)$ one obtains a differential problem for each $A_k$ in (\ref{Vdesc})
\begin{equation}
\frac{d\tilde{A}_k(t)}{dt}=-i[\tilde{A}_k(t),H_A+\beta H_\mathrm{ext}(t)].
\end{equation}
with periodic terms. This kind of equations can be studied with the well-established Floquet theory (see for example \cite{Chicone,Ince}). Particularly it is possible to predict if its solution is a periodic function. In such a case, the operator in the new picture will have a formal decomposition similar to (\ref{Veigendesc}), $\tilde{A}_k(t)=\sum_\omega A_k(\omega)e^{i\omega t}$, where now $A_k(\omega)$ are some operators which do not necessarily satisfy a concrete eigenvalue problem. However note that the importance of such a decomposition is that the operators $A_k(\omega)$ are time-independent. This allows us to follow a similar procedure to that used for time-independent Hamiltonians with the secular approximation, and we will obtain 
\begin{eqnarray*}
\frac{d\tilde{\rho}_A(t)}{dt}=&-&i[\alpha^2H_{\mathrm{LS}},\tilde{\rho}_A(t)]\\
&+&\alpha^2\sum_{\omega}\sum_{k,\ell}\gamma_{k\ell}(\omega)\left[ A_\ell(\omega)\tilde{\rho}_A(t)A_k^\dagger(\omega)-\frac{1}{2}\{A_k^\dagger(\omega)A_\ell(\omega),\tilde{\rho}_A(t)\}\right],
\end{eqnarray*}
where $H_{\mathrm{LS}}=\sum_\omega\sum_{k,\ell}S_{k\ell}(\omega)A_k^\dagger(\omega)A_\ell(\omega)$ of course. By coming back to Schr\"odinger picture,
\begin{eqnarray*}
\frac{d\tilde{\rho}_A(t)}{dt}=&-&i[H_A+\beta H_\mathrm{ext}(t)+\alpha^2H_{\mathrm{LS}}(t),\tilde{\rho}_A(t)]\\
&+&\alpha^2\sum_{\omega}\sum_{k,\ell}\gamma_{k\ell}(\omega)\left[ A_\ell(\omega,t)\tilde{\rho}_A(t)A_k^\dagger(\omega,t)-\frac{1}{2}\{A_k^\dagger(\omega,t)A_\ell(\omega,t),\tilde{\rho}_A(t)\}\right];
\end{eqnarray*}
here $A_k(\omega,t)=U(t,0)A_k(\omega)U^\dagger(t,0)$. For these manipulations it is very useful to decompose (when possible) the unitary propagator as product of non-commuting exponential operators $U(t,0)=e^{-iH_0t}e^{-i\tilde{H}t}$, where the transformation $e^{iH_0t}H(t)e^{-iH_0t}=\tilde{H}$ removes the explicit time-dependence of the Hamiltonian.

A more systematic approach to this method based on the Floquet theory can be found for instance in the work of H. -P. Breuer and F. Petruccione \cite{BrPe02,BrPe97}, and in \cite{Hanggi99}.

There is a very important thing to note here. As a difference to the usual interaction picture with respect to time-independent Hamiltonians, the change of picture given by a time-dependent Hamiltonian can affect the ``strength'' $\alpha$ of $V$ in a non-trivial way. If there exist some time, $t=t_c$ such that $\|H_\mathrm{ext}(t_c)\|=\infty$, the validity of the weak coupling limit may be jeopardized. For example one may obtain things like ``$1/0$'' inside of $\tilde{V}(t)$.

\section{Microscopic description: non-Markovian case} \label{sectionMicrosnonMarkov}
The microscopic description of non-Markovian dynamics is much more involved than the Markovian one, and developing efficient methods to deal with it is actually an active area of research nowadays. One of the reasons for this difficulty is that the algebraic properties like contraction semigroups and/or evolution families are lost and more complicated structures arise (e.g. evolution families which are only eventually contractive). We have already mentioned the advantages of being consistent with the UDM description if we start from a global product state. According to section \ref{sectionContractions} if the evolution family $\mathcal{E}_{(t,s)}$ describes the evolution from $s$ to $t$ and the initial global product state occurs at time $t_0$, then it must be eventually contractive from that point, i.e. $\|\mathcal{E}_{(t,t_0)}\|\leq1$ for $t\geq t_0$. However it is possible that $\|\mathcal{E}_{(t,s)}\|>1$ for arbitrary $t$ and $s$. Up to date, there is not a clear mathematical characterization of evolution families which are only eventually contractive and this makes it difficult to check whether some concrete model is consistent.

There are non-Markovian cases where the environment is made of a few degrees of freedom (see for example \cite{Dobrovitski,Paladino,Oxtoby}, and references therein). This allows one to make efficient numerical simulations of the whole closed system and after that, just tracing out the environmental degrees of freedom, we obtain the desired evolution of the open system.

However when the environment is large, a numerical simulation is completely inefficient, unless the number of parameters involved in the evolution can be reduced (for example in case of Gaussian states \cite{Rivas09} or using numerical renormalization group procedures \cite{Bulla,Javi}). For remaining situations, the choice of a particular technique really depends on the context. For concreteness, we sketch just a few of them in the following sections.

\subsection{Integro-differential models}
As we have seen in section \ref{sectionNaka-Zwan} the exact evolution of a reduced system can be formally written as an integro-differential equation (\ref{Naka-Zwan}),
\[
\frac{d}{dt}\mathcal{P}\tilde{\rho}(t)=\int_{0}^tdu\mathcal{K}(t,u)\mathcal{P}\tilde{\rho}(u),
\]
with a generally very complicated kernel
\[
\mathcal{K}(t,u)=\mathcal{P}\mathcal{V}(t)\mathcal{G}(t,u)\mathcal{Q}\mathcal{V}(u).
\]
Typically this kernel cannot be written in a closed form and, furthermore, the integro-differential equation is not easily solvable. On the other hand, it is possible to perform a perturbative expansion of the kernel on the strength of $\mathcal{V}(t)$ and proceed further than in the Markovian case, where the series is shortened at the first non-trivial order (for an example of this see \cite{Vacchini-Breuer}). Of course, by making this the complete positivity of the reduced dynamics is not usually conserved; and one expects a similar accuracy for the time-convolutionless method (at the same perturbative order) explained in the next section, which is simpler to solve.

As an alternative to the exact integro-differential equation, several phenomenological approaches which reduce to Markovian evolution in some limits has been proposed. For example, if $\mathcal{L}$ is the generator of a Markovian semigroup, an integro-differential equation can be formulated as
\begin{equation}
\frac{d\rho(t)}{dt}=\int_{0}^{t}k(t-t')\mathcal{L}[\rho(t')],
\end{equation}
where $k(t-t')$ is some function which accounts for ``memory'' effects, and simplifies the complicated Nakajima-Zwanzig kernel. Of course, in the limit $k(t-t')\rightarrow\delta(t-t')$ the Markovian evolution is recovered. The possible choices for this kernel which assure that the solution, $\mathcal{E}_{(t,0)}$, is a UDM have been studied in several works \cite{Barnertt01,Daffer04,Breuer-Vacchini}, being the exponential ansazt $k(t-t')\propto ge^{g(t-t')}$ the most popular.

Another interesting phenomenological integro-differential model is the so-called post-Markovian master equation, proposed by A. Shabani and D. A. Lidar \cite{ShabaniLidar},
\begin{equation}
\frac{d\rho(t)}{dt}=\mathcal{L}\int_{0}^{t}k(t-t')e^{\mathcal{L}(t-t')}[\rho(t')].
\end{equation}
It also approaches the Markovian master equation when $k(t-t')\rightarrow\delta(t-t')$. Further studies about the application of these models and the conditions to get UDMs can be found in \cite{MaPe06,Maniscalco07,Huang-Yi08}.

Apart from these references, the structure of kernels which preserve complete positivity has been analyzed in \cite{Koss-Reb}.

We should note however that, it has been recently questioned whether these phenomenological integro-differential equations reproduce some typical features of non-Markovian dynamics \cite{Mazzola}.

\subsection{Time-convolutionless forms} \label{sectionTCL}
The convolution in integro-differential models is usually an undesirable characteristic because, even when the equation can be formulated, the methods to find its solution are complicated. The idea of the time-convolutionless (TCL) forms consist of removing this convolution from the evolution equation to end up with an ordinary differential equation. Behind this technique lies the property already expressed in section \ref{sectionTemporalContinuity} that a UDM, $\mathcal{E}_{(t,t_0)}$, can be expressed as the solution of the differential equation (\ref{TCLForm})
\[
\frac{d\rho_A(t)}{dt}=\mathcal{L}_t[\rho_A(t)],
\]
with generator $\mathcal{L}_t=\frac{d\mathcal{E}_{(t,t_0)}}{dt}\mathcal{E}^{-1}_{(t,t_0)}$. In general this equation generates an evolution family that is not contractive, i.e. $\|\mathcal{E}_{(t_2,t_1)}\|>1$, for some $t_1$ and $t_2$ and so non-Markovian. However, it generates an eventually contractive family from $t_0$, $\|\mathcal{E}_{(t,t_0)}\|\leq1$, as the map starting from $t_0$, $\mathcal{E}_{(t,t_0)}$, is usually assumed to be a UDM by hypothesis.

To formulate this class of equations from a microscopic model is not simple, although it can been done exactly in some cases such as, for a damped harmonic oscillator \cite{HuPazZhang,Grabert97} or for spontaneous emission of a two level atom \cite{Garraway,BrPe99}. In general, the recipe is resorting to a perturbative approach again. Details of this treatment are carefully explained in the original work \cite{Shibata77} or in some textbooks (see \cite{Fain02,BrPe02} for instance). The key ingredient is to insert the backward unitary operator $\tilde{\mathcal{U}}(s,t)$ where $s\leq t$ defined in (\ref{backwardU}), $\tilde{\rho}(s)=\tilde{\mathcal{U}}(s,t)\tilde{\rho}(t)$, into the formal solution of the $\mathcal{Q}\tilde{\rho}(t)$, Eq. (\ref{Qformalsolution}),
\[
\mathcal{Q}\tilde{\rho}(t)=\mathcal{G}(t,0)\mathcal{Q}\tilde{\rho}(0)+\int_{0}^tds\mathcal{G}(t,s)\mathcal{Q}\mathcal{V}(s)\mathcal{P}\tilde{\mathcal{U}}(s,t)\tilde{\rho}(t).
\]
Now, by introducing the identity $\mathds{1}=\mathcal{P}+\mathcal{Q}$ between $\tilde{\mathcal{U}}(s,t)$ and $\tilde{\rho}(t)$, and after considering a factorized initial condition we obtain
\[
[\mathds{1}-\Upsilon(t)]\mathcal{Q}\tilde{\rho}(t)=\Upsilon(t)\mathcal{P}\tilde{\rho}(t),
\]
here
\[
\Upsilon(t)=\alpha\int_{0}^tds\mathcal{G}(t,s)\mathcal{Q}\mathcal{V}(s)\mathcal{P}\tilde{\mathcal{U}}(s,t),
\]
where we have rewritten $\mathcal{V}$ as $\alpha\mathcal{V}$.

Next, note that $\Upsilon(0)=0$ and $\Upsilon(t)|_{\alpha=0}=0$, so that the operator $[\mathds{1}-\Upsilon(t)]$ may be inverted at least for not too large couplings and/or too large $t$. In such a case we would obtain
\[
\mathcal{Q}\tilde{\rho}(t)=[\mathds{1}-\Upsilon(t)]^{-1}\Upsilon(t)\mathcal{P}\tilde{\rho}(t).
\]
Introducing this in the second member of equation (\ref{InterP2}), the first term vanishes as $\mathcal{P}\mathcal{V}(t)\mathcal{P}=0$, and we arrive at the differential equation
\[
\frac{d\mathcal{P}\rho(t)}{dt}=\hat{\mathcal{L}}_t[\mathcal{P}\rho(t)],
\]
where
\[
\hat{\mathcal{L}}_t=\alpha\mathcal{P}\mathcal{V}(t)[\mathds{1}-\Upsilon(t)]^{-1}\Upsilon(t)\mathcal{P}.
\]
We can formally write the inverse operator as a geometric series
\[
[\mathds{1}-\Upsilon(t)]^{-1}=\sum_{n=0}^\infty \Upsilon^n(t),
\]
and then the generator as
\[
\hat{\mathcal{L}}_t=\alpha\mathcal{P}\mathcal{V}(t)\sum_{n=1}^\infty\Upsilon^n(t)\mathcal{P}.
\]
On the other hand, $\Upsilon^n(t)$ can be expanded in powers of $\alpha$ by introducing the expansions of $\mathcal{G}(t,s)$ and $\tilde{\mathcal{U}}(s,t)$, thus we can construct successive approximations to $\hat{\mathcal{L}}_t$ in powers of $\alpha$,
\[
\hat{\mathcal{L}}_t=\sum_{n=2}\alpha^n\hat{\mathcal{L}}^{(n)}_t.
\]
Here the sum starts from $n=2$ because at lowest order $\Upsilon(t)\simeq\alpha\int_{0}^tds\mathcal{Q}\mathcal{V}(s)\mathcal{P}$, and
\[
\hat{\mathcal{L}}^{(2)}_t=\int_{0}^tds\mathcal{P}\mathcal{V}(t)\mathcal{Q}\mathcal{V}(s)\mathcal{P}=\int_{0}^tds\mathcal{P}\mathcal{V}(t)\mathcal{V}(s)\mathcal{P},
\]
which leads actually to the differentiated version of equation (\ref{weakcoupling2}), as expected at lowest order. After a little bit of algebra, the next term in the expansion turns out to be
\[
\hat{\mathcal{L}}^{(3)}_t=\int_{0}^tds\int_s^tdt_1\mathcal{P}\mathcal{V}(t)\mathcal{V}(t_1)\mathcal{V}(s)\mathcal{P},
\]
and further terms can be calculated in the same fashion.

In addition, D. Chru\'sci\'nski and A. Kossakowski \cite{Chru-Koss10} have recently shown the equivalence between integro-differential equations and TCL equations. They connect both generators by means of a Laplace transformation showing some kind of complementarity between both methods. That is, if the integral kernel is simple, the TCL generator is highly singular, and viceversa.

On the other hand, given a TCL equation and an initial condition at time $t_0$, the problem of knowing whether its solution $\mathcal{E}_{(t,t_0)}$ is a UDM, i.e. a completely positive map, is an open question for the non-Markovian regime. Note that the perturbative expansion at finite order of the TCL generator can generally break the complete positivity (in the same fashion as a finite perturbative expansion of a Hamiltonian dynamics breaks the unitary character). So the classification of TCL generators which are just eventually contractive from a given time $t_0$ is an important open problem in the field of dynamics of open quantum systems. Partial solutions have been given in the qubit case \cite{Hall} and in the case of commutative generators $\mathcal{L}_t$ \cite{Commutative,Chru-KossPre}. For this latter one $[\mathcal{L}_{t_1},\mathcal{L}_{t_2}]=0$, for all $t_1$ and $t_2$, so that the time-ordering operator has no effect and the solution is given by
\[
\mathcal{E}_{(t,t_0)}=e^{\int_{t_0}^t\mathcal{L}_{t'}dt'}.
\]
Therefore, certainly $\mathcal{E}_{(t,t_0)}$ is completely positive if and only if $\int_{t_0}^t\mathcal{L}_{t'}dt'$ has the standard form (\ref{generatorMarkov}) for every final $t$.

\subsection{Dynamical coarse graining method} \label{sectionDCG}
The dynamical coarse-graining method was recently proposed by G. Schaller and T. Brandes \cite{SchallerBrandes1,SchallerBrandes2}. Thiswas suggested as away to avoid the secular approximation under weak coupling while keeping the complete positivity in the dynamics. Later on, this method has been used in some other works \cite{BenattiDCG1,BenattiDCG2}. Here we briefly explain the main idea and results.

Consider the total Hamiltonian (\ref{totalH}) and the usual product initial state $\rho(0)=\rho_A(0)\otimes\rho_B$, where $\rho_B$ is a stationary state of the environment. In the interaction picture the reduced state at time $t$ is
\begin{equation}\label{coarsegrained1}
\tilde{\rho}_A(t)=\tr_B\left[U(t,0)\rho_A(0)\otimes\rho_B U^\dagger(t,0)\right],
\end{equation}
where
\[
U(t,0)=\mathcal{T}e^{-i\alpha\int_0^t\tilde{V}(t')dt'},
\]
is the unitary propagator, and $\alpha$ again accounts for the strength of the interaction. Next we introduce the expansion of the propagator up to the first non-trivial order in equation (\ref{coarsegrained1})
\[
\tilde{\rho}_A(t)=\rho_A(0)-\frac{\alpha^2}{2}\mathcal{T}\int_0^tdt_1\int_0^tdt_2\tr_B\left[\tilde{V}(t_1),\left[\tilde{V}(t_2),\rho_A(0)\otimes\rho_B\right]\right]+\mathcal{O}(\alpha^3),
\]
here we have assumed again that the first order vanishes $\tr[\tilde{V}(t)\rho_B]=0$ like in (\ref{ceroprimermomento}). Let us take a closer look to structure of the double time-ordered integral
\begin{eqnarray*}
&\mathcal{T}&\int_0^tdt_1\int_0^tdt_2\tr_B\left[\tilde{V}(t_1),\left[\tilde{V}(t_2),\rho_A(0)\otimes\rho_B\right]\right]\\
&=&\int_0^tdt_1\int_0^tdt_2\theta(t_1-t_2)\tr_B\left[\tilde{V}(t_1),\left[\tilde{V}(t_2),\rho_A(0)\otimes\rho_B\right]\right]\\
&+&\int_0^tdt_1\int_0^tdt_2\theta(t_2-t_1)\tr_B\left[\tilde{V}(t_2),\left[\tilde{V}(t_1),\rho_A(0)\otimes\rho_B\right]\right]\\
&\equiv&\mathcal{L}^{(2)}_1[\rho_A(0)]+\mathcal{L}^{(2)}_2[\rho_A(0)]+\mathcal{L}^{(2)}_3[\rho_A(0)],
\end{eqnarray*}
where by expanding the double commutators we have found three kind of terms, $\mathcal{L}^{(2)}_1$, $\mathcal{L}^{(2)}_2$ and $\mathcal{L}^{(2)}_3$. The first one appears twice, it is
\begin{eqnarray*}
\frac{\mathcal{L}^{(2)}_1}{2}[\rho_A(0)]&=&-\int_0^tdt_1\int_0^tdt_2\theta(t_1-t_2)\tr_B\left[\tilde{V}(t_1)\rho_A(0)\otimes\rho_B\tilde{V}(t_2)\right]\\
&\mbox{}&+\theta(t_2-t_1)\tr_B\left[\tilde{V}(t_1)\rho_A(0)\otimes\rho_B\tilde{V}(t_2)\right]\\
&=&-\int_0^tdt_1\int_0^tdt_2[\theta(t_1-t_2)+\theta(t_2-t_1)]\\
&\mbox{}&\tr_B\left[\tilde{V}(t_1)\rho_A(0)\otimes\rho_B\tilde{V}(t_2)\right]\\
&=&-\int_0^tdt_1\int_0^tdt_2\tr_B\left[\tilde{V}(t_1)\rho_A(0)\otimes\rho_B\tilde{V}(t_2)\right],
\end{eqnarray*}
because obviously the step functions cancel each other. Taking in account the other analogous term (that is the corresponding to interchange $t_1\leftrightarrow t_2$ in the double commutators) we can write the total contribution as
\[
\mathcal{L}^{(2)}_1[\rho_A(0)]=-2\tr_B\left[\Lambda\rho_A(0)\otimes\rho_B\Lambda\right],
\]
where $\Lambda=\int_0^t\tilde{V}(t')dt'$.

The second term goes like
\begin{eqnarray*}
\mathcal{L}^{(2)}_2[\rho_A(0)]&=&\int_0^tdt_1\int_0^tdt_2\theta(t_1-t_2)\tr_B\left[\tilde{V}(t_1)\tilde{V}(t_2)\rho_A(0)\otimes\rho_B\right]\\
&\mbox{}&+\theta(t_2-t_1)\tr_B\left[\tilde{V}(t_2)\tilde{V}(t_1)\rho_A(0)\otimes\rho_B\right]\\
&=&\int_0^tdt_1\int_0^tdt_2[\theta(t_1-t_2)+\theta(t_2-t_1)]\\
&\mbox{}&\tr_B\left[\tilde{V}(t_2)\tilde{V}(t_1)\rho_A(0)\otimes\rho_B\right]\\
&\mbox{}&+\theta(t_1-t_2)\tr_B\left\{[\tilde{V}(t_1),\tilde{V}(t_2)]\rho_A(0)\otimes\rho_B\right\}\\
&=&\tr_B\left[\Lambda^2\rho_A(0)\otimes\rho_B\right]\\
&\mbox{}&+\int_0^tdt_1\int_0^tdt_2\theta(t_1-t_2)\tr_B\left\{[\tilde{V}(t_1),\tilde{V}(t_2)]\rho_A(0)\otimes\rho_B\right\}.
\end{eqnarray*}
And similarly the remaining term can be expressed as
\begin{eqnarray*}
\mathcal{L}^{(2)}_3[\rho_A(0)]&=&\tr_B\left[\rho_A(0)\otimes\rho_B\Lambda^2\right]\\
&\mbox{}&-\int_0^tdt_1\int_0^tdt_2\theta(t_1-t_2)\tr_B\left\{\rho_A(0)\otimes\rho_B[\tilde{V}(t_1),\tilde{V}(t_2)]\right\}.
\end{eqnarray*}
Thus, everything together gives $\tilde{\rho}_A(t)\equiv\rho_A(0)+\mathcal{L}^t[\rho_A(0)]+\mathcal{O}(\alpha^3)$ with
\begin{eqnarray*}
\mathcal{L}^t[\rho_A(0)]&=&-\frac{\alpha^2}{2}\left(\mathcal{L}^{(2)}_1[\rho_A(0)]+\mathcal{L}^{(2)}_2[\rho_A(0)]+\mathcal{L}^{(2)}_3[\rho_A(0)]\right)\\
&\equiv&-i[H_{\mathrm{LS}}^t,\rho_A(0)]+\mathcal{D}^t[\rho_A(0)].
\end{eqnarray*}
Here the self-adjoint operator $H_{\mathrm{LS}}^t$ is
\[
H_{\mathrm{LS}}^t=\frac{\alpha^2}{2i}\int_0^tdt_1\int_0^tdt_2\theta(t_1-t_2)\tr_E\left\{[\tilde{V}(t_1),\tilde{V}(t_2)]\rho_B\right\},
\]
and the ``dissipative'' part is given by
\[
\mathcal{D}^t[\rho_A(0)]=\alpha^2\tr_B\left[\Lambda\rho_A(0)\otimes\rho_B\Lambda-\frac{1}{2}\left\{\Lambda^2,\rho_A(0)\otimes\rho_B\right\}\right].
\]
By using the spectral decomposition of $\rho_B$ it is immediate to check that for any fixed $t$, $\mathcal{L}^t$ has the Kossakowski-Lindblad form (\ref{generatorMarkov}).

Since at first order
\begin{equation} \label{coarsegrained2}
\tilde{\rho}_A(t)=(\mathds{1}+\mathcal{L}^t)\rho_A(0)+\mathcal{O}(\alpha^3)\simeq e^{\mathcal{L}^t}\rho_A(0),
\end{equation}
the ``dynamical coarse graining method'' consists of introducing a ``coarse graining time'' $\tau$ and defining a family of generators of semigroups by
\[
\bar{\mathcal{L}}^\tau\equiv\frac{\mathcal{L}^\tau}{\tau}.
\]
And the corresponding family of differential equations
\[
\tilde{\rho}_A^\tau(t)=\bar{\mathcal{L}}^\tau\tilde{\rho}_A^\tau(t),
\]
which solution
\begin{equation}\label{coarsegrained3}
\tilde{\rho}_A^\tau(t)=e^{t\bar{\mathcal{L}}^\tau}\rho_A(0),
\end{equation}
will obviously coincide with (\ref{coarsegrained3}) up to second order when taking $\tau=t$, this is $\rho_A^{\tau=t}(t)=\rho_A(t)$.

Therefore we have found a family of semigroups (\ref{coarsegrained2}) which have the Kossakowski-Lindblad form for every $\tau$ and give us the solution to the evolution for a particular value of $\tau$ ($\tau=t$). This implies that complete positivity is preserved as we desire. In addition, one may want to write this evolution as a TCL equation, i.e. to find the generator $\mathcal{L}_t$ such that the solution of
\[
\frac{d}{dt}\rho_A(t)=\mathcal{L}_t[\rho_A(t)]
\]
fulfills $\rho_A(t)=\rho_A^t(t)$. That is given as in (\ref{TCLForm}),
\[
\mathcal{L}_t=\left[\frac{d}{dt}e^{\bar{\mathcal{L}}^tt}\right]e^{-\bar{\mathcal{L}}^tt}=\left[\frac{d}{dt}e^{\mathcal{L}^t}\right]e^{-\mathcal{L}^t}.
\]
Particularly for large times $t\rightarrow\infty$, this generator $\mathcal{L}_t$ approach the usual weak coupling generator (\ref{weakcoupling4}) (in the interaction picture). This is expected because the upper limit of the integrals in the perturbative approach will approach to infinity and the non-secular terms $e^{i(\omega'-\omega)t}$ will disappear for large $t$ (because of proposition \ref{propAlfredo}). This is actually equivalent to what we obtain by taking the limit $\alpha\rightarrow0$ in the rescaled time. Details are found in \cite{SchallerBrandes1}.

\section{Conclusion}\label{conclusions}
We have reviewed in some detail the mathematical structure and the main physical features of the dynamics of open quantum systems, and presented the general properties of the dynamical maps underlying open system dynamics. In addition, we have analyzed some mathematical properties of quantum Markovian processes, such as their differential form or their steady states. After that we have focused our attention on the microscopic derivations of dynamical equations. First of all we have analyzed the Markovian case, recalling the weak and singular coupling procedures, as well as some of their properties; and secondly, we have discussed briefly some methods which go beyond Markovian evolutions.

Several interesting topics have not been tackled in this article. Examples are influence functional techniques \cite{BrPe02,Brandes,Weiss08,FeymannVernon,CaldLegg}, stochastic Schr\"odinger equations \cite{Carmichael99,GardinerZoller04,Gisin,Diosi,PlenioKnight,Ines1}, or important results such as the quantum regression theorem, which deals with the dynamics of multi-time time correlation functions \cite{Carmichael99,Brandes,GardinerZoller04,Lax,Ines2}. Moreover, recently several groups have addressed the important problem of introducing a measure to quantify the non-Markovian character of quantum evolutions \cite{Cubbit,BreuerFinlandeses,NoMarkovianidad,CPSun,UshaDevi,EAndersson,chinoemail,Darek-Andrzej-Angel}. However, despite the introductory nature of this work, we hope that it will be helpful to clarify certain concepts and
derivation procedures of dynamical equations in open quantum systems, providing a useful basis to study further results as well as helping to unify the concepts and methods employed by different communities.

\subsection*{Acknowledgments}
We thank Ram\'on Aguado, Fatih Bayazit, Tobias Brandes, Heinz-Peter Breuer, Daniel Burgarth, Alex Chin, Dariusz Chru\'sci\'nski, Animesh Datta, Alberto Galindo, Pinja Haikka, Andrzej Kossakowski, Alfredo Luis, Sabrina Maniscalco, Laura Mazzola, Jyrki Piilo, Miguel \'Angel Rodr\'iguez, David Salgado, Robert Silbey, Herbert Spohn, Francesco Ticozzi, Shashank Virmani and Michael Wolf for fruitful discussions. And we are specially grateful to Bruce Christianson, Adolfo del Campo, In\'es de Vega, Alexandra Liguori, and Martin Plenio for careful reading and further comments. This work was supported by the STREP project CORNER, the Integrated project Q-ESSENCE, the project QUITEMAD S2009-ESP-1594 of the Consejer\'{\i}a de Educaci\'{o}n de la Comunidad de Madrid and MICINN FIS2009-10061.


\begin{thebibliography}{00}
%
\bibitem{Louisell73} W. H. Louisell, \textit{Quantum Statistical Properties of Radiation}, John Wiley \& Sons Inc, New York, 1973.
%
\bibitem{Agarwal} G. S. Agarwal, Master equation methods in quantum optics, in: E. Wolf (ed.), \textit{Progress in Optics Vol. 11}, Elsevier, Amsterdam, 1973, 1.
%
\bibitem{Haake} F. Haake, Statistical treatment of open systems by generalized master equations. Springer Tracts in Modern Physics \textbf{66} (1973), 98.
%
\bibitem{Daviesbook76} E. B. Davies, \textit{Quantum Theory of Open Systems}, Academic Press, London, 1976.
%
\bibitem{RevKoss} V. Gorini, A. Frigerio, M. Verri, A. Kossakowski and E. C. G. Sudarshan, Rep. Math. Phys. \textbf{13} (1978), 149.
%
\bibitem{Spohn80} H. Spohn, Rev. Mod. Phys. \textbf{53} (1980), 569.
%
\bibitem{Kraus83} K. Kraus, \textit{States, Effects, and Operations Fundamental Notions of Quantum Theory}, Lect. Notes Phys. \textbf{190}, Springer, Berlin, 1983.
%
\bibitem{Leggett87} A. J. Leggett, S. Chakravarty, A. T. Dorsey, M. P. A. Fisher, A. Garg and W. Zwerger, Rev. Mod. Phys. \textbf{59} (1987), 1.
%
\bibitem{AlickiLendi87} R. Alicki and K. Lendi, \textit{Quantum Dynamical Semigroups and Applications},
Lect. Notes Phys. \textbf{286}, Springer, Berlin, 1987.
%
\bibitem{Cohen92} C. Cohen-Tannoudji, J. Dupont-Roc and G. Grynberg, \textit{Atom-Photon Interactions}, John Wiley \& Sons, New York, 1992.
%
\bibitem{Carmichael93} H. Carmichael, \textit{An Open Systems Approach to Quantum Optics}, Springer, Berlin, 1993.
%
\bibitem{Bohn96} A. Bohm, H. -D. Doebner and P. Kielanowski (eds.), \textit{Irreversibility and causality. Semigroups and rigged Hilbert spaces}, Lect. Notes Phys. \textbf{504}, Springer, Berlin, 1996.
%
\bibitem{Blum96} K. Blum, \textit{Density Matrix Theory and Applications}, Plenum Press, New York, 1996.
%
\bibitem{Gardiner97} C. W. Gardiner, \textit{Handbook of Stochastic Methods}, Springer, Berlin, 1997.
%
\bibitem{Carmichael99} H. J. Carmichael, \textit{Statistical Methods in Quantum Optics} (2 vols), Springer, Berlin, 1999.
%
\bibitem{Zwanzig01} R. Zwanzig, \textit{Nonequilibrium Statistical Mechanics}, Oxford University Press, New York, 2001.
%
\bibitem{Puri01} R. R. Puri, \textit{Mathematical Methods of Quantum Optics}, Springer, Berlin, 2001.
%
\bibitem{AlickiFanes01} R. Alicki and M. Fannes, \textit{Quantum Dynamical Systems}, Oxford University Press, New York, 2001.
%
\bibitem{Lindblad01} G. Lindblad, \textit{Non-Equilibrium Entropy and Irreversibility}, Reidel Publishing Company, Dordrecht, 2001.
%
\bibitem{Fain02} B. Fain, \textit{Irreversibilities in Quantum Mechanics}, Kluwer Academic Publishers, New York, 2002.
%
\bibitem{BrPe02} H. -P. Breuer and F. Petruccione, \textit{The Theory of Open Quantum Systems}, Oxford University Press, Oxford, 2002.
%
\bibitem{Giovanna02} B.-G. Englert and G. Morigi, Five Lectures on Dissipative Master Equations, in: A. Buchleitner and K. Hornberger, (eds.), \textit{Coherent Evolution in Noisy Environments}, Lect. Notes Phys. \textbf{611}, 55, Springer, Berlin, 2002.
%
\bibitem{Brandes} T. Brandes, Quantum Dissipation http://wwwitp.physik.tu-berlin.de/brandes/public\_html/publications/notes/notes.html (2003).
%
\bibitem{Benatti03} F. Benatti and R. Floreanini (eds.), \textit{Irreversible Quantum Dynamics}, Lect. Notes Phys. \textbf{622}, Springer, Berlin, 2003.
%
\bibitem{Decoherence03} E. Joos, H. D. Zeh, C. Kiefer, D. J. W. Giulini, J. Kupsch, I. -O. Stamatescu, \textit{Decoherence and the Appearance of a Classical World in Quantum Theory}, Springer, Berlin, 2003.
%
\bibitem{GardinerZoller04} C. W. Gardiner and P. Zoller, \textit{Quantum Noise}, Springer, Berlin, 2004.
%
\bibitem{quimicos04} V. May and O. Kuhn, \textit{Charge and Energy Transfer Dynamics in
Molecular Systems}, Wiley-VCH, Weinheim, 2004.
%
\bibitem{Benatti05} F. Benatti and R. Floreanini, Int. J. Mod. Phys. B \textbf{19} (2005), 3063.
%
\bibitem{Franceses06} S. Attal, A. Joye and C. -A. Pillet (eds.), \textit{Open Quantum Systems} (3 vols.), Springer, Berlin, 2006.
%
\bibitem{quimicos06} A. Nitzan, \textit{Chemical Dynamics in Condensed Phases: Relaxation, Transfer, and Reactions in Condensed Molecular Systems}, Oxford University Press, New York, 2006.
%
\bibitem{Tanimura} Y. Tanimura, J. Phys. Soc. Jpn. \textbf{75} (2006), 082001.
%
\bibitem{DecoherenceBrain07} M. A. Schlosshauer, \textit{Decoherence and the Quantum-to-Classical Transition}, Springer, Berlin, 2007.
%
\bibitem{Weiss08} U. Weiss, \textit{Quantum Dissipative Systems}, World Scientific, Singapore, 2008.
%
\bibitem{ModernCohen08} S. Kryszewski and J. Czechowska-Kryszk, Master equation - tutorial approach, arXiv:0801.1757 (2008).
%
\bibitem{Hornberger09} K. Hornberger, Introduction to Decoherence Theory, in: A. Buchleitner,
C. Viviescas and M. Tiersch, (eds.), \textit{Entanglement and Decoherence}, Lect. Notes Phys. \textbf{768}, 221, Springer, Berlin, 2009.
%
\bibitem{Rivas09} A. Rivas, A. D. K. Plato, S. F. Huelga and M. B. Plenio, New J. Phys. \textbf{12} (2010), 113032.
%
\bibitem{Hu12} C. H. Fleming and B. L. Hu, Ann. Phys. (N.Y.) \textbf{327} (2012), 1238.
%
\bibitem{NC00} M. A. Nielsen and I. L. Chuang, \textit{Quantum Computation and Quantum Information}, Cambridge University Press, Cambridge, 2000.
%
\bibitem{Kreyszig78} E. Kreyszig, \textit{Introductory Functional Analysis with Applications}, Wiley, New York, 1978.
%
\bibitem{reedsimon1} M. Reed and B. Simon, \textit{Methods of Modern Mathematical Physics I}, Academic Press, San Diego, 1980.
%
\bibitem{Boccara} N. Bocccara, \textit{Functional Analysis: An Introduction for Physicists}, Academic Press, 1990.
%
\bibitem{Galindoproblemas} A. Galindo and P. Pascual, \textit{Problemas de Mec\'anica Cu\'antica}, Eudema, Madrid, 1989.
%
\bibitem{Spivak} M. Spivak, \textit{Calculus}, Cambridge University Press, Cambridge, 1994.
%
\bibitem{Yosida} K. Yosida, \textit{Functional Analysis}, Springer, Berlin, 1980.
%
\bibitem{reedsimon2} M. Reed and B. Simon, \textit{Methods of Modern Mathematical Physics II}, Academic Press, San Diego, 1975.
%
\bibitem{EngelNagel} K. -J. Engel and R. Nagel, \textit{One-Parameter Semigroups for Linear Evolution Equations}, Springer, New-York, 2000.
%
\bibitem{EvFamilynoDiff} The most simple example is in one dimension, take nowhere differentiable function $f(t)$ and define $T_{(t,s)}=f(t)/f(s)$ for $f(s)\neq0$.
%
\bibitem{Schwabl} F. Schwabl, \textit{Quantum Mechanics}, Springer, Berlin, 2007.
%
\bibitem{Sc26} E. Schr\"odinger, Phys. Rev. \textbf{28} (1926), 1049.
%
\bibitem{OpenQuantumSystemRemark} One point to stress here is that a quantum system that interchanges information with a classical one in a controllable way (such as a pulse of a laser field on an atom) is also considered isolated.
%
\bibitem{Zeilinger} R. G\"ahler, A. G. Klein and A. Zeilinger, Phys. Rev. A \textbf{23} (1981), 1611.
%
\bibitem{Weinberg} S. Weinberg, Phys. Rev. Lett. \textbf{62} (1989), 485.
%
\bibitem{Wineland} J. J. Bollinger, D. J. Heinzen, Wayne M. Itano, S. L. Gilbert and D. J. Wineland, Phys. Rev. Lett. \textbf{63} (1989) , 1031.
%
\bibitem{GP90} A. Galindo and P. Pascual, \textit{Quantum Mechanics} (2 vols.), Springer, Berlin, 1990.
%
\bibitem{unitaryLiouville} It can be shown it is the only linear (isomorphic) isometry of trace-class operators. Actually the proof of theorem \ref{WignerUDM} is similar.
%
\bibitem{ShashThesis} S. Virmani, \textit{Entanglement Quantification \& Local Discrimination}, PhD Thesis, Imperial College London, 2001.
%
\bibitem{SB01} P. {\v{S}}telmachovi\v{c} and V. Bu\v{z}ek, Phys. Rev. A \textbf{64} (2001), 062106.
%
\bibitem{FactorizedDynamics} Leaving aside the case of factorized global dynamics $U(t_1,t_0)=U_A(t_1,t_0)\otimes U_B(t_1,t_0)$, see: D. Salgado and J. L. S\'anchez-G\'omez, Comment on ``{D}ynamics of open quantum systems initially entangled with environment: Beyond the Kraus representation'' 	 arXiv:quant-ph/0211164 (2002); H. Hayashi, G. Kimura and Y. Ota, Phs. Rev. A \textbf{67} (2003), 062109.
%
\bibitem{Peres} A. Peres, \textit{Quantum Theory: Concepts and Methods}, Kluwer Academic Publishers, Dordrecht, 1995.
%
\bibitem{Ballentine} L. E. Ballentine, \textit{Quantum Mechanics: A Modern Development}, World Scientific, Singapore, 1998.
%
\bibitem{JSS04} T. F. Jordan, A. Shaji and E. C. G. Sudarshan, Phys. Rev. A \textbf{70} (2004), 052110.
%
\bibitem{SS05} A. Shaji and E. C. G. Sudarshan, Phys. Lett. A \textbf{341} (2005), 48.
%
\bibitem{Sh05} A. Shaji, \textit{Dynamics of initially entangled open quantum systems}, PhD Thesis, University of Texas at Austin, 2005.
%
\bibitem{CTZ08} H. A. Carteret, D. R. Terno and K. {\.{Z}}yczkowski, Phys. Rev. A \textbf{77} (2008), 042113.
%
\bibitem{Modi10} K. Modi and E. C. G. Sudarshan, Phys. Rev. A \textbf{81} (2010), 052119.
%
\bibitem{CesarKavanAlan} C. A. Rodríguez-Rosario, K. Modi and A. Aspuru-Guzik, Phys. Rev. A \textbf{81}, 012313 (2010).
%
\bibitem{SSF04} D. Salgado, J. L. S\'anchez-G\'omez and M. Ferrero, Phys. Rev. A \textbf{70} (2004), 054102.
%
\bibitem{TKOJM04} D. M. Tong, L. C. Kwek, C. H. Oh, J. -L. Chen and L. Ma, Phys. Rev. A \textbf{69} (2004), 054102.
%
\bibitem{Pe94} P. Pechukas, Phys. Rev. Lett. \textbf{73} (1994), 1060.
%
\bibitem{vWL00} A. J. van Wonderen and K. Lendi, J. Phys. A: Math. Gen. \textbf{33} (2000), 5557.
%
\bibitem{RMKSS08} C. A. Rodr\'iguez, K. Modi, A. Kuah, A. Shaji and E. C. G. Sudarshan, J. Phys. A: Math. Gen. \textbf{41} (2008), 205301.
%
\bibitem{St55} W. Stinespring, Proc. Am. Math. Soc. \textbf{6} (1955), 211.
%
\bibitem{Kossakowski1} A. Kossakowski, Rep. Math. Phys. \textbf{3} (1972), 247.
%
\bibitem{Kossakowski2} A. Kossakowski, Bull. Acad. Polon. Sci. Math. \textbf{20} (1972), 1021.
%
\bibitem{Ruskai} M. B. Ruskai, Rev. Math. Phys. \textbf{6} (1994), 1147.
%
\bibitem{compositionCP} Note that the composition two completely positive maps is also a completely positive map. However the composition of two arbitrary operators can be completely positive regardless of whether they are individually completely positive or not.
%
\bibitem{canonicalDM} C. A. Rodriguez-Rosario and E. C. G. Sudarshan, Non-Markovian Open Quantum Systems, arXiv:0803.1183 (2008).
%
\bibitem{Norris97} J. R. Norris, \textit{Markov Chains}, Cambridge University Press, Cambridge, 1997.
%
\bibitem{Ethier05} S. N. Ethier and T. G. Kurtz, \textit{Markov Processes:
Characterization and Convergence}, Wiley, New Jersey, 2005.
%
\bibitem{SufficiencyGernot} Another proof of sufficiency can be founded in Appendix E of \cite{SchallerBrandes1}.
%
\bibitem{GoKoSh76} V. Gorini, A. Kossakowski and E. C. G. Sudarshan, J. Math. Phys. \textbf{17} (1976), 821.
%
\bibitem{Lindblad76} G. Lindblad, Commun. Math. Phys. \textbf{48} (1976), 119.
%
\bibitem{Alexandra08} F. Benatti, A. M. Liguori and A. Nagy, J. Math. Phys. \textbf{49} (2008), 042103.
%
\bibitem{HornJonson} R. A. Horn and C. R. Johnson, \textit{Topics in Matrix Analysis}, Cambridge University Press, Cambridge, 1991.
%
\bibitem{Schirmer-Wang} S. G. Schirmer and X. Wang, Phys. Rev. A \textbf{81}, 062306 (2010).
%
\bibitem{Bougarth} D. Burgarth and V. Giovannetti, New J. Phys. \textbf{9} (2007), 150.
%
\bibitem{Davies70} E. B. Davies, Commun. Math. Phys. \textbf{19} (1970), 83.
%
\bibitem{Spohn76} H. Spohn, Rep. Math. Phys. \textbf{10} (1976), 189.
%
\bibitem{Spohn77} H. Spohn, Lett. Math. Phys. \textbf{2} (1977), 33.
%
\bibitem{Evans77} D. E. Evans, Commun. Math. Phys. \textbf{54} (1977), 293.
%
\bibitem{Frigerio77} A. Frigerio, Lett. Math. Phys. \textbf{2} (1977), 79.
%
\bibitem{Frigerio78} A. Frigerio, Commun. Math. Phys. \textbf{63} (1978), 269.
%
\bibitem{Frigerio-Spohn} A. Frigerio and H. Spohn, Stationary states of quantum dynamical semigroups
and applications, in: L. Accardi, V. Gorini, G. Paravicini, (eds.), Proceedings of Mathematical Problems in the
Theory of Quantum Irreversible Processes, pp. 115-135, Laboratoria di Cibernetica del CNR, 1978.
%
\bibitem{Baumgartner1} B. Baumgartner, H. Narnhofer and W. Thirring, J. Phys. A: Math. Theor. \textbf{41} (2008), 065201.
%
\bibitem{Baumgartner2} B. Baumgartner and H. Narnhofer, J. Phys. A: Math. Theor. \textbf{41} (2008), 395303.
%
\bibitem{Ticozzi1} F. Ticozzi and L. Viola, IEEE T. Automat. Contr. \textbf{53} (2008), 2048.
%
\bibitem{Ticozzi2} F. Ticozzi and L. Viola, Automatica \textbf{45} (2009), 2002.
%
\bibitem{Ticozzi3} F. Ticozzi, S. G. Schirmer and X. Wang, IEEE T. Automat. Contr. \textbf{55} (2010), 2901.
%
\bibitem{Ticozzi4} F. Ticozzi, R. Lucchese, P. Cappellaro, L. Viola, Hamiltonian Control of Quantum Dynamical Semigroups: Stabilization and Convergence Speed, arXiv:1101.2452 (2011).
%
\bibitem{JansenBoon} L. Jansen and H. Boon, \textit{Theory of Finite Groups. Application in Physics}, North-Holland, Amsterdam, 1967.
%
\bibitem{Nakajima} S. Nakajima, Progr. Theor. Phys. \textbf{20} (1958), 984.
%
\bibitem{Zwanzig} R. Zwanzig, J. Chem. Phys. \textbf{33} (1960), 1338.
%
\bibitem{Chicone} C. Chicone, \textit{Ordinary Differential Equations with Applications}, Springer, New York, 2006.
%
\bibitem{Ince} E. L. Ince, \textit{Ordinary Differential Equations}, Dover, New York, 1956.
%
\bibitem{Whitney08} R. S. Whitney, J. Phys. A: Math. Theor. \textbf{41} (2008), 175304.
%
\bibitem{Davies1} E. B. Davies, Commun. Math. Phys. \textbf{39} (1974), 91.
%
\bibitem{Davies2} E. B. Davies, Math. Ann. \textbf{219} (1976), 147.
%
\bibitem{GoKo76} V. Gorini and A. Kossakowski, J. Math. Phys. \textbf{17} (1976), 1298.
%
\bibitem{FrGo76} A. Frigerio and V. Gorini, J. Math. Phys. \textbf{17} (1976), 2123.
%
\bibitem{FrNoVe76} A. Frigerio, C. Novellone and M. Verri, Rep. Math. Phys. \textbf{12} (1977), 279.
%
\bibitem{Ziman1} M. Ziman, P. {\v{S}}telmachovi\v{c} and V. Bu\v{z}ek, Open Syst. Inf. Dyn. \textbf{12} (2005), 81.
%
\bibitem{Ziman2} M. Ziman and V. Bu\v{z}ek, Phys. Rev. A \textbf{72} (2005), 022110.
%
\bibitem{Taj} D. Taj and F. Rossi, Phys. Rev. A \textbf{78} (2008), 052113.
%
\bibitem{Schwartz} L. Schwartz, \textit{Th\'eorie des Distributions}, Hermann, Paris, 1973.
%
\bibitem{Vladimirov} V. S. Vladimirov, \textit{Methods of the Theory of Generalized Functions}, Taylor \& Francis, London, 2002.
%
\bibitem{Zueco} R. Doll, D. Zueco, M. Wubs, S. Kohler and P. H\"anggi, Chem. Phys. \textbf{347}, 243 (2008).
%
\bibitem{Arfken} G. B. Arfken and H. J. Weber, \textit{Mathematical Methods for Physicists}, Academic Press, San Diego, 1985.
%
\bibitem{DumckeandSpohn} R. D\"umcke and H. Spohn, Z. Phys. \textbf{B34} (1979), 419.
%
\bibitem{ZhaoChen02} Y. Zhao and G. H. Chen, Phys. Rev. E \textbf{65} (2002), 056120.
%
\bibitem{Silbey} A. Su\'arez, R. Silbey and I. Oppenheim, J. Chem. Phys \textbf{97} (1992), 5101.
%
\bibitem{Gaspard} P. Gaspard and M. Nagaoka, J. Chem. Phys. \textbf{111} (1999), 5668.
%
\bibitem{BenattiSlipagge1}  F. Benatti, R. Floreanini and M. Piani, Phys. Rev. A \textbf{67} (2003), 042110.
%
\bibitem{BenattiSlipagge2} F. Benatti, R. Floreanini and S. Breteaux, Laser Physics \textbf{16} (2006), 1395.
%
\bibitem{KMSpapers} R. Kubo, J. Phys. Soc. Jpn. \textbf{12} (1957), 570; P. C. Martin and J. Schwinger, Phys. Rev. \textbf{115} (1959), 1342.
%
\bibitem{CavesMilburn}  C. M. Caves and G. J. Milburn, Phys. Rev. A \textbf{36} (1987), 5543.
%
\bibitem{DaviesSpohn} E. B. Davies and H. Spohn, J. Stat. Phys. \textbf{19} (1978), 511.
%
\bibitem{AlickiHt} R. Alicki, J. Phys. A: Math. Gen. \textbf{12} (1979), L103.
%
\bibitem{BrPe97} H. -P. Breuer and F. Petruccione, Phys. Rev. A \textbf{55} (1997), 3101.
%
\bibitem{Hanggi99} S. Kohler, T. Dittrich and P. H\"anggi, Phys. Rev. E \textbf{55} (1999), 300.
%
\bibitem{Dobrovitski} V. V. Dobrovitski and H. A. De Raedt, Phys. Rev. E \textbf{67} (2003), 056702.
%
\bibitem{Paladino} E. Paladino, M. Sassetti, G. Falci, U. Weiss, Phys. Rev. B \textbf{77} (2008), 041303(R).
%
\bibitem{Oxtoby} N. P. Oxtoby, A. Rivas, S. F. Huelga and R. Fazio, New J. Phys. \textbf{11} (2009), 063028.
%
\bibitem{Bulla} R. Bulla, H. -J. Lee, N. -H. Tong and M. Vojta, Phys. Rev. B \textbf{71} (2005), 045122; R. Bulla, T. Costi and T. Pruschke, Rev. Mod. Phys. \textbf{80} (2008), 395.
%
\bibitem{Javi} J. Prior, A. W. Chin, S. F. Huelga and M. B. Plenio, Phys. Rev. Lett. \textbf{105} (2010), 050404.
%
\bibitem{Vacchini-Breuer} B. Vacchini and H. -P. Breuer, Phys. Rev. A \textbf{81} (2010), 042103.
%
\bibitem{Barnertt01} S. M. Barnett and S. Stenholm, Phys. Rev. A \textbf{64} (2001), 033808.
%
\bibitem{Daffer04} S. Daffer, K. W\'odkiewicz, J. D. Cresser and J. K. McIver, Phys. Rev. A \textbf{70} (2004), 010304.
%
\bibitem{Breuer-Vacchini} H. -P. Breuer and B. Vacchini, Phys. Rev. Lett. \textbf{101} (2008), 140402.
%
\bibitem{ShabaniLidar} A. Shabani and D. A. Lidar, Phys. Rev. A \textbf{71} (2005), 020101(R).
%
\bibitem{MaPe06} S. Maniscalco and F. Petruccione, Phys. Rev. A \textbf{73} (2006), 012111.
%
\bibitem{Maniscalco07} S. Maniscalco, Phys. Rev. A \textbf{75} (2007), 062103.
%
\bibitem{Huang-Yi08} X. L. Huang and X. X. Yi, EPL \textbf{82} (2008), 50001.
%
\bibitem{Koss-Reb} A. Kossakowski and R. Rebolledo, Open Syst. Inf. Dyn. \textbf{16} (2009), 259.
%
\bibitem{Mazzola} L. Mazzola, E. -M. Laine, H. -P. Breuer, S. Maniscalco and J. Piilo, Phys. Rev. A \textbf{81} (2010), 062120.
%
\bibitem{HuPazZhang} B. L. Hu, J. P. Paz and Y. Zhang, Phys. Rev. D \textbf{45} (1992), 2843.
%
\bibitem{Grabert97} R. Karrlein and H. Grabert, Phys. Rev. E \textbf{55} (1997), 153.
%
\bibitem{Garraway} B. M. Garraway, Phys. Rev. A \textbf{55} (1997), 2290.
%
\bibitem{BrPe99} H. -P. Breuer, B. Kappler and F. Petruccione, Phys. Rev. A \textbf{59} (1999), 1633.
%
\bibitem{Shibata77} F. Shibata, Y. Takahashi and N. Hashitsume, J. Stat. Phys. \textbf{17} (1977), 171.
%
\bibitem{Chru-Koss10} D. Chru\'sci\'nski and A. Kossakowski, Phys. Rev. Lett. \textbf{104} (2010), 070406.
%
\bibitem{Hall} M. J. W. Hall, J. Phys. A: Math. Theor. \textbf{41} (2008), 269801.
%
\bibitem{Commutative} D. Chru\'sci\'nski, A. Kossakowski, P. Aniello, G. Marmo and F. Ventriglia, A class of commutative dynamics of open quantum systems, arXiv:1004.1388 (2010).
%
\bibitem{Chru-KossPre} D. Chru\'sci\'nski and A. Kossakowski, General form of quantum evolution, arXiv:1006.2764 (2010).
%
\bibitem{SchallerBrandes1} G. Schaller and T. Brandes,  Phys. Rev. A \textbf{78} (2008), 022106.
%
\bibitem{SchallerBrandes2} G. Schaller and T. Brandes,  Phys. Rev. A \textbf{79} (2009), 032110.
%
\bibitem{BenattiDCG1} F. Benatti, R. Floreanini, U. Marzolino, EPL \textbf{88} (2009), 20011.
%
\bibitem{BenattiDCG2} F. Benatti, R. Floreanini, U. Marzolino, Phys. Rev. A \textbf{81} (2010), 012105.
%
\bibitem{FeymannVernon} R. P. Feynman and F. L. Vernon, Ann. Phys. (N. Y.) \textbf{24} (1963), 118.
%
\bibitem{CaldLegg} A. O. Caldeira and A. J. Leggett, Physica \textbf{121A} (1983), 587.
%
\bibitem{Gisin} N. Gisin and I. C. Percival, J. Phys. A: Math. Gen. \textbf{25} (1992), 5677.
%
\bibitem{Diosi} L. Di\'osi, N. Gisin, J. Halliwell and I. C. Percival, Phys. Rev. Lett. \textbf{74} (1995), 203.
%
\bibitem{PlenioKnight} M. B. Plenio and P. L. Knight, Rev. Mod. Phys. \textbf{70} (1998), 104.
%
\bibitem{Ines1} I. de Vega, D. Alonso and P. Gaspard, Phys. Rev. A \textbf{71} (2005), 023812.
%
\bibitem{Lax} M. Lax, Phys. Rev. \textbf{129} (1963), 2342.
%
\bibitem{Ines2} D. Alonso and I. de Vega, Phys. Rev. Lett. \textbf{94} (2005), 200403.
%
\bibitem{Cubbit} M. M. Wolf, J. Eisert, T. S. Cubitt, and J. I. Cirac, Phys. Rev. Lett. \textbf{101} (2008), 150402.
%
\bibitem{BreuerFinlandeses} H. -P. Breuer, E. -M. Laine and J. Piilo, Phys. Rev. Lett. \textbf{103} (2009), 210401.
%
\bibitem{NoMarkovianidad} A. Rivas, S. F. Huelga and M. B. Plenio, Phys. Rev. Lett. \textbf{105} (2010), 050403.
%
\bibitem{CPSun} X. -M. Lu, X. Wang and C. P. Sun, Phys. Rev. A \textbf{82} (2010), 042103.
%
\bibitem{UshaDevi} A. R. Usha Devi, A. K. Rajagopal and Sudha, Phys. Rev. A \textbf{83} (2011), 022109.
%
\bibitem{EAndersson} E. Andersson, J. D. Cresser, M. J. W. Hall, Canonical form of master equations and characterization of non-Markovianity, arXiv:1009.0845 (2010).
%
\bibitem{chinoemail} S. C. Hou, X. X. Yi, S. X. Yu and C. H. Oh, An alternative non-Markovianity measure by divisibility of dynamical map, arXiv:1102.4659 (2011).
%
\bibitem{Darek-Andrzej-Angel} D. Chrus\'cin\'ski, A. Kossakowski and A. Rivas, Phys. Rev. A \textbf{83} (2011), 052128.

\end{thebibliography}
\end{document}